\documentclass[11pt]{article} 

\usepackage{amsmath,amsthm,latexsym,amssymb,amsfonts,epsfig}

\newtheorem*{Theorem}{Theorem}
\newtheorem*{Definition}{Definition}

\newcommand{\be}{\begin{equation}}
\newcommand{\ee}{\end{equation}}
\newcommand{\ba}{\begin{eqnarray}}
\newcommand{\ea}{\end{eqnarray}}

\title{{\sf Canonical Quantum Gravity, Constructive QFT}\\
 {\sf and Renormalisation}} 
\author{
{\sf T. Thiemann}$^1$\thanks{{\sf 
thomas.thiemann@gravity.fau.de}}\\
\\
{\sf $^1$ Inst. for Quantum Gravity, FAU Erlangen -- N\"urnberg,}\\
{\sf Staudtstr. 7, 91058 Erlangen, Germany}\\
}
\date{{\small\sf \today}}

\makeatletter
\@addtoreset{equation}{section}
\makeatother

\begin{document} 

\maketitle

{\sf

\begin{abstract}
The canonical approach to quantum gravity has been put on a firm mathematical 
foundation in the recent decades. Even the quantum dynamics can be rigorously 
defined, however, due to the tremendously non-polynomial character of the 
gravitational interaction, the corresponding Wheeler-DeWitt operator valued 
distribution suffers from quantisation ambiguities that 
need to be fixed.

In a very recent series of works we have employed methods from constructive 
quantum 
field theory in order to address those ambiguities. Constructive QFT trades
quantum fields for random variables and measures thereby phrasing 
the theory in the language of quantum statistical physics. The connection to 
the canonical formulation is made via Osterwalder-Schrader reconstruction.
It is well known in quantum statistics that the corresponding ambiguities 
in measures can be fixed using 
renormalisation. The associated renormalisation flow can thus be used to 
define a canonical renormalisation programme.

The purpose of this article is to review and further develop these ideas and 
to put them into context with closely related earlier and parallel 
programmes.
\end{abstract}

\section{Introduction}
\label{s1}

The canonical approach to quantum gravity has been initialised long time ago
\cite{1}. However, the mathematical foundations of the theory remained veiled 
due to the tremendous non-linearity of the gravitational interaction. This 
has much changed with the reformulation of General Relativity as a Yang-Mills 
type gauge theory in terms of connection rather than metric variables 
\cite{1a} and
has culminated in a research programme now known as Loop Quantum Gravity 
(LQG) (see e.g. \cite{2} for monographs and recent reviews on the subject). 
The qualifier ``loop'' stems from the fact that for gauge theories of 
Yang-Mills type it 
has proved useful to formulate the theory in terms of holonomies of the 
connection along closed paths (loops) in order to maintain manifest 
gauge invariance. Such so-called (Wilson) loop variables are widely used for 
instance in (lattice) QCD \cite{3}.  

LQG has succeeded in providing a rigorous mathematical framework: The 
representation theory of the canonical commutation relations and the $^\ast$ 
relations has been studied and a unique representation has been singled out
\cite{4} that allows for a unitary representation of the spatial 
diffeomorphism group. Moreover, the generators of temporal diffeomorphisms,
sometimes referred to as Wheeler-DWitt operators, could be rigorously 
quantised on the corresponding Hilbert space \cite{5} and in contrast
to the perturbative approach to quantum gravity \cite{6}, no ultraviolet 
divergences were found. It should be emphasised, that this was achieved 
1. in the continuum, rather than on a lattice, i.e. there is no artificial 
cut-off left 
over, 2. for the physical Lorentzian signature rather than unphysical 
Euclidian one, and 3. non-perturbatively and background independently, that is,
one does not perturb around a classical background metric and then quantises 
the fluctuations which thus manifestly preserves the diffeomorphism 
covariance of all constructions. 

However, the theory is not yet completed: Due 
to the tremendously non-polynomial nature of the gravitational interaction,
the ususal factor ordering ambiguity in the quantisation of operator valued 
distributions which are non-linear in the fields is much more severe. Thus,
the operators defined in \cite{5} suffer from those ambiguities. Moreover,
the following problem arises: In the classical theory, the canonical generators
of spacetime diffeomorphisms (that is, their Hamiltonian vector fields) form 
a Lie algebroid (that is, a Lie algebra except that the structure 
constants are replaced by structure functions on the phase space) known as 
the hypersurface algebroid \cite{7}. The structure functions are themselves 
promoted to operator valued distributions upon quantisation, thus it becomes 
even harder to find a quantisation of those generators such that the algebroid 
is represented without anomalies than it would be for an honest Lie algebra.
Specifically, the commutator between two temporal diffeomorphism 
generators is supposed to 1. be proportional to a linear combination of spatial
diffeomorphism generators with operator valued distributions as coefficients
and 2. in an ordering, such that the following holds:
The image of any such commutator of a dense domain of vectors in 
the Hilbert space must be in the kernel of the space of spatially 
diffeomorphism invariant distributions on that domain. In \cite{9}
it was shown that both conditions 1. and 2. hold, however, the coefficients 
in that linear combination do not qualify as quantisations of their classical
counterpart. Thus while the quantisation of the hypersurface algebroid closes,
it does so with the wrong operator valued distributions as coefficients.

Thus, the status of LQG can be summarised as follows:\\
As compared to \cite{1} it is now possible to ask and answer precise 
questions about the mathematical consistency of the whole framework. 
As compared to the perturbative approach \cite{6} the framework does 
not suffer from ultraviolet divergences and one does not have to worry 
about the convergence of a perturbation series due to the manifestly 
non-perturbative definition of LQG. However, just as in the perturbative 
approach, one needs further input in order to draw predictions from the theory, 
although of a different kind:
In the perturbative approach, there are an infinite number of counter terms 
necessary due to non-perturbative non-renormalisability all of which come with 
coefficients that have to be measured
but one can argue that only a finite
number of them is of interest for processes involving energies not exceeding   
a certain threshhold (effective field theory point of view). 
In LQG there are in principle infinitely
many quantisation ordering prescriptions possible, each of which comes 
with definite coefficients in order to yield the correct naive continuum limit
but it is not clear which ordering to choose so that presently one resorts 
to the principle of least technical complexity.

Various proposals have been made in order to improve the situation. 
In \cite{10} one exploits the fact that classically one can always trade 
a set of first class constraints by a single weighted sum of their squares
(called the master constraint). Since a single constraint always closes with 
itself and the weights can be chosen such that the master constraint  
commutes with spatial diffeomorphisms, one can now focus on the 
quantisation ambiguities involved in the master constraint without having to 
worry about anomalies. In \cite{11} the case of General Relativity coupled 
to perfect fluid matter was considered which allows to solve the constraints 
before quantisation so that the remaining quantisation ambiguity now only 
rests in the corresponding physical Hamiltonian that drives the time evolution 
of the physical (that is, spacetime diffeomorpohism invariant)  
observables. In \cite{12} the constraints are quantised on a suitable space
of distributions with respect to a dense domain of the Hilbert space rather 
than the Hilbert space itself in order to find a representation of the
the hypersurface algebroid directly on that space of distributions which would 
at least partially fix the afore mentioned ordering ambiguity. 

It transpires that additional input is necessary in order to fix the 
quantisation ambiguity in the dynamics of LQG and thus 
to complete the definition of the theory.
This would also put additional faith into applications 
of LQG for instance to quantum cosmology \cite{12a} 
(where the amount 
of ambiguity is drastically reduced)
which are believed to be approximations of LQG by enabling to make the 
connection between LQG and those approximations precise including an error
control, see \cite{12b} for recent progress in that respect. In the recent 
proposal \cite{13} 
which we intend to review in this article, the authors were inspired 
by Wilson's observation \cite{13} that renormalisation methods help 
to identify among the principally infinitely many interaction terms in 
Hamiltonians relevant for condensed matter physics the finitely many 
relevant ones that need to be measured. This insight implies that a theory 
maybe perturbatively non-renomalisable but non-perturbatively renormalisable, 
also known as asymptotically safe \cite{14}. The asymptotic safety approach
to quantum gravity for Euclidian \cite{15} and Lorentzian signature 
\cite{15a} precisely rests on that idea and has received 
much attention recently. In fact, there is much in common between our proposal 
and asymptotically safe quantum gravity (especially for Lorentzian 
signature) and we will have the opportunity 
to spell out more precisely points of contact in the course of this article.
Also, there is a large body of work on renormalisation \cite{16} in the 
so-called spin foam 
approach \cite{17} and the related group field theory 
\cite{16a} and tensor model\footnote{In principle any field theory with a 
polynomial Lagrangian can be written as 
a (cloured) tensor model as follows: Pick any orthonormal basis wrt 
the measure appearing in the action, expand the field in that basis, call 
the expansion coefficients 
a coloured (by the spacetime or internal indices) tensor in an infinite 
dimensional $\ell_2$ space and call the integral over polynomials in those basis 
functions that appear in the action upon expanding the 
fields interaction terms of those tensors. 
If the basis carries 
labels in $\mathbb{N}^{nd}$ we obtain a coloured tensor model with tensors 
of rank $n$.} \cite{16b} approach to quantum gravity. 
The spin foam approach is loosely connected to 
LQG in the following sense: The states of the Hilbert space underlying LQG 
are labelled by collections of loops, that is, 3D graphs. A spin foam 
is an operator that maps such states excited on a graph to states excited 
on another graph. The operator depends on a specific class of 4D cell complex
(foam) such  
that its boundary 3D complex is dual to the union of the two graphs 
corresponding to the incoming and outgoing Hilbert spaces.   
The operator is supposed to form the rigging map \cite{18} of LQG, i.e. a 
generalised 
projector onto the joint kernel of the Wheeler-Dewitt constraints. We say 
that the connection is loose because the rigging nature of current 
spin foams in 4D is not confirmed yet. In any case, a spin foam operator 
can be formulated as a state sum model and thus renormalisation ideas 
apply. For applications of renormalisation group ideas in the cosmological
sector of LQG see \cite{18a}.

Most of the work on renormalisation is either within classical 
statistical physics (e.g. \cite{19})
or the Euclidian (also called constructive) approach to quantum 
field theory \cite{20}. In the Euclidian 
approach, the quantum field, which is an operator valued distribution 
on Minkowski space, is replaced by a distribution valued random variable 
on Euclidian space.
While the dynamics in the Minkowski theory is given by Heisenberg equations, 
in the Euclidian theory it is encoded in a measure on the space of random 
variables. We are then back in the realm of statistical physics because 
losely speaking the measure can be considered as a Gibbs factor for a 
Hamiltonian (sometimes called Euclidian action) in 4 spatial dimensions. 
How then should one use renormalisation
ideas for quantum gravity? Quantum gravity is not a quantum field theory 
on Minkowski space (unless one works in the perturbative regime, but then 
it is non renormalisable). Also, while the Minkowski and Euclidian signature 
of metrics 
are related by simple analytic rotation in time from the real to the 
imaginary axis, this does not even work for classical metrics with curvature, 
not to mention the quantum nature of the metric (in ordinary QFT, the metric is
just a non-dynamic background structure). One can of course start with 
Euclidian signature GR and try to build a measure theoretic framework, 
but then the relation to the Lorentzian signature theory is unclear. Moreover,
while as an Ansatz for the Euclidian signature measure we can take the 
exponential of the Euclidian Einstein-Hilbert action, that action is not 
bounded from below and thus the measure cannot be a probability measure
which is one of the assumptions of constructive QFT. Finally, in contrast 
to constructive QFT, in quantum gravity expectation values (operator language) 
or means (measure language) of basic operators (or random variables) such as 
the metric tensor have no direct physical meaning because coordinate 
transformations are considerd as gauge transformations, hence none of the 
basic fields correspond to observables. 

In our approach \cite{13} we will use the framework \cite{11}, that is, we 
do not consider vacuum GR but GR coupled to matter which acts as a dynamical 
reference field. This enables us 1. to solve 
the spatial diffeomorphism and Hamiltonian constraints classically, 2. to 
work directly on the physical Hilbert space (i.e. the generalised
kernel of all constraints equipped with the inner product induced by the 
rigged Hilbert space structure), 3. to have at our disposal immediately
the gauge invariant degrees of freedom such that the physical Hilbert 
space is the representation space of a $^\ast$ representation of 
those observables and 4. to be equipped with a physical Hamiltonian that 
drives the physical time evolution of those observables. Concretely and out 
of mathematical convenience 
we use the perfect fluid matter suggested in \cite{21} but for what follows 
these details are not important. Important is only that it is possible
to rephrase GR coupled to matter as a conservative Hamiltonian system
and that all the machinery that was developed for LQG can be imported. 
Now the quantisation ambiguity rests of course in the physical Hamiltonian 
and it is that object and its renormalisation on which we focus our attention.

As we just explained, we can bring GR coupled to matter somewhat closer 
to the usual setting of ordinary QFT or statistical physics but still we 
cannot apply the usual path integral renormalisation scheme because we work
in the canonical (or Hamiltonian) framework. The idea is then to make use 
of Feynman-Kac-Trotter-Wiener like ideas in order to generate a Wiener measure 
theoretic framework from the Hamiltonian setting and vice versa to use 
Osterwalder-Schrader reconstruction to map the measure theoretic (or path 
integral) frameork to the Hamiltonian one. This way we can map between the 
two frameworks and thus import path integral renormalisation techniques 
into the Hamiltonian framework which are strictly equivalent to those employed
in path integral renormalisation. In order that this works one needs to check 
of course that the Wiener measure constructed obeys at least a minimal subset 
\cite{22} of 
Osterwalder-Schrader axioms \cite{23} in order for the reconstruction to be 
applicable, most importantly reflection positivity.

This was one of the goals of \cite{13}, namely to define a renormalisation 
group flow directly within the Hamiltonian setting with strict equivalence to
the path integral flow. Specifically the flow is a flow of Osterwalder-Schrader
triples $({\cal H}, H, \Omega)$ consisting of a Hilbert space $\cal H$ a 
self-adjoint Hamiltonian $H$ thereon bounded from below and a vacuum vector 
$\Omega\in {\cal H}$ annihilated by $H$. While physically well motivated, 
of course, one does not need to do this. Indeed, renormalisation 
techniques for Hamiltonians and vacua directly within the Hamiltonian setting
were invented before and we devote the next section for putting our framework
into context with schemes closely related to ours. The fact that
we have a precise relation between Hamiltonian and path integral 
renormalisation makes it possible to bring Hamiltonian formulations 
of Quantum Gravity such as LQG and path integral formulations such as 
asymptotically safe Quantum Gravity, into closer contact.\\
\\
The architecture of this article is as follows:\\
\\
In the second section we give an incomplete overview over and 
sketch Hamiltonian renormalisation frameworks 
closely related to ours and point out differences and similarities.

In the third section we review how classical General Relativity coupled to 
suitable matter can be brought into the form of a conservative Hamiltonian
system as well as the LQG quantisation thereof. The necessity to 
remove quantisation ambiguities will be highlighted.

In the fourth section we recall some background material on constructive 
QFT, the Feynman-Kac-Trotter-Wiener construction as well as Osterwalder-Schrader
reconstruction.

In the fifth section we derive the natural relation between families of 
cylindrically defined measures, coarse graining, renormalisation group flows
and their fixed points. We then use Osterwalder-Schrader reconstruction 
to map the flow into the Hamiltonian framework. This section 
contains new material as compared to \cite{13} in the sense that we 
1. develop some systematics in the choice of coarse graining maps that
are motivated by naturally available structures in the classical theory,
2. clarify the importance of the choice of random variable or stochastic 
process when performing OS reconstruction and 3. improve the derivation 
of the Hamiltonian renormalisation flow by adding the uniqueness of the 
vacuum as an additional assumption (also made in the OS framework of Euclidian
QFT \cite{20}) as well as some machinery concerning degenerate contraction 
semigroups and associated Kato-Trotter formulae.

In the sixth section we mention concrete points of contact between 
the scheme developed here and others in the conext of density matrix, 
entanglement and projective renormalisation.  

In the seventh section we sketch a relation between Hamiltonian 
renormalisation via Osterwalder Schrader reconstruction and the functional 
renormalisatiin group which is the underlying technique of the asymptotic 
safety programme.  

In the eighth section we summarise, spell out implications of the 
renormalisation programme for the anomaly free implementation of the 
hypersurface algebroid  
and outline the next steps when trying 
to apply the framework to interacting QFT and finally canonical 
quantum gravity such as LQG.

In appendix A we prove some properties for a coarse graining scheme 
appropriate for non Abelian gauge theories, in appendix B we prove a 
lemma on the existence of certain Abelian $C^\ast-$algebras needed for
the construction of stochastic processes during OS reconstruction, in 
appendix C we collect some renormalisation terminology for readers 
more familiar with actions rather than measures, in appendix D we 
give a proof for the Kato-Trotter product formula for semi-groups and 
projections in the simple case that the semi-group has a bounded generator
and in appendix E we prove a strong limit identity between 
projections needed in section \ref{s5.3}.

\section{Overview over related Hamiltonian renormalisation schemes}
\label{s2}

Purpose of this section is not to give a 
complete scan of the vast 
literature on the subject of Hamiltonian renormalisation
but just to give an overview over those programmes that we believe are closest 
to ours. Also we leave out many finer details as we just want to sketch 
their relation to our framework in broad terms. In sections \ref{s7}, \ref{s8}
we will give a few more details on the connection of our approach with 
the density matrix and functional renormalisation group.\\
\\
The starting point are of course the seminal papers by Kadanoff 
\cite{24} and Wilson \cite{25}.  Kadanoff introduced the concept of a 
block spin transformation in statistical physics, i.e. a coarse graining 
transformation in {\it real space} (namely on the location of the spin degrees
of freedom on the lattice) rather than in some more abstract space (e.g. 
momentum space blocking/suppressing as used e.g. in the asymptotically safe 
quantum 
gravity approach). 
This kind of real space coarse graining map is widely used 
not only in statistical physics but also in the path 
integral approach to QFT as for instance in lattice QCD \cite{26}. 
On the other hand, Wilson introduced the concept of {\it Hamiltonian 
diagonalisation} to solve the Kondo problem (the low temperature 
behaviour of the electrical resistance in metals with impurities). 
This defines a 
renormalisation group flow directly on the space of Hamiltonians and its lowest 
lying energy eigenstates. More precisely, one considers a family of Hamiltonians 
labelled by an integer valued cut-off on the momentum mode label of the 
electron anhilation 
and creation operators. The renormalisation group flow is defined by 
diagonalising the Hamiltonian given by a certain cut-off label, and to 
use the eigenstates so computed to construct the matrix elements of the 
Hamiltonian at the next cut-off label. To make this practical, Wilson 
considered a {\it truncation}, at each renormalisation step, of the full energy 
spectrum to the $10^{3}$ {\it lowest lying energy levels} which was sufficient
for the low temperature Kondo problem. This is in fact nothing but the
concrete application ot the  Rayleigh-Ritz method.
The concept of truncation plays 
an important role also in most other renormalisation schemes, as otherwise 
the calculations become unmanagable. 

The next step was done by Wegner \cite{27} as well as Glazek and Wilson
\cite{28}
which can be considered as a generalisation of the Hamiltonian methods of 
\cite{25}. It could be called {\it perturbative 
Hamiltonian block diagonalisation} and was 
applied in QFT already (e.g. \cite{29} and references therein). Roughly 
speaking, one introduces a momentum cut-off on the modes of the annihilation 
and creation operators involved in the free part of the Hamiltonian, then 
perturbatively (with respect to the coupling constant)
constructs unitarities which at least {\it block diagonalise}
that Hamiltonian with respect to a basis defined by modes that lie below half 
the cut-off and and those that lie between half and the full cut-off, and then 
projects the Hamiltonian onto the Hilbert space defined by the modes below 
half of the cut-off to define a new Hamiltonian at half the cut-off. 
This can be done for each value of the cut-off and thus defines a flow of 
Hamiltonians (and vacua defined as their ground states). Another 
branch of work closely related to this is the projective programme 
due to Kijowski \cite{29a}. Here a flow of Hamiltonians on Hilbert spaces 
for different resolutions is given by the partial traces of the 
corresponding density matrices given by minus their exponential 
(Gibbs factors - assuming that hese are trace class). See also 
\cite{29b,29c,29d} for more recent 
work on renormalisation building on this programme.    

In these developments the spectrum of the Hamiltonian was directly used to 
define the flow. Another proposal was made by White \cite{30} who defined
the {\it density matrix renormalisation group}. This is a real space 
renormalisation group flow which considers the reduced density matrix 
corresponding to the tensor product split of a vector (e.g. the ground state 
of a Hamiltonian) of the total Hilbert space into     
two factors corresponding to a block and the rest (or at least a much larger
``superblock''). This density matrix is diagonalised and then the Hilbert 
space is {\it truncated} by keeping only a certain fixed number of 
{\it highest lying eigenvalues} of the reduced density matrix. Finally 
the Hamiltonian corresponding to the block is projected and then 
the resulting structure is considered as the new structure on the coarser 
lattice resulting from collapsing the blocks to new vertices (we are skipping
here some finer details). This method thus makes use of entanglement ideas 
since the reduced density matrix defines the 
degree of entanglement via its von Neumann entropy.

A variant of this is the {\it tensor renormalisation 
group approach} due to Levin and Nave \cite{31}. It is based on the fact that 
each vector in a finite tensor product of {\it finite dimensional Hilbert 
spaces} can 
be written as a {\it matrix product state}, i.e. the coefficients of the 
vector with respect to the tensor product base can be written as a trace 
of a product of matrices of which there are in general as many as the 
dimensionality of the Hilbert space. One now performs a real space 
renormalisation scheme directly
in terms of those matrices which are considered to be located on a lattice
with as many vertices as tensor product factors. Importantly, this work 
connects renormalisation to the powerful numerical machinery of tensor 
networks \cite{32}.       
 
Finally, as observed by Vidal \cite{33} and Evenbly and Vidal \cite{34}, one 
can improve \cite{30,31} by building in an additional unitary disentanglement
step into the tensor network renormalisation scheme. This is quite natural 
because a tensor network can also be considered as a quantum circuit with
the truncation steps involved considered as isometries but a quantum circuit 
in quantum computing \cite{35} consists of a network of unitary gates some of 
which 
have a disentangling nature depending on the state that they act upon.
The resulting scheme is called multi-scale entanglement renormalisation ansatz
(MERA). 

As this brief and incomplete discussion reveals, there are numerous proposals
in the literature for how to renormalise quantum systems. They crucially 
differ from each other in the choice of the coarse graining map. There are 
various aspects that discriminate between these maps, such as:\\
1. Real space versus other labels\\
The degrees of freedom to be coarse grained are labelled by points in spacetime
or else (momentum, energy, ...).\\
2. Kinematic versus dynamical\\
Real space block spin transformations are an example of a kinematic coarse 
graining, i.e. the form of the action, a Hamiltonian, its vacuum vector,  
its associated reduced density matrix and the corresponding degree of 
entanglement do not play any role. By contrast, Hamiltonian block 
diagonalisation, density matrix and entanglement renormalisation take such 
dynamical information into account.\\
3. Truncated versus exact\\
In principle any renormalisation scheme can be performed exactly, e.g. 
in real space path integral renormalisation one can just integrate the 
excess degrees of freedom that live on the finer lattice but not on the 
coarser thus obtaining the measure (or effective action) on the coarser 
lattice from that of the finer one. The same is true e.g. for the procedure 
followed in asymptotically safe quantum gravity.
However, in practice this may quickly
become unmanagable and thus one resorts to approximation methods e.g. 
by truncation in the space of coupling constants, energy eigenstates or 
reduced density matrix eigenstates.\\
\\
For the newcomer to the subject, this plethora of suggestions may appear 
confusing. Which choice of coarse graining is preferred? Do different choices 
lead to equivalent physics? What can be said about the convergence of various 
schemes and what is the meaning of the fixed point(s) if it (they) exist(s)? 
The physical intuition is that different schemes should give equivalent 
results if 1. the corresponding fixed point conditions capture necessary 
and sufficient properties that the theory should have in order to 
qualify as a continuum theory and 2. {\it when 
performed exactly}. The first condition is obvious, we start from what we 
believe to be an initial guess for how the theory looks at different 
resolutions and then formulate a coarse graining 
flow whose fixed points are such that 
they qualify to define a continuum theory. The seond condition entails that
the coarse graining maps just differ in the separation of the total set 
of degrees of 
freedom into subsets corresponding to coarse and fine resolution, hence 
corresponds to choices of 
coordinate systems which of course can be translated into each other. However,
when truncations come into play, this equivalence is lost because different 
schemes truncate different sets of degrees of freedom which are generically 
no longer in bijection. It is conceivable therefore that dynamically driven 
truncation schemes perform better at identifying the correct fixed point 
structure of the theory in the sense that they may converge faster and are 
less vulnerable to truncation errors or automatically 
pick the truncation of  {\it irrelevant} couplings. This seems 
to be confirmed in spin system examples but we are not aware of a general 
proof. Recently, the importance of the kinematic versus dynamic issue 
has also been emphasised for the LQG and spin foam approach \cite{36}.

In our work we currently are not concerned with issues of computationability,
that is, we consider an exact scheme. Next, as far as the coarse graining 
map is concerned, we currently favour a kinematic scheme. The reason for doing 
this is that kinematic schemes are {\it naturally suggested by measure 
theoretic questions}. Namely, measures on spaces of infinitely many degrees of
freedom are never of the type of the exponential of some action times a 
normalisation constant times Lebesgue measure. Neither of these three 
ingredients is well defined. What is well defined are integrals  
of certain probe functions of the field with respect to that measure. 
These probe functions in 
turn are naturally chosen to depend on test functions that one integrates
the field against. Thus these test functions provide a natural notion 
of resolution, discretisation and coarse graining. By integrating the 
measure against 
probe functions one obtains a family of measures labelled by the test 
functions involved. The relation between test functions at different resolution
induces a corresponding relation between members of the family of measures
which {\it must hold exactly} for a true measure of the continuum QFT. 
In turn, such consistency relations called {\it cylindrical consistency} 
can be used to {\it define} a measure on a space of infinitely many 
degrees of freedom \cite{37}, called a {\it projective limit}. 
The idea is then to formulate
measure renormalisation in such a way that its fixed points solve the 
consistency relations. This approach has been advocated in \cite{38} for 
Euclidian Yang-Mills theory and in \cite{39} for spin foams. Note that 
spin foams strictly speaking do not construct measures but rather are supposed 
to construct a rigging map so that Hamiltonian methods come also into 
play. Indeed, in \cite{36} it was shown that 
the cylindrically consistent coarse graining of the rigging map and its 
underlying spacetime lattice, 
thought of as an 
anti-linear functional on the kinematical Hilbert space,
induces a coarse 
graining of the spatial lattice on its boundary and
thus the Hilbert space thereon, equipping it with a system of 
consistent embeddings, a structure similar to
inductive limits of Hilbert spaces (an inductive structure requires
in addition the 
injections to be isometric). That latter 
structure underlies the kinematical Hilbert space
of LQG and a renormalisation procedure based on inductive limits was already
proposed in \cite{40} due to the similarity of LQG to
lattice gauge theory. \\
\\
Another reason for why picking real space coarse 
graining schemes as compared to say, momentum space based ones is their 
background independence which especially important for quantum gravity. 
In our work, as we consider the version of LQG in which 
the constraints already have been solved, we will work with probability 
measures. As we will see, the connection between inductive limits of 
Hilbert spaces and projective limits of path integral measures can be made 
crystal clear in this case. The price we pay by using an exact, kinematical 
scheme is that the fixed point (or renormalised) 
Hamiltonian becomes spatially non-local at finite resolution.
However, in the free QFT examples studied \cite{13}, which are spatially 
local in the continuum, by blocking the known
fixed point theory from the 
continuum one can see that this is natural and {\it must happen} 
for such schemes,
hence it is not a reason for concern but in fact physical reality. The 
degree of spatial non-locality in fact decreases as we increase the resolution
scale. 

When applying the framework to interacting QFT one will have to resort to 
some kind of approximation scheme and possibly tools from entanglement 
renormalisation combined with tensor network techniques may prove useful.
However, note that QFT of bosonic fields (gravity is an example)
deals with infinite dimensional 
Hilbert spaces even when the theory depends only on a finite number of degrees 
of freedom say by discretising it on a lattice and confining 
it to finite volume. Thus, to apply tensor network techniques which, to the 
best
of our knowledge, require the factors in the tensor product to be finite 
dimensional Hilbert spaces, one would have to 
cut off the dimensions of those Hilbert spaces right from the beginning, i.e.
one would have to work with {\it three} cut-offs rather than two
(see e.g. \cite{41} where 
quantum group representations are used in gauge theories rather than classical 
group representations and performs real space renormalisation 
or \cite{41a} where one combines both the UV and the 
dimension cut-off into one by turning the dimension of tensor spaces 
in tensor models into a finite coarse graining 
parameter and otherwise performs the 
asymptotic safety programme which is often formulated in the presence 
of a cut-off anyway). 

Some sort of truncation or approximation has 
to be made in practice when treating complex systems numerically.  
The physical insight behind the tensor network
and density matrix/entanglement renormalisation  
developments, namely that the dynamically interesting vectors in a Hilbert 
space appear to lie in a ``tiny'' subspace thereof is presumably a profound one 
and the truncation of the Hilbert space to the corresponding subspaces appears
to be well
motivated by the model (spin) systems studied so far. Still, what one would 
like to have is some sort of error control 
or convergence criteria on those truncations. We appreciate that this is a 
hard task for the future. For the time being we phrase 
our framework without incorporating a cut-off on the  
dimension of Hilbert spaces as we are not yet concerned with numerical 
investigations, however, we may have to use some of these ideas 
in the future.

\section{Canonical Quantum Gravity coupled to reference matter}
\label{s3}

The physical idea is quite simple and goes back to \cite{42}: General 
Relativity is a gauge theory, the 
gauge group being the spacetime diffeomorphism group. Thus the basic tensor 
and spinor fields in terms of which one writes the Einstein-Hilbert action and 
the action of the standard model coupled to the metric (or its tetrad) are 
not observable. However, the value of, say a scalar field $\Phi$ at that  
spacetime point $X_y$, at which four reference scalar fields 
$\phi^0,..., \phi^3$ take values $y^0,..,y^3$, that is, 
$\Phi(X_y);\;\;\phi^\mu(X_y)=y^\mu$ is spacetime diffeomorphism invariant.
For this to work, the relation $\phi^\mu(X)=y^\mu$ must of course 
be invertible, in particular the reference scalar fields must not vanish 
anywhere or anytime. This seems to be a property of dark matter \cite{49}.

These kind of relational observables have been further 
developed by various authors, see in particular \cite{43,44}. When one 
couples General Relativity and such reference matter preserving general 
covariance, it becomes possible to formulate the theory in a manifestly 
gauge invariant way. The form of that gauge invariant formulation of course 
strongly depends on the type of reference matter used and its Lagrangian.
In what follows we use the concrete model \cite{11} out of mathematical 
convenience, but we emphasise that the same technique works in a fairly
general context. In the next subsection, that model will be introduced and 
the classical gauge invariant formulation will be derived. After that we 
quantise it using LQG methods which will be introduced in tandem. 

\subsection{Gaussian dust model}
\label{s3.1}     
   
The Lagrangian of the theory takes the form
\be \label{3.1}
L=L_{EH}+L_{SM}+L_D
\ee
where $L_{EH}$ is the Einstein-Hilbert Lagrangian, $L_{SM}$ the standard model 
Lagrangian coupled to GR via the metric, its tetrad or its spin connection
and $L_D$ is the Gaussian dust Lagrangian \cite{21}
\be \label{3.2}
L_D=-\frac{1}{2}\sqrt{|\det(g)|}\;\{ g^{\mu\nu} \;[\rho\; 
(\nabla_\mu T)\;(\nabla_\nu T)+2(\nabla_\mu T)(W_j \nabla_\nu S^j)]+\rho\}
\ee
where $g$ is the Lorentzian signature metric tensor, $\nabla_\mu$ its 
Levi-Civita covariant differential, 
$\phi^0:=T,\; \phi^j:=S^j;\; j=1,2,3$ are the reference scalar fields 
introduced above and $\rho, W_j$ are additional four scalar fields. The latter
four fields appear without derivatives and thus give rise to primary 
constraints in addition to those present even in vacuum GR. One easily show 
that the contribution of $L_D$ to the energy momentum tensor is of perfect 
fluid type. Further physical properties and motivations are discussed in 
\cite{21}. For what follows it suffices to know that the equations of motion
for $T,S^j$ say that $\nabla_\mu T$ is a timelike 
geodesic co-tangent and that $S^j$ 
is constant along the geodesic spray. Thus, those geodesics can be 
interpreted as worldlines of dynamically coupled test observers.

The full constraint analysis of (\ref{3.1}) is carried out in \cite{11}. There 
are secondary constraints and the full set of constraints contains those of 
first and second class (see \cite{44} for a modern treatment of Dirac's 
algorithm \cite{45}). One has to introduce a Dirac bracket and solves the second
class constraints in the course of which the variables $\rho, W_j$ are 
eliminated. The remaining constraints are then of first class and read
\be \label{3.3}
C^{{\rm tot}}=C+\frac{P-q^{ab} T_{,a} C_b}{\sqrt{1+q^{ab} T_{,a} T_{,b}}},
\;\;
C_a^{\rm tot}=C_a+P T_{,a}+ P_j S^j_{,a}
\ee
Here $C$ is the Wheeler-DeWitt constraint function (including standard matter), 
$C_a,\; a=1,2,3$ are the spatial diffeomorphism functions (including standard
matter). The Dirac bracket
reduces to the Poisson bracket on all the variables involved in (\ref{3.3})
and $P,P_j$ are the momenta conjugate to $T, S^j$, e.g.
$\{P(x),T(y)\}=\delta(x,y)$. 
Here $a=1,2,3$ are tensorial indices on the spatial hypersurface $\sigma$ 
of the 
Arnowitt-Deser-Misner foliation underlying the Hamiltonian formulation of GR
\cite{46} with intrinsic metric tensor $q_{ab}$. 
For the moment it is just necessary to know that $C,C_a$ do not 
involve the variables $T, P, S^j, P_j$.

The constraints (\ref{3.3}) encode the spacetime diffeomorphism gauge symmetry
in Hamiltonian form, in particular they represent the hypersurface deformation 
algebra \cite{7}. It is possible to solve these remaining constraints, to
determine the complete set of gauge invariant (so called Dirac) observables 
and to determine the physical Hamiltonian $H$ that drives their physical 
time evolution \cite{11}. Equivalently, we may gauge fix (\ref{3.3}). 
The above interpretation of $T, S^j$ suggest to use the gauge conditions 
$G=T-t, G^a=\delta^a_j S^j-x^a$. The stabilisation of these gauge conditions 
fixes the Lagrange multipliers $\lambda, \lambda^a$ in the gauge generator 
\be \label{3.4}   
K:=C^{{\rm tot}}(\lambda, \vec{\lambda}):=\int_\sigma\; d^3x\;
[\lambda C^{{\rm tot}}+\lambda^a C_a^{{\rm tot}}]
\ee
namely
\ba \label{3.5}
\dot{G}(t,x)&=&\{K,G(x)\}+\partial_t G(t,x)=
\frac{\lambda(x)}{\sqrt{1+q^{ab} T_{,a} T_{,b}}}+\lambda^a T_{,a}-1=0,\;\;
\nonumber\\
\dot{G}^a(t,x) &=&\{K,G^a(x)\}+\partial_t G^a(t,x)=\lambda^b S^j_{,b} 
\delta^a_j=0
\ea
which when evaluated at $G=G^a=0$ yields the unique solution $\lambda=1, 
\lambda^a=0$. Likewise, in this gauge the constraints can be uniqely solved 
for $P=-C,\; P_j=-\delta_j^a C_a$ while $T, S^j$ are pure gauge. This shows 
that the physical degrees of freedom are those not involving $T,P,S^j, P_j$.

For any function $F$ independent of these variables the reduced or physical 
Hamiltonian is that function on the phase space coordinatised by the physical
degrees of freedom which generates the same time evolution as $K$ when 
the constraints, gauge conditions and stabilising Lagrange multipliers 
are installed
\be \label{3.6}
\{H,F\}:=
\{K,F\}_{C^{{\rm tot}}=\vec{C}^{{\rm tot}}=G=\vec{G}=\lambda-1=\vec{\lambda}=0}
=\{\int_\sigma d^3x C, F\}
\ee
which shows that
\be \label{3.7}
H=\int_\sigma d^3x C
\ee
Thus the final picture is remarkably simple: The physical phase space is 
simply coordinatised by all metric and standard matter degrees of freedom 
(and their conjugate momenta) while the physical Hamiltonian is just the 
integral of the usual Wheeler-DeWitt constraint. The influence of the 
reference matter now only reveals itself in the fact that $H$ is not  
constrained to vanish as it only involves the geometry and standard matter
contribution $C$ of $C^{{\rm tot}}$ 
and that the number of physical degrees of freedom 
has increased by four as compared to the system without reference matter.
This phenomenon is of course well known from the electroweak interaction: 
One can solve the three isospin SU(2) Gaus constraints for three of the four 
degrees of freedom sitting in the complex valued 
Higgs isodublett, leaving a single scalar 
Higgs field and three massive rather than massless vector bosons. See 
\cite{47} for further discussion.\\   
\\
We close this subsection with three remarks:\\ 
First,
a complete discussion requires to 
show that the gauge cut $G=G^a=0$ on the constraint surface of the phase space
be reachable from anywhere on the constraint surface. As (\ref{3.5}) shows,
this requires that $S^j_{,a}$ be invertible. We thus impose this as an 
anholonomic constraint on the total phase space. One easily verifies from
(\ref{3.5}) that this condition is gauge invariant, i.e. compatible with 
the dynamics.\\
Second, the simplicity of the final picture is due to the 
particular choice of reference matter. Other reference matter most
likely will increase the complexity, see e.g. \cite{48} which produces a
square root Hamiltonian! One may argue that the dust is a form of cold dark 
matter \cite{49}, but it is unclear whether this is physically viable.
Nevertheless, the present model serves as a proof
of principle, namely that GR coupled to standard matter and reference 
can be cast into 
the form of a conservative Hamiltonian system.\\
Third, it should be appreciated that the reference matter helps us to 
accomplish a huge step in the quantum gravity programme: It frees us from
quantising and solving the constraints, construct the physical inner 
product, the gauge invariant observables and their physical time evolution.
All of these steps are of tremendous technical difficulty \cite{2}. 
All we are left to do is to quantise the physical degrees of freedom 
and the physical Hamiltonian.

\subsection{Loop Quantum Gravity quantisation of the reduced physical system}
\label{s3.2}

In order to keep the technical complexity to a minimum, we consider just the 
contribution to $H$ coming from the gravitational degrees of freedom, see 
\cite{5,2} for more details on standard matter coupling. The Hamiltonian 
directly written in terms of $SU(2)$ gauge theory variables reads (we drop 
some numerical coefficients that are not important for our discussion)
\ba \label{3.8}
H &=& H_E+H_L
\nonumber\\
H_E &=& \int_\sigma\; {\rm Tr} (F\wedge \{V,A\})
\nonumber\\
V &=& \int_\sigma \; \sqrt{|\det(E)|}
\nonumber\\
H_L &=& \int_\sigma\; {\rm Tr} (\{\{H_E,V\},A\}\wedge \{\{H_E,V\},A\}
\wedge \{V,A\})
\ea
Here $A$ is an SU(2) connection and $E$ an SU(2) non-Abelian electric field 
that one would encounter also in a SU(2) Yang-Mills theory. However, the 
geometric interpretation of $A,E$ is different: Namely, 
$e^a_j:=E^a_j/\sqrt{|\det(E)|}$ is a triad, that is, 
$q^{ab}=\delta^{jk} e^a_j e^b_k$ is the inverse spatial metric. 
Here, as before, $a,b,c,..=1,2,3$ denote spatial tensor indices while 
now $j,k,l,..=1,2,3$ denote su(2) Lie algebra indices.
Further, let 
$\Gamma_a^j$ be the spin connection of $e^a_j$. Then 
$K_{ab}:=(A_a^j-\Gamma_a^j)e^k_b \delta_{jk}$ has the meaning of the extrinsic 
curvature of the ADM slices \cite{46} on the kernel of the SU(2) Gauss 
constraint
\be \label{3.9}
C_j:=\partial_a E^a_j+\epsilon_{jkl} A_a^k E^a_m \delta^m_l
\ee
The important quantity $V$ is recognised as the total volume of the 
hypersurface $\sigma$ and $H_E, H_L$ are known as the Euclidian and Lorentzian 
contributions to $H$. See \cite{5} for further details. The Poisson brackets 
displayed are with respect to the standard symplectic structure
\be \label{3.10}
\{A_a^j(x), A_b^k(y)\}=     
\{E^a_j(x), E^b_k(y)\}=     
\{E^a_j(x), A_b^k(y)\}-\kappa \delta^a_b \delta_j^k \delta(x,y)=0
\ee
where $\hbar\kappa=\ell_P^2$ is the Planck area. The definition of the phase 
space is completed by the statement that the elementary fields $A, E$ are 
real valued
\be \label{3.11}
[A_a^j(x)]^\ast-A_a^j(x)=       
[E^a_j(x)]^\ast-E^a_j(x)=0
\ee
The traces involved in \ref{3.8} are carried out by introducing the Lie algebra 
valued 1-forms $A=A_a^j\; \tau_j\; dx^a$ where $2i\tau_j$ are the Pauli 
matrices and $F=2(dA+A \wedge A)$ is the curvature of $A$. The non polynomiality
of GR is hidden in the Poisson brackets that appear in (\ref{3.8}). The reason 
for why we use these particular Poisson bracket structure
will become clear only later.

To quantise the theory we start from functions on the phase space that are 
usually employed in lattice gauge theory, see e.g. \cite{50}, namely 
non-Abelian magnetic 
holonomy and electric flux variables
\be \label{3.12}
A(c):={\cal P}\exp(\int_c A),\;\;
E_f(S)=\int_S {\rm Tr}(f \ast E)
\ee
where $\cal P$ denotes path ordering, $c$ is a piecewise analytic real curve,
$S$ a piecewise real analytic surface, $f$ an su(2) valued function and 
$\ast E=\epsilon_{abc} E^a dx^b\wedge dx^c/2$ the pseudo 2-form corresponding 
to the su(2) valued vector density $E$. Note that $A(c)$ is SU(2) valued while 
$E_f(S)$ is su(2) valued 
\be \label{3.13}
A(c)^\ast =(A(c)^{-1})^T=A(c^{-1})^T,\;\;E_f(S)^\ast=-E_f(S)^T
\ee
where $c^{-1}$ is the same curve as $c$ but with the opposite orientation.
The simplest non-trivial Poisson brackets are 
\be \label{3.14}
\{E_f(S),A(c)\}=\kappa\; A(c_1)\; f(c\cap S)\; A(c_2)
\ee
in case that $S\cap c$ is a single point in the interior of both $c,S$, see 
\cite{4} for a complete discussion. The relations (\ref{3.13}), (\ref{3.14}) 
are the defining relations of a non-commutative abstract
$^\ast-$algebra $\mathfrak{A}$ generated by fluxes and  
complex valued smooth functions $F$ of a finite number of holonomy variables 
\cite{4}.
It is the free algebra generated by them and divided by the two sided ideal 
generated by the canonical commutation relations 
$E_f(S) A(c)-A(c) E_f(S)=i \hbar \{E_f(S), F\}$ and the adjointness
relations (\ref{3.13}). See \cite{4} for more details.   

Interestingly, the physical Hamiltonian $H$ has a large symmetry group. Namely,
it is invariant under the group 
$\mathfrak{G}=SU(2)_{{\rm loc}} \rtimes Diff(\sigma)$ where 
$SU(2)_{{\rm loc}}$ denotes the group of local SU(2) valued gauge 
transformations and $Diff(\sigma)$ denotes the group of (piecewise real 
analytic) diffeomorphisms of $\sigma$. An element of $\mathfrak{G}$ is given 
by a pair $\mathfrak{g}=(g,\varphi)$ which acts on the basic variables as 
\be \label{3.15}
\alpha_{(g,\varphi)}(A)=-dg g^{-1}+g [\varphi^\ast A] g^{-1},\; 
\alpha_{(g,\varphi)}(\ast E)=g[\varphi^\ast (\ast E)] g^{-1}
\ee
where $\varphi^\ast$ denotes the pull-back action of diffeomorphisms 
on differential forms. This action lifts to the algebra $\mathfrak{A}$,
specifically
\be \label{3.15a}
\alpha_{(g,\varphi)}(A(c))=g(b(c))\; A(\varphi(c))\; g(f(c))^{-1},\;\;
\alpha_{(g,\varphi)}(E_f(S))=
E_{[g^{-1} f g]\circ\varphi^{-1}}(\varphi(S))
\ee
where $b(c),f(c)$ denote beginning and final point of $c$
and this simple covariant transformation behaviour 
was part of the reason why the particular ``smearing'' of $A$ along 
curves involved in holonomies is used.
Note also the different character of the two groups: While we still have to find 
the gauge invariant observables with respect to the Gauss constraint 
the diffeomorphism constraint is already solved. The diffeomorphisms in 
$\mathfrak{G}$ are thus to be considered as active diffeomorphisms rather than 
passive ones.

The mathematical problem in quantising the theory consists in constructing a 
$^\ast$ representation of $\mathfrak{A}$, that is, a representation 
$(\pi,{\cal H})$ of elements $a\in \mathfrak{A}$ as operators $\pi(a)$ 
densely defined on a common, invariant domain $\cal D$
of a Hilbert space $\cal H$ such that the $^\ast$ relations are implemented as 
adjointness relations and such that the canonical commutation relations 
are implemented as commutators between them. Thus we want in particular
that 
\be \label{3.16}
\pi(a^\ast)=[\pi(a)]^\dagger,\; \pi(a+b)=\pi(a)+\pi(b),\; 
\pi(ab)=\pi(a)\pi(b),\;\pi(z a)=z \pi(a),\;[\pi(a),\pi(b)]=\pi(c);\;
\ee
for all $a,b,c\in \mathfrak{A},\;z\in \mathbb{Z}$ if $ab-ba=c$. 
In QFT this problem is known to typically have 
an uncountably infinite number of unitarily inequivalent solutions, there
is no Stone-von Neumann uniqueness theorem when the number of degrees of 
freedom in infinite. Hence, to make progress, we must use additional 
physical input. That input can only come from the Hamiltonian. Thus, 
we require in addition that the representation supports $H$ as a self-adjoint 
operator ($H$ is real valued) also densely defined on $\cal D$ and 
such that $\cal H$ carries a unitary representation $U$ of $\mathfrak{G}$ (such 
that its generators are self-adjoint by Stone's theorem). Using 
the powerful machinery of the Gel'fand-Naimark-Segal construction 
\cite{51}, the 
representation property and the unitarity property can be granted if 
we find a positive linear and $\mathfrak{G}$ invariant functional 
$\omega:\; \mathfrak{A}\to \mathbb{C}$
on 
$\mathfrak{A}$, that is
\be \label{3.17}
\omega\circ\alpha_{\mathfrak{g}}=\omega,\;\; \omega(a^\ast a)\ge 0
\ee

In \cite{4} it was found that there is a {\it unique} $\omega$ satisfying 
(\ref{3.17}). While the derivation is somewhat involved, the final result
can be described in a compact form. The dense domain $\cal D$ consists of
functions of the form 
\be \label{3.18}
\psi(A)=\psi_\gamma(\{A(c)\}_{c\in E(\gamma)});\;\;
\psi_\gamma \in C^\infty(SU(2)^{|E(\gamma)|},\mathbb{C})
\ee
i.e. $\psi_\gamma$ is 
a complex valued, smooth functions of a finite number of holonomy variables.
The union of the curves of these holonomies forms a finite graph $\gamma$  
where $E(\gamma)$ denotes the set of its edges. Note that the elements 
of $\mathfrak{A}$ that just depend on the connection are themselves of the 
form (\ref{3.18}) and thus their action by multiplication
\be \label{3.19}
[\pi(f)\psi](A):=f(A) \; \psi(A)
\ee
is densely defined. The fluxes are densely defined when acting by derivation
\be \label{3.20}
[\pi(E_f(S)) \psi](A):=i\hbar \{E_f(S),\psi(A)\}
\ee
which also solves the canonical commutation relations. 

To see that the 
adjointness conditions hold we need the inner product. To define it, we note 
that graphs defined by finiteley many piecewiese analytic curves 
are partially ordered by set theoretic inclusion and they are 
directed in the sense that for any two graphs $\gamma_1,\gamma_2$ there 
exists $\gamma_3$ with $\gamma_1,\gamma_2\subset \gamma_3$, for instance 
$\gamma_3=\gamma_1\cup \gamma_2$. Then we can decompose all edges 
of $\gamma_1,\gamma_2$ with respect to the edges of $\gamma_3$, use the 
algebraic relations of the holonomy $A(c^{-1})=A(c)^{-1},\; 
A(c\circ c')=A(c)A(c')$ where $c\circ c'$ is the composition of curves 
$f(c)=b(c')$ in order to write $\psi_1, \psi_2$ excited over 
$\gamma_1,\gamma_2$ respectively as functions excited over $\gamma_3$.
Thus it is sufficient to know the inner product of functions excited 
over the same graph $\gamma$ which is given by
\be \label{3.20}
<\psi,\psi'>_{{\cal H}}:=\int_{SU(2)^{|E(\gamma)|}}\; 
\prod_{k=1}^{|E(\gamma)|} \; d\mu_H(h_k)\;
\overline{\psi(\{h_k\})}\;\psi'(\{h_k\})
\ee
where $\mu_H$ is the Haar measure on SU(2).
One can check that the adjointness relations are indeed satisfied, in fact 
$\pi(E_f(S))$ is an unbounded but essentially self-adjoint operator (i.e.
a symmetric operator with unique self-adjoint extension). 

In fact (\ref{3.20}) defines a cylindrical family of measures $\mu_\gamma$,
one for every graph $\gamma$. One has to check that (\ref{3.20}) is well 
defined because a function excited on $\gamma$ can be written also as a 
function excited over any finer graph $\gamma'$ by extending it trivially to 
the additional edges. This is in fact the case \cite{4}. Then the 
Kolmogorov type extension theorems grant that the family extends to an 
honest continuum measure $\mu$ on the quantum configuration space 
$\overline{{\cal A}}$ of 
distributional connections. We will not go into the details here which can be 
found in \cite{4} but just mention for the interested reader that this 
space coincides with the so-called Gel'fand spectrum of the Abelian C$^\ast$ 
algebra that one obtains by completing the space of functions (\ref{3.18})
in the sup norm. It follows that the Hilbert space is given by
${\cal H}=L_2(\overline{{\cal A}},d\mu)$.  

By construction, the Hilbert space $\cal H$ carries a unitary representation 
$U$ of 
$\mathfrak{G}$ given by 
\be \label{3.21}
(U(\mathfrak{g})\psi)(A)=
\psi_\gamma(\{\alpha_{\mathfrak{g}}(A(c))\}_{c\in E(\gamma)})
\ee
To check this, one uses the properties of the Haar measure 
(translation invariance) and the diffeomorphism invariance of (\ref{3.20})
which does not care about the location and shape of the curves involved.

The Hilbert space comes equipped with an explicitly known orthonormal basis 
called spin network functions (SNWF). This makes use of harmonic analysis on 
compact groups $G$ \cite{52}, in particular the Peter \& Weyl theorem which 
states that the matrix element functions of the irreducible representations 
of $G$, which are all finite dimensional and unitary without loss of generality,
are mutually orthogonal, unless equivalent, with respect to the inner product
defined by the Haar measure on $G$, moreover, they span the whole Hilbert space.
As the irreducible representations of SU(2) are labelled by spin quantum 
numbers, the name SNWF comes at no surprise. More in detail a SNWF 
$T_{\gamma,j,\iota}$ is labelled by a graph $\gamma$, a tuple 
$j=\{j_c\}_{c\in E(\gamma)}$ of spin quantum numbers decorating the edges and 
a tuple $\iota=\{\iota_v\}_{v\in V(\gamma)}$ of intertwiners decorating the 
vertices $v$ in the vertex set $V(\gamma)$ of $\gamma$. Here an intertwiner
$\iota_v$ projects the tensor product of irreducible representations 
corresponding to the edges incident at $v$ onto one of 
the irreducible representations appearing in its decomposition into 
irreducibles (Clebsch-Gordan theory). 
Besides providing an ONB convenient for concrete calculations,
SNF make it easy to solve the Gauss constraint: A detailed analysis 
\cite{2} shows that
(\ref{3.9}) can be quantised in the given representation and just imposes 
that the space of intertwiners be restricted to those projecting on the 
trivial (spin zero) representation. We call such intertwiners 
gauge invariant. Hence the joint kernel of the Gauss 
constraints is a closed subspace of ${\cal H}$ which is explicitly known.
We will abuse the notation and will not distinguish between that subspace and 
${\cal H}$ and henceforth consider the Gauss constraint as solved. All operators
considered in what follows are manifestly gauge invariant and preserve that 
subspace. 

As a historical remark, solutions of the Gauss constraint are excited 
on closed graphs since there is no non-trivial intertwiner between the 
trivial representation and a single irreducible one, hence open 
ends are forbidden. For closed graphs,
one can alternatively label SNWF by homotopically independent closed paths 
(loops) with a common starting point (vertex) on that graph. Originally
one used loops as labels, hence the name Loop Quantum Gravity (LQG).

One of the many unfamiliar features of $\cal H$ is that it is not separable
which easily follows from the uncountable cardinality of the set of graphs.
This is a direct consequence of the diffeomorphism invariance of the inner 
product: Two graphs that are arbitrarily close but disjoint are simultaneously
also arbitrarily far apart under the inner product. Thus if the measure 
clusters for far apart support of the smearing functions (here the graphs) 
then the orthogonality of the corresponding spn network functions comes at 
no surprise. A direct consequence of this is that the diffeomorphism 
operators $U(\varphi)$ do not act (strongly) continuously, hence a 
generator of of infinitesimal diffeomorphisms generated by the integral 
curves of vector fields cannot exist. Yet another direct consequence is that 
the connection operator $A$ itself does not exist, only its holonomies do.\\  
\\ 
The remaining task is to quantise the Hamiltonian and it is at this 
point where the afore mentioned quantisation ambiguities arise.        
The strategy followed in \cite{5} is as follows: It turns out that the 
volume operator appearing in (\ref{3.8}) can be quantised on ${\cal H}$ as an 
essentially self-adjoint operator whose spectrum is pure point (discrete)
\cite{53}. 
It is densely defined on the span of the SNWF and it acts vertex wise, 
with no contribution from gauge (in)variant vertices that are not at least
three (four) valent or from vertices whose incident edges have tangents 
in a common two dimensional or one dimensional space. Next the 
holonomy along an open curve $c$ can be expanded as $A(c)=1_2+\int_c\; A+...$ 
and along a closed curve $\alpha$ as 
$A(\alpha)=1_2+\int_{S,\partial S=\alpha} \; F$ so that the functions 
$A, F$ that appear in (\ref{3.8}) can be approximated by suitable holonomies
where the approximation is in terms of the ``length'' of the curves involved
which are matched with the coordinate volume assigned by the Lebesgue measure 
$d^3x$ appearing in (\ref{3.8}) approximating the integral by a Riemann
sum (this is a regularisation step). Suppose then that somehow a well 
defined operator $H_E$ can be defined by replacing the classical functions by 
operators and the Poisson brackets by commutators times $\hbar$. Then 
the same argument can be applied to the Lorentzian piece. As a final piece of 
information, one uses the observation that a spatially diffeomorphism 
invariant operator, densely defined on the span of SNWF
cannot have non-trivial matrix elements between 
SNWF excited over different graphs \cite{8}. This has the following
consequence: Let ${\cal H}_\gamma$ be the 
closed linear 
span of SNWF excited precisely over $\gamma$. Then, if $H$ is supposed 
to preserve its classical diffeomorphism invariance upon quantisation
we necessarily 
have 
\be \label{3.22}
H=\oplus_\gamma H_\gamma,\;  
{\cal H}=\oplus_\gamma {\cal H}_\gamma
\ee
where each $H_\gamma$ is self-adjoint on ${\cal H}_\gamma$, in 
particular it preserves this space. Let now 
$P_\gamma:\; {\cal H}\to {\cal H}_\gamma$ be the orthogonal projection. 
Then the following concrete expression for $H$ can be given \cite{11}
(again we drop 
some numerical coefficients and set $\hbar=1$)
\ba \label{3.23}
H_{E,\gamma} &=& P_\gamma\; H'_{E,\gamma}\; P_\gamma
\\
H'_{E,\gamma} &=& i \sum_{v\in V(\gamma)}\;\;\;\;
\sum_{c_1,c_2,c_3\in E(\gamma);\;c_1\cap c_2\cap c_3=v}
\nonumber\\
&& \epsilon^{IJK} {\rm Tr}(
[A(\alpha_{\gamma,v,c_I,c_J})-A(\alpha_{\gamma,v,c_I,c_J})^{-1}]\;
A(c_K)[V,A(c_K)^{-1}])
+ h.c.
\nonumber\\
H_{L,\gamma} &=& P_\gamma\; H'_{L,\gamma}\; P_\gamma
\nonumber\\
H'_{L,\gamma} &=& i \sum_{v\in V(\gamma)}\;\;\;\;
\sum_{c_1,c_2,c_3\in E(\gamma);\;c_1\cap c_2\cap c_3=v}
\nonumber\\
&& \epsilon^{IJK} {\rm Tr}(
A(c_I)[[H'_{E,\gamma},V],A(c_I)^{-1}]\;
A(c_J)[[H'_{E,\gamma},V],A(c_J)^{-1}]\;
A(c_K)[V,A(c_K)^{-1}])
+ h.c.
\nonumber
\ea
The sum is over vertices of $\gamma$ and triples of edges incident at them
(taken with outgoing orientation).
For each vertex $v$ and pairs of edges $c,c'$ outgoing from $v$ one defines 
$\alpha_{\gamma,v,c,c'}$ as that loop within $\gamma$ starting at $v$    
along $c$ and ending at $v$ along $(c')^{-1}$ with the minimal number of 
elements of $E(\gamma)$ used (if that loop is not unique, we average over 
them). It has been shown, that the concrete expression (\ref{3.23}) has
the correct
semiclassical limit in terms of expectation values with respect to semiclassical 
coherent states \cite{54} on sufficiently fine graphs of cubic topology 
\cite{55}.\\
\\
Remarkably, (\ref{3.23}) defines an essentially self-adjoint, diffeomorphism
invariant, 
continuum Hamiltonian operator for Lorentzian quantum gravity in four 
spacetime dimensions,
densely defined on the physical continuum Hilbert space ${\cal H}$ {\it 
which is manifestly free ov ultraviolet divergences}. That is, while 
for each given graph $\gamma$ the theory looks like a lattice gauge theory 
on $\gamma$, the theory is defined on all lattices simultaneously which 
makes it a continuum theory. Moreover, note that the vector $\Omega=1$ 
has norm unity and that $H\Omega=1$.\\
\\
Yet, one cannot be satisfied with (\ref{3.23}) for the following reasons:\\
1.\\
While it is true that one can give a better motivated derivation than we  
could sketch here for reasons of space, there are some ad hoc steps 
involved. \\
2.\\
There are several  
ordering ambiguities involved in (\ref{3.23}): Not only could we have 
written the factors in different orders but instead of using the 
fundamental representation to approximate connections in terms of holonomies 
we could have used higher spin representions \cite{56} or an avarage over 
several of them and in each case we would have different coefficients 
appearing in front of these terms.\\
3.\\
Of particular concern is definition of the minimal loop. While this gives 
good semiclassical results on sufficiently fine lattices, the theory lives on 
all lattices also those which are very coarse and on those the expression 
(\ref{3.23}) is doubtful, because the Riemann approximation mentioned above 
would suggest to use a much finer loop. In fact, one is supposed to take the 
regulator (i.e. the coordinate volume $\epsilon$ of the Riemann approximants) 
away
and in that limit the loop would shrink to zero. One can justify that 
this does not happen by using a sufficiently weak operator topoloy \cite{5}. 
Namely,
there exist diffeomorphism invariant distributions (linear functionals) $l$ on 
the dense span of SNWF $\psi$ \cite{8} and we define an operator $O_\epsilon$ to 
converge to an operator $O$ in that topology if $l([O_\epsilon-O]\psi)\to 0$ 
for all $l, \psi$. Now due to diffeomorphism invariance we can deform 
for any $\epsilon$ the small loop to any diffeomorphic one as long as we do not
cross other edges of the graph, in particular we can deform is as close as we 
want to the minimal one. Then the result mentioned above about the matrix 
elements of diffeomorphism invariant operators in fact forces us to choose 
that loop precisely, not only approximately. Of course, while the 
diffeomorphism symmetry of $H$ makes the space of diffeomorphism invariant 
distributions a natural space to consider, it is still not perfectly justified
to use it in order to define a topology.\\
4.\\  
The naive dequantisation 
of (\ref{3.23}) will perform poorly on very coarse graphs and will be far 
from the continuum expression (\ref{3.8}) but one could 
argue that that vectors supported on coarse graphs simply do not qualify as 
good semiclassical states.\\
5.\\
Using the same argument as in 3. there is nothing sacred about the minimal 
loop and one could take again other loops and/or average of over them
with certain weights. However, then the locality of (\ref{3.23}) is lost.\\
6.\\
The block diagonal or superselection structure \ref{3.22}) which is forced 
on us by the 
non-separability of the Hilbert space and its spatial diffeomorphism 
covariance appears unphysical, one would expect that the Hamiltonian 
creates also new excitations.\\
\\
It transpires that we must improve (\ref{3.23}) and the discussion 
has indicated a possible solution: Blocking free QFT from the continuum 
(i.e. restricting the Hilbert space to vectors of finite spatial 
resolution) with respect 
to a kinematic real space coarse graining scheme exactly produces such 
a high degree of non-locality at finite resolution even if the continuum 
measure or the continuum Hamiltonian is local \cite{13,16,26}. This bears the 
chance that what we see in (\ref{3.23}) is nothing but a naive guess of 
a continuum Hamiltonian which is blocked from the continuum but whose 
off-block diagonal form 
we cannot determine with the technology used so far. Accordingly 
this calls for shifting our strategy which was already started in 
\cite{55} in the sense that the block diagonal structure was dropped but 
only one infinite graph was kept:\\
\\
We take the above speculation serious and consider the operators 
$H_\gamma$ as projections onto the subspaces ${\cal H}_\gamma$ of $\cal H$
of a continuum Hamiltonian $H$ but we will drop the unphysical 
block diagonal structure \ref{3.22} which arises from the non-separability 
of ${\cal H}$. Rather the relation between the $H_\gamma$ is to be imposed 
by a renormalisation scheme induced by the path integral renormalisation 
scheme adopted in quantum statistical physics. To do this, we must 
first derive a path integral measure $\mu_\gamma$ from the OS data 
${\cal H}_\gamma, H_\gamma,\Omega_\gamma$ where $\Omega_\gamma$ is 
the vacuum of $H_\gamma$ by the usual Feynman-Kac-Trotter-Wiener formalism.
Then we can compute the flow of the $\mu_\gamma$ in the ususal way and 
then translate into a flow of OS data by OS reconstructing them from the 
measures. The fixed points of the flow will then define the 
possible continuum theories and these may be ``phases'' 
quite different from (\ref{3.23}).
The details of this programme will be the subject of the 
following sections.

\section{Constructive QFT, Feynman-Kac-Trotter-Wiener construction and 
Osterwalder-Schrader reconstruction}
\label{s4}

The purpose of this section is to provide some background information on 
constructive QFT and related topics such as the Feynman-Kac-Trotter-Wiener
construction of measures (path integrals) from a Hamiltonian formulation 
(operator formulation) and vice versa 
the Osterwalder-Schrader reconstruction of a Hamiltonain framework from 
a measure. Our description will be minimal. The prime textbook references are 
\cite{58,20}.

\subsection{Measure theoretic glossary}
\label{s4.1}

Let $S$ be a set. A collection $B$ of so-called measurable 
subsets of $S$ is called a 
$\sigma-$algebra if i. it is closed under taking complements with respect to 
$S$, ii. closed under taking countable unions and iii. $B$ contains the empty 
set $\emptyset$. The pair $(S,B)$ is called a measurable space. A measure 
space is a triple $(S,B,\mu)$ where $(S,B)$ is a measure space and $\mu$
is a positive set function $\mu:\; B\to \mathbb{R}^+_0\cup \{+\infty\}\; 
s\mapsto \mu(s)$
which is $\sigma-$additive, that is, for any pairwise disjoint 
$s_I\cap s_J=\emptyset, I\not=J;\; I,J=\in \mathbb{N}$ we have 
\be \label{4.1}
\mu(\cup_I s_I)=\sum_I \mu(s_I)
\ee
The measure $\mu$ is called a probability measure if $\mu(S)=1$. One uses 
the notation
\be \label{4.2}
\mu(s)=\int_s\;d\mu(p)=\int_S\; d\mu(p)\; \chi_s(p)
\ee
where $\chi_s(p)=1$ if $p\in s$ and $\chi_s(p)=0$ else is called the 
characteristic function of $s\in B$.  

Consider now a second measurable space $(\tilde{S}, \tilde{B})$. A function 
$X:\; S\to \tilde{S}$ is called measurable or a random variable 
if the pre-images 
$X^{-1}(\tilde{s})=\{p\in S; f(p)\in \tilde{s}\}$ 
of measurable sets $\tilde{s}\subset \tilde{S}$
are measurable in $S$.
Let ${\cal F}$ be the set of 
random variables $X:\; S\to \tilde{s}$ then for $X\in {\cal F}$ 
the set function 
\be \label{4.2}
\tilde{\mu}(\tilde{s}):=\mu(X^{-1}(\tilde{s})),\; \tilde{s}\in \tilde{B}
\ee
defines also a probability measure called the distribution of $\Phi$. We 
consider real valued functions $f:\; \tilde{S}\to \mathbb{R}$ of the simple form 
\be \label{4.3}
f(\tilde{p})=\sum_n \; z_n\; \chi_{\tilde{s}_n}(\tilde{p});\; 
z_n\in \mathbb{R},\;
\tilde{s}_n \in \tilde{B}
\ee
where the sum is over at most finitely many terms 
and define their integral as 
\ba \label{4.4}
\tilde{\mu}(f) &=& \sum_n z_n \mu(s'_n)=\int_{\tilde{S}} \; 
d\tilde{\mu}(\tilde{p}) 
[\sum_n z_n \chi_{\tilde{s}_n}(\tilde{p})]=\int_{\tilde{S}}\;
d\tilde{\mu}(\tilde{p})\; f(\tilde{p})
\nonumber\\
&=& \sum_n z_n \mu(X^{-1}(\tilde{s}_n))
=\int_S \; d\mu(p)\; [\sum_n z_n \chi_{X^{-1}(\tilde{s}_n)}(p)]
\nonumber\\
&=&\int_S \; d\mu(p)\; [\sum_n z_n \chi_{\tilde{s}_n}(X(p))]
=\int_S \; d\mu(p)\; (f\circ X)(p)
=\mu(f\circ X)
\ea
One can show that this identity extends from simple functions to Borel 
functions i.e. measurable functions 
$f:\tilde{S} \to \mathbb{R}$ where $\mathbb{R}$ is equipped with the 
Borel $\sigma-$algebra (the smallest $\sigma$ algebra containing all open 
intervals). We can then  also extend it to those complex functions whose 
real and imaginary parts are Borel by linearity. 

A stochastic process indexed by an index set $\cal I$ is a family 
$\{X_i\}_{i\in {\cal I}}$ of random variables $X_i: S\to \tilde{S}$. 
For any 
finite subset $I=\{i_1,..,i_N\}\subset {\cal I}$ we have the joint distribution
\be \label{4.5}
\tilde{\mu}_I(\tilde{s}_1\times ..\times \tilde{s}_N):=
\mu(\cap_{k=1}^N X_{i_k}^{-1}(\tilde{s}_k))
\ee
The probability measures $\mu'_I$ are called cylinder measures.
For any 
complex valued Borel function $f:\; \tilde{S}^N\to \mathbb{C}$ we have similarly
as in (\ref{4.3})
\be \label{4.6}
\int_S\; d\mu(p)\; f(\{X_{i_k}(p)\}_{k=1}^N)
=\int_{\tilde{S}^N} \; d\tilde{\mu}_I(\tilde{p}_1,..,\tilde{p}_n)\; 
f(\tilde{p}_1,..,\tilde{p}_n)
\ee
Functions on $S$ of the form $f_I(p)=f(\{X_{i_k}(p)\}_{k=1}^N)$ are called 
cylinder functions. 

In what follows we assume that for each $N\in \mathbb{N}_0$ 
there exists a distinguished system ${\cal W}_N$ of complex 
valued, bounded elementary functions $W$ 
on $N$ copies of $\tilde{S}$ such that the corresponding cylinder functions 
enjoy the following properties:\\
1. They generate an Abelian $^\ast$ algebra,
that is, for all $I,I'\in \cal{I}$ the product $W_I \; W'_{I'}$ is a finite,
complex linear combination of suitable 
$W^{\prime\prime}_{I^{\prime\prime}}, I^{\prime\prime}\in\cal{I},\;
W^{\prime\prime}\in {\cal W}_{|I^{\prime\prime}|}$ and also 
$\overline{W_I}$ is of that form.\\
2. ${\cal W}_N$ contains the constant function.\\
3. For each $I\in \cal{I}$, the moments $\mu(W_I),\; W\in {\cal W}_{|I|}$ 
determine $\tilde{\mu}_I$ uniquely.\\
4. These properties show that the $W_I$ are 
$L_2(d\tilde{\mu}_I,\tilde{S}^{|I|})$ 
functions. We require their span to be dense. \\ 
\\
We saw that a probability measure $\mu$ together with a stochastic process 
gives rise 
to a family of cylindrical probabilty 
measures $(\tilde{\mu}_I)_{I\in \mathfrak{I}}$ on $\tilde{S}^{|I|}$. The 
converse question is under which circumstances a cylindrical 
family of cylinder probability measures
determines a measure $\mu$. A necessary criterion is as follows:
The set $\cal{I}$ is partially ordered and directed by inclusion,
that is, for each $I,J\in \cal{I}$ we find $K\in \cal{I}$ such that
$I,J\subset K$ (for instance $K=I\cup J$). Suppose that $I\subset J$. Then 
\be \label{4.7}
\mu(X_I^{-1}(\tilde{s}_I))=\tilde{\mu}_I(\tilde{s}_I)=\mu(X_J^{-1}(
\tilde{s}_I\times 
(\tilde{S})^{|J|-|I|}))=\tilde{\mu}_J(\tilde{s}_I\times 
(\tilde{S})^{|J|-|I|})
\ee
where $X_I=\{X_i\}_{i\in I},\;\tilde{s}_I\subset \tilde{B}^{|I|}$. 
Furthermore for 
any permutation $\pi$ on $N=|I|$ elements set 
$\pi\cdot I=\{i_{\pi(1)},..,i_{\pi(N)}\}$ and 
$\pi\cdot \tilde{s}_I=\{(\tilde{p}_{\pi(1)},..,\tilde{p}_{\pi(N)});\;\;
(\tilde{p}_1,..,\tilde{p}_N)\in \tilde{s}_I\}$. Then 
\be \label{4.8}
\mu(X_{\pi\cdot I}^{-1}(\pi \cdot \tilde{s}_I)=
\tilde{\mu}_{\pi\cdot I}(\pi \cdot \tilde{s}_I)
=\mu(X_I^{-1}(\tilde{s}_I)=\tilde{\mu}_I(\tilde{s}_I)
\ee
Even more generally,
a partial order on the set $\mathfrak{I}$ of finite subsets $I$ of $\cal I$ 
is a transitive, 
reflexive and antisymmetric relation,
that is, $I<J\; \wedge \;J<K\;\Rightarrow\; I<K$ and $I<I$ and 
$I<J\; \wedge \;J<I\;\Rightarrow\; I=J$ for all $I,J,K\in \mathfrak{I}$. The 
set $\mathfrak{I}$ is called directed with respect to $<$ provided that for  
all $I,J\in \mathfrak{I}$ we find $K\in \mathfrak{I}$ such that $I,J<K$.
For $I<J$ we may have surjective maps 
$P_{JI}: \tilde{S}^{|J|}\to \tilde{S}^{|I|}$ such that 
$X_I(p)=P_{JI}(X_J(p))$ and such that for $I<J<K$ whe have 
$P_{JI}\circ P_{KJ}=P_{KI}$. Then similar as in (\ref{4.7}) we necessarily
must have for $I<J$
\be \label{4.7a}
\tilde{\mu}_I(\tilde{s}_I)=\tilde{\mu}_J(P_{JI}^{-1}(\tilde{s}_I))
\ee
It turns out that these two conditions, either (\ref{4.7}), (\ref{4.8}) 
or (\ref{4.7a}) are also 
sufficient in fortunate cases (for instance if $\tilde{S}=\mathbb{R}$, 
which is the classical Kolmogorov theorem, see \cite{37}). That is, we can 
then reconstruct the measure space $(S, B,\mu)$ and a stochastic process
$\{X_i\}_{i\in {\cal I}}$ such that the $\tilde{\mu}_I$ are the cylinder 
measures of $\mu$. It follows that the $W_I\in {\cal W}_{|I|}, 
I\in \cal{I}$
lie dense in $L_2(S, d\mu)$.\\
\\
Physical meaning:\\
We consider the elements $p\in S$ to be spacetime fields $\Phi$ or spatial
fields $\phi$ respectively. The index set $\cal I$ will 
have the meaning of a set of test functions or more generally
distributions whose elements $i$ label the 
random variables $X_i$. These map the fields smeared with test functions
to a finite dimensional manifold (usually copies of $\mathbb{R}$ or more 
generally of a Lie group). For instance for a scalar field $\Phi$ we may 
consider the random variable 
$X_F(\Phi)=\exp(i\int_{\mathbb{R}\times \sigma} d^4x F(x) \Phi(x))$ which takes 
values in $\tilde{S}=U(1)$. It is also customary to consider the field 
$p=\Phi$ itself as a random variable indexed by the same index set or 
to simply write $X_i(p)=p(i)$ as an abbreviation.

\subsection{Constructive QFT}
\label{s4.2}

The application of interest of the previous subsection is a stochastic 
process indexed by either $\mathbb{R}\times L$ or just by $L$ where 
the label set $L$ is a certain set 
of distributions on the spatial manifold. We distinguish
between random variables $\Phi$ indexed by a pair $(t,f)\in \mathbb{R}\times L$
and random variables $\phi$ indexed by $f\in L$. Some examples are:\\
\\
1. Real quantum scalar fields with smooth smearing:\\
Consider $L={\cal S}(\mathbb{R}^3)$, the space of smooth test functions
of rapid decrease 
and $\tilde{S}=\mathbb{R}$
equipped with the Borel $\sigma-$algebra. 
Then $\phi(f)=<f,\phi>=\int_\sigma\; d^3x \; f(x)\; \phi(x)$ and      
$\Phi(t,f)=<f,\Phi(t,.)>$. 
Given $F:=(f_1,..,f_N)\in L^N$ consider  
$\phi(F)=(\phi(f_1),..,\phi(f_N)\in \mathbb{R}^N$.  
The space ${\cal W}_N$ 
of elementary 
functions on $N$ copies of $\mathbb{R}$ can be chosen to be generated by 
the exponentials 
\be \label{4.9}  
w_{r_1,..,r_N}(\phi(F))=\exp(i\sum_{k=1}^N r_k <f_k,\phi>)
\ee
with $r_1,..,r_N\in \mathbb{R}$ labelling the (necessarily one dimensional)
unitary irreducible representations of $U(1)$.

In fact, since in this case the space 
$L$ is a vector space, it is sufficient to consdider the functions 
$w(\phi(f))=\exp(i\phi(f)),\; f\in L$. 
Analogously the space of elementary functions 
for the time dependent fields can be chosen as ($w_k\in {\cal W}_{N_k}$)
\be \label{4.10}
W(\Phi(t_1,F_1),..,\Phi(t_T,F_T))=w_T(\Phi(t_T,F_T))..w_1(\Phi(t_1,F_1))
\ee
which of course reduces to 
\be \label{4.10a}
\exp(i\Phi(t_T, f'_T))..\exp(i\Phi(t_1,f'_1)),\;\;f'_k\in L
\ee
for certain $f'_k\in L$.
Obviously the Abelian $^\ast-$algebra and boundedness conditions are 
satisfied. That these elementary functions suffice to determine the cylindrical 
measures requires a more involved argument (Bochner's theorem, \cite{37}).\\
\\
2. Real quantum scalar fields with distributional smearing:\\
Consider a subset $L=\subset{\cal S}'(\mathbb{R}^3)$ of the 
tempered distributions and $\tilde{S}=U(1)$
equipped with the Borel $\sigma-$algebra. In applications to scalar fields
coupled to General Relativity elements $f\in L$ are typically 
$\delta-$distributions supported at a single point.  

Then $\phi(f):=\exp(i<f,\phi>)$ and      
$\Phi(t,f)=\exp(i<f,\Phi(t,.)>)$ where $<f,\phi>$ is the evaluation of $f\in L$
on $\phi$. Given $F:=(f_1,..,f_N)\in L^N$ consider  
$\phi(F)=(\phi(f_1),..,\phi(f_N)\in U(1)^N$.  
The space ${\cal W}_N$ 
of elementary 
functions on $N$ copies of $U(1)$ can be chosen to be generated by 
the exponentials 
\be \label{4.9a}  
w_{r_1,..,r_N}(\phi(F))=\exp(i\sum_{k=1}^N r_k <f_k,\phi>)
\ee
with $r_1,..,r_N\in \mathbb{R}$ labelling the (necessarily one dimensional)
unitary irreducible representations of $U(1)$.
Analogously the space of elementary functions 
for the time dependent fields can be chosen as ($w_k\in {\cal W}_{N_k}$)
\be \label{4.10b}
W(\Phi(t_1,F_1),..,\Phi(t_T,F_T))=w_T(\Phi(t_T,F_T))..w_1(\Phi(t_1,F_1))
\ee
In this case we could still equip $L$ with the structure of a real vector space 
if we extend $L$ to the finite real linear combinations $\tilde{L}$ of its 
generating set $L$. Since this is no longer possible for the non-Abelian 
gauge theory example below, we will refrain from doing this, in order to 
highlight the structural similarity between the examples.\\
\\
3. Non-Abelian gauge fields for compact gauge groups $G$:\\
A form factor is a distribution 
\be \label{4.11}
f^a_c(x)=\int_c dy^a \delta^{(3)}(x,y)
\ee
where $c$ is a one dimensional path in i$\sigma$. We take $\tilde{S}=G$ 
equipped with the natural Borel $\sigma-algebra$
and 
\be \label{4.12}
\phi(c):=\phi(f_c):= 
{\cal P}\exp(\int_c \; \phi)
={\cal P}\exp(\phi(f_c));\;\;
\phi(f_c)=\int_\sigma\; d^3x \; f_c^a(x) \phi_a(x)        
\ee
where we have identified $\phi$ as a G connection and $\cal P$ denotes 
path ordering. Thus (\ref{4.12}) is the direct analog of the scalar 
field construction (note that the Lie generators are anti-self adjoint
since $G$ is compact so that (\ref{4.12}) is unitary) and 
$\phi(f_c)$ is simply the holonomy of $\phi$ along $c$. Likewise
\be \label{4.13}
\Phi(t,c):={\cal P}(\exp(\int_c \Phi(t,.))
\ee
Note that the form factors do not form a vector space, in general they cannot 
be added (unless two curves share a boundary point) and they can never
be multiplied by a non-integer real scalar (there is a certain 
groupoid structure behind this \cite{2}). Accordingly, our space of 
generating set of elementary functions 
${\cal W}_N$ on $N$ copies of $G$ need to be more sophisticated. We consider 
the space $L$ of form factors and for each $F=(f_{c_1},..,f_{c_N})\in L^N$
the pairing $\phi(F)=(\phi(c_1),..,\phi(c_N))\in G^N$.
Then a possible choice of generating set ${\cal W}_N$ 
of elementary functions is  
\be \label{4.14}
w_\delta(\phi(F))
=\prod_{k=1}^N \sqrt{d_{j_k}}\;
[\pi_{j_k}(\phi(c_k))]_{m_k,n_k}
\ee
with $\delta:=\{(j_1,m_1,n_1),..,(j_N,m_N,n_N)\}$.
In fact it is sufficient to consider
mutually disjoint (up to end points), piecewise real analytic curves $c_k$. 
Here $j$ labels an irreducible representation $\pi_j$ of $G$ of dimension 
$d_j$ and $[\pi_j(g)]_{m,n};\;m,n=1,..,d_j$ its matrix element functions.
By the Peter\& Weyl theorem these functions suffice to determine the 
cylindrical measures uniquely at least if they are absolutely continuous 
with respect to the product Haar measure. 
Likewise we consider the elementary functions 
\be \label{4.15}
W_{\delta_1,..,\delta_T}(\Phi(t_1,F_1),..,\Phi(t_T,F_T))
=w_{\delta_N}(\Phi(t_T,F_T))..
w_{\delta_1}(\Phi(t_1,F_1))
\ee
The fact that these functions satisfy all
requirements is the statement of Clebsch-Goradan decomposition theory together
with the properties of the holonomy to factorise along segments of a curve
(note the piecewise analyticity condition).\\
\\
This ends our list of examples.
We will denote the measure related to the stochastic process 
$\{\Phi(t,f)\}$ by $\mu$ and 
the measure related to the stochastic process 
$\{\phi(f)\}$ by $\nu$. As the notation suggests, $\Phi$ is a field defined 
on spacetime $M=\mathbb{R}\times \sigma$ while 
$\phi$ is a field defined on space $\sigma$. Note that $M=\mathbb{R}\times 
\sigma$ with $\sigma$ any 3D manifold is a consequence of the requirement 
of global hyperbolicity \cite{59}.\\
\\  
The measures $\mu$ underlying a relativistic QFT are not only probability 
measures. In 
addition, they need to satisfy a set of axioms \cite{20,23} called 
Osterwalder-Schrader axioms which, however, are tailored to $M=\mathbb{R}^4$,
stochastic processes with $L$ being a vector space and with an Euclidean 
background metric at one's disposal. In quantum gravity and more generally
in non-Abelian gauge theories one typically must or may want to drop
some of these structures. As a consequence we will only keep those axioms 
that can also be applied in this more general context.  

Some of them generalise to stochastic 
processes not indexed by a vector space, some do not. 
Some generalise from the manifold $\mathbb{R}^4$ to the general spacetime 
manifold $\mathbb{R}\times \sigma$ allowed by global hyperbolicity, some do
not. Fortunately, those 
that do generalise are sufficient for the reconstruction process 
\cite{22}. We call them the minimal OS axioms and we call a probability 
measure that satisfy them an OS measure. They are can be phrased as 
follows:\\
Let $\theta(t,x):=(-t,x)$ and $T_s(t,x):=(t+s,x)$ denote time reflection 
and time translation respectively. Let $w_k\in {\cal W}_{N_k},\; k=1,..,T,\; 
F_k\in L^{N_k}, \; t_k\in \mathbb{R}$ and 
\ba \label{4.16}
W_{(t_1,F_1),..,(t_T,F_T)} &:=& w_T(\Phi(t_T,F_T))...w_1(\Phi(t_1,F_1))
\nonumber\\
R\cdot W_{(t_1,F_1),..,(t_T,F_T)} &=& W_{(-t_1,f_1),..,(-t_T,f_T)},
\nonumber\\
U(s)\cdot W_{(t_1,F_1),..,(t_T,F_T)} &=& W_{(t_1+s,F_1),..,(t_T+s,F_T)}
\ea
Then we have the following conditions on the generating functional
\be \label{4.17}
\mu(W{(t_1,F_1),..,(t_T,F_T)})
\ee
I. Time Reflection invariance:\\
\be \label{4.18}
\mu(W_{(-t_1,F_1),..,(-t_N,F_T)})=
\mu(W_{(t_1,F_1),..,(t_T,F_T)})
\ee
II. Time translation invariance\\
\be \label{4.20}
\mu(W_{(t_1+s,F_1),..,(t_N+s,F_T)})=
\mu(W_{(t_1,F_1),..,(t_N,F_T)})
\ee
III. Time translation continuity\\
\be \label{4.21}
\lim_{s\to 0} \mu([W_{(t_1,F_1),..,(t_T,F_T)}]^\ast\;
W_{(t'_1+s,F'_1),..,(t'_{T'}+s,F'_{T'})})
=
\mu([W_{(t_1,F_1),..,(t_T,F_T)}]^\ast\;
W_{(t'_1,F'_1),..,(t'_{T'},F'_{T'})})
\ee
IV. Reflection positivity\\
Consider the vector space $V$ of the complex span of functions of the form 
$W_{(t_1,F_1),..,(t_T,F_T)}$ with $t_1,..,t_T>0$. Then for any 
$\Psi,\Psi'\in V$
\be \label{4.22}
<\Psi,\Psi'>:=\mu(\overline{\Psi}\; R\cdot\Psi'),\;\;
<\Psi,\Psi>\ge 0
\ee
Note that the stochastic process indexed by $\mathbb{R}\times L$ considers 
random variables $\Phi(t,f)$ at sharp points of time. It is often argued
that this index set provides an insufficient ``smearing'' in the time direction
and fails to cover interacting QFT at least in 3+1 spacetime dimensions
(in 1+1 and 2+1 dimensions, there are examples for which this works
\cite{60}). However, this argument rests on perturbative results as on 
3+1 dimensioonal Minkowski space so far no interacting QFT 
(obeying the Wightman axioms) has been rigorously constructed. It is still 
conceivable \cite{60a}
that in a non-perturbative construction of the theory, for which 
constructive QFT is designed, one can deal with fields at sharp time. 
One could of course be more general and consider stochastic processes indexed 
by some $L$ which now also includes smearing in the time direction and 
the formulation of reflection positivity will then constrain to elements 
of $L$ with positive time support, however, then the Wiener measure 
construction sketched below will not work. Our viewpoint is that this more 
general situation can be obtained from the sharp time construction because 
integrals of smearing functions with respect to time can be approximated 
by Riemann sums which in turn are nothing but integrals with respect to
sharp time smearing functions. \\
\\
At the moment it is rather unclear how and why $\mu,\Phi$ define a relativistic
QFT. This will become clear in the next subsection.

\subsection{Osterwalder-Schrader (OS) Reconstruction}
\label{s4.3}

The following abstract argument is standard \cite{20}. See \cite{13} for 
a proof adapted to the notation in this article.\\
\\
Due to reflection positivity (\ref{4.22}) defines a postive semi-definite 
sesqui-linear form on $V$. We compute its null space $N$ and complete the 
quotient of equivalence classes $V/N$ in the inner product (\ref{4.20}) 
to a Hilbert space. Given $\Psi\in V$ we denote its equivalence class 
$\Psi+N$ by $[\Psi]_\mu$ where we keep track of the measure dependence of 
the quotient construction. By construction the $D=[V]_\mu$ is dense in 
$\cal H$. 
Since the constant function $\Psi=1\in V$ we define 
a ``vacuum'' vector by $\Omega:=[1]_\mu$. Finally we define for $s\ge 0$
\be \label{4.23}
K(s)[\Psi]_\mu:=[U(s)\Psi]_\mu
\ee
The constraint $s\ge 0$ is due to the time support condition in the 
definition of $V$.
One must show that this is well defined (independent of the representative)
\cite{20}. By virtue of their definition (\ref{4.16}) the $U(s)$ form 
a one parameter Abelian group of operators $U(s) U(s')=U(s+s')$
on $L_2(S,d\mu)$. This 
implies that the $K(s)$ form a one parameter Abelian semi-group due to
the constraint $s\ge 0$ (again one must show that the definition 
is well-defined). Time translation continuity (\ref{3.21}) translates
into weak continuity of the semi-group. Furthermore, by time translation 
invariance \ref{4.20} the $U(s)$ define unitary, in particular bounded 
operators on $L_2(S,d\mu)$ which translates into the statement that the 
$K(s)$ form a contraction semi-group. Thus \cite{20} there exists a 
positive self-adjoint operator $H$, called ``Hamiltonian'' 
on $\cal H$ such that $K(s)=e^{-s H}$.
Obviously $K(s)\Omega=\Omega$, thus $\Omega$ is a ground state for $H$
which justifies the name ``vacuum''. \\
\\
This elegant argument is deceivingly simple. To actually compute the 
{\it Osterwalder-Schrader triple} $({\cal H}, \Omega, H)$ from 
$\mu$ and to relate it to the fields and Hamiltonian in terms of which 
one would construct the quantum theory using canonical quantisation 
is not clear yet. However, one can again use the following abstract argument
\cite{13}. Suppose that there is an Abelian $C^\ast-$algebra $\mathfrak{B}$
of bounded operators on ${\cal H}$ such that $\mathfrak{B}\Omega$ is dense
(the $C^\ast-$ norm is inherited from the uniform operator topology).
It is not difficult to show that this is always the case when ${\cal H}$ 
is separable which is the only case that we will consider in our application
to renormalisation, but it it also holds in many non-separable situations,
see appendix \ref{sb} for a proof.
Let $\Delta(\mathfrak{B})$ be its Gel'fand spectrum \cite{61} (which is a 
compact space) i.e. the space of all $^\ast$ homomorphisms 
$\phi:\; \mathfrak{B}\to \mathbb{C}$.
Then by Gel'fands theorem, $\mathfrak{B}$ can be thought
of as the space $C(\Delta(\mathfrak{B}))$ i.e. the continuous functions 
on the spectrum which is an Abelian $C^\ast-$algebra with respect to the 
sup-norm. The correspondence (Gel'fand isomorphism) 
is given by $\hat{b}(\phi)=\phi(b)$ for all $\phi\in \Delta(\mathfrak{B})$
and in fact this is an isometric isomorphism of $C^\ast-$algebras. Consider 
now the linear functional
\be \label{4.24}
\nu(\hat{b}):=<\Omega,b\Omega>
\ee 
which by construction is positive $\nu(|\hat{b}|^2)=||b\Omega||^2$. By the 
Riesz-Markov theorem \cite{58} there exists a (regular Borel) probability 
measure 
on $S':=\Delta(\mathfrak{B})$ which by abuse of notation we also 
denote by $\nu$ such that 
\be \label{4.25}
\nu(\hat{b})=\int_{\Delta(\mathfrak{B})}\; d\nu(\phi)\; \hat{b}(\phi)
\ee
That is to say, the Hilbert space $\cal H$ obtained from OS reconstruction 
can be thought of as $L_2(\Delta(\mathfrak{B}),d\nu)$ under the 
isomorphism $b\Omega\mapsto \hat{b}$, in particular $\Omega$ corresponds to 
the constant function equal to $1$. We thus have managed to cast $\cal H$ into 
the language of measure theory on the set $S'=\Delta(\mathfrak{B})$. 
The fields $\phi$ that come out of this 
construction are random variables indexed by some index set $L'$, that 
is, we have shown that we can always construct such a measure and a 
corresponding stochastic process. We think of the field $\phi$ as the 
spatial configuration fields underlying a canonical
quantisation approach.
A priori, however, it is not clear what $L'$ is altough it must be related 
in some way to $\mathbb{R}^+\times L$. In the case of free fields one can show 
that in fact one can choose $\mathfrak{B}$ in such a way that
$L'=L$ due to the quotient construction involved in $\cal H$
but even then it is a priori not clear how $\Phi(t,f)$ and $\phi(f),\; f\in L$ 
are related. Again in the case of free fields one shows that $\phi(f)$ 
can be thought of as $\Phi(0,f)$, the spacetime field at sharp time zero.
However, in general the relation between the 
stochastic processes underlying $\Phi$ and $\phi$ may be more complex. 
In any case, the operator $H$ translates in this 
language into the operator   
\be \label{4.26}
\hat{H} \hat{b}:=\widehat{H b \Omega}
\ee

\subsection{Feynman-Kac-Trotter-Wiener (FKTW) construction}
\label{s4.4}

Given an OS triple $({\cal H}, \Omega, H)$ 
we saw at the end of the previous subsection that without
loss of generality we can assume that ${\cal H}=L_2(S',d\nu)$ where 
$\nu$ is a probability measure on $S$ equipped with a Borel $\sigma-$algebra
and that we are given a stochastic propcess $\phi(f),\; f\in L$ indexed by 
some index set $L$, at least when $\cal H$ is separable (which will be the 
case in our applications). Moreover, $\Omega=1$ in this presentation of 
$\cal H$ is cyclic for some $C^\ast-$algebra of functions on $S'$. We pick 
some set ${\cal W}_N,\; N\in \mathbb{N}_0$ of elementary functions 
$w\in {\cal W}_N$ 
subject to the conditions 1.-4. spelled out just after (\ref{4.6}) and 
for  $F=(f_1,..,f_N)\in L^N$ have $\phi(F)=
(\phi(f_1),..,\phi(f_N))\in (S')^N$
as well as
\be \label{4.27}      
w_F(\phi)=w(\phi(F))
\ee
Let now $T\in \mathbb{N}_0,\; t_1<t_2<..<t_T$ and $F_k\in L^{N_k},\; w_k\in 
{\cal W}_{N_k}$. We 
consider the expectation value functional  
\be \label{4.28}
<\Omega,w_{T,F_T}\; e^{-(t_T- t_{T-1})H}\; 
w_{T-1,F_{T-1}}\; e^{-(t_{T-1}- t_{T-2})H}\;...\;
e^{-(t_2- t_1)H}\;w_{1,F_1}\;\Omega>
\ee
Consider now a stochastic process $\Phi(s,f)$ indexed by $(s,f)\in 
\mathbb{R}\times L$ and the elementary functions
\be \label{4.28}
W_{(t_k,F_k)_{k=1}^T}(\Phi)=w_{T,F_T}(\Phi(t_N,.))\;..\;
w_{1,F_1}(\Phi(t_1,.))
\ee
Then the Wiener measure $\mu$, if it exists, evaluated on 
(\ref{4.28})
\be \label{4.29}
\mu(W_{(t_k,F_k)_{k=1}^T})
\ee
is supposed to equal (\ref{4.27}). The non-trivial question is why this 
should be the case, under which circumstances and how to construct 
$\mu$. For this we consider the integral kernel $K_\beta$ of the operator 
$e^{-\beta H},\; \beta>0$, that is, 
\be \label{4.30}
[e^{-\beta H}\psi](\phi)=:\int_S\; d\nu(\phi')\; 
K_\beta(\phi,\phi')\;\psi(\phi')
\ee 
Note the semigroup property  
\be \label{4.31}
\int_{S'}\; d\nu(\phi)\; K_{\beta_1}(\phi_1,\phi) K_{\beta_2}(\phi,\phi_2)
=K_{\beta_1+\beta_2}(\phi_1,\phi_2)
\ee
Define $S:=\prod_{t\in \mathbb{R}} S'$. For each $T\in \mathbb{N}_0$ consider 
$t_1<..<t_T$ and measurable sets 
$s'_{t_k} \subset S'$ and define the set function
\be \label{4.31}
\mu([\times_{k=1}^T\; s'_{t_k}]\;\times\;[\times_{t\not\in \{t_1,..,t_T\}} S'])
:=\int_{[S']^T} d\nu(\phi_1)...d\nu(\phi_T) 
\prod_{k=1}^T \chi_{s'_k}(\phi_k)\; \prod_{k=1}^{T-1}
K_{t_{k+1}-t_k}(\phi_{k+1},\phi_k)
\ee
It is not clear that this is a positive set function, but when it is, it is
called the Wiener measure generated by the OS triple. For sufficient criteria
for this property called Nelson-Symanzik positivity in the case of scalar 
fields, see \cite{62}. Basically one needs to show 
that matrix elements of $e^{-\beta H}$ between positive functions 
is positive. 
Note that for $s'_k=S'$ for all $k$ we get
\be \label{4.32}
\mu(S)
=<\Omega, e^{-(s_T-s_1)H}\Omega>=1
\ee
This shows that $\mu$ is a probability measure on 
$S$. For quantum mechanical Schr\"odinger Hamiltonians one can use the 
Trotter product formula and the Wienner measure of the heat kernel 
to prove positivity \cite{63}  (Feynman-Kac formula).\\
\\
One can now show the following \cite{13}:
\begin{Theorem} \label{th4.1} ~\\
i. Suppose that OS data $({\cal H},H,\Omega)$ are given and that the 
corresponding Wiener measure $\mu$ exists. Then $\mu$ is an OS measure 
and its OS reconstruction reproduce the given OS data up to unitary 
equivalence.\\
ii. Suppose that an OS measure $\mu$ is given thus producing OS data 
$({\cal H}, \Omega, H)$. Then the corresponding Wiener measure exists 
and reproduces $\mu$ up to equivalence of measure spaces.
\end{Theorem}
Here measure spaces $(S_j,B_j,\mu_j);\; j=1,2$ are called equivalent if there 
exists a bijection $F:\; S_1\to S_2$ such that both $F,F^{-1}$ are measurable
and such that $\mu_1=\mu_2\circ F$. The reason for why we generically only
reproduce an equivalent and not an identical starting point lies in the 
large freedom in the choice of the stochastic process $\phi$ when performing
the OS reconstruction step.

\section{Renormalisation}
\label{s5}

\subsection{Motivation}
\label{5.1}

Our motivation for renormalisation comes from the current state of affairs
with respect to the definition of the quantum dynamics in LQG as outlined 
in section \ref{s3}. In that 
case the Hilbert space ${\cal H}=L_2(S',d\nu)$ is precisely of the form 
we envisage here. Moreover we have a vacuum $\Omega$ for a candidate 
Hamiltonian $H$ that, however, we are not sure whether all steps of the 
quantisation process that led to $H$ are justified. Namely we have 
defined $H$ as $H_\gamma$ on certain mutually orthogonal 
subspaces ${\cal H}_\gamma$ preserving it using a choice of discretisation 
of the classical continuum expression which has naively the correct 
dequantisation if the graph $\gamma$ fills the spatial manifold 
$\sigma$ sufficiently densely. The definition of elementary 
functions in (\ref{4.14}) precisely reproduces the SNWF and thus the 
spatial connection defines a stochastic process indexed by graphs.\\
\\
As already mentioned at the emnd of 
section \ref{s3} we would like to take a fresh look 
at the problem. As usual in constructive QFT, if $\sigma$ is not already 
compact, we replace it with a compact manifold $\sigma_R$ where $R$ is 
an infrared cut-off (IR) which we remove in the end $R\to\infty$ 
(thermodynamic limit). In order not to clutter the notation, the dependence 
on $R$ of all considerations that follow will be suppressed.
Next we do not consider all finite graphs $\gamma$ 
(taking all finite graphs leads to a non-separable Hilbert space) but only a 
controllable countable family $\cal M$ therein which however is such that the 
discretised 
classical variables (configuration and momentum fields) in terms of which 
we perform the quantisation separate the points of the classical phase space 
when all the graphs in $\cal M$ are at our disposal. 
The set $\cal M$ is supposed to partially ordered and directed.
The motivation for doing so stems from the spatial diffeomorphism invariance
of the classical LQG Hamiltonian: The algebraic form of the Hamiltonian 
discretised on diffeomorphic graphs is identical. This is precisely 
the starting point of the algebraic quantum gravity proposal \cite{55}
where it was emphasised that one can quantise gravity in terms of abstract
graphs which gain their physical meaning only after choosing an embedding
supplied, for anstance, by a semiclassical state.

To have some intuitive picture in mind consider 
$\sigma=\mathbb{R}^3$ with toroidal compactification $\sigma_R=T^3$ (where 
each direction has length $R$ with respect to the Euclidian background metric
on $\mathbb{R}^3$ and with periodic boundary conditions installed) 
and $\Gamma$ the set of all finite graphs $\sigma_R$ of cubic topology.
This is still an uncountable set which we now restrict to a countable one 
as follows. Each element of $\Gamma$ is uniquely 
labelled by $M\in \mathbb{N}$ where 
$M^3$ is the number of vertices of the graph (one could generalise 
this and have different numbers of vertices in each direction). We pick 
once and for all a coordinate system and locate the vertices of $\gamma_M$
at the points 
\be \label{5.1}
m \epsilon_M,\; m\in \mathbb{Z}_M^3,\; \mathbb{Z}_M=\{0,1,..,M-1\},\; 
\epsilon_M=\frac{R}{M}
\ee
where the edges of the graph are straight lines into the coordinate directions 
between the vertices.
We equip ${\cal M}:=\mathbb{N}$ with the following partial order: $M<M'$ iff 
$\frac{M'}{M}\in \mathbb{N}$. Note that this implies $\gamma_M\subset 
\gamma_{M'}$ since 
\be \label{5.2}
m\epsilon_M=m\frac{M'}{M}\epsilon_{M'}=:m' \epsilon_{M'}
\ee
with $m'\in \mathbb{Z}_{M'}^3$ and because the edges of the 
graphs are straight lines into the coordinate directions. 
This is certainly not a linear order because 
not all natural numbers are in relation but still 
equips $\Gamma$ with a direction:
Given $M,M'$ take for instance $M^{\prime\prime}=MM'$ then 
$M,M'<M^{\prime\prime}$ (more efficiently take $M^{\prime\prime}$ as 
the least common multiple). It is clear that for $M$ sufficiently large 
discretised phase space variables obtained by integrating continuum 
variables over $0-$ or $1-$ dimensional subsets of $\gamma_M$ (vertices or 
edges) or by integrating 
momentum variables over $3-$ or $2-$ dimensional subsets of the cell complex 
corresponding to $\gamma_M$ (faces and cubes)
will separate the points of the continuum phase space. 
Instead of $\gamma_M$ one could also use the cubic cell complex 
$\gamma_M^\ast$ dual to $\gamma_M$ defined by saying that 
the barycentres of the cubes of 
$\gamma_M^\ast$ coincide with the vertices of $\gamma_M$. However, in the 
spirit of economy we will not use the additional structure $\gamma_M^\ast$ 
in what follows.

\subsection{Discretisation of phase space}
\label{s5.2}

In canonical quantisation we start with a continuum phase space coordinatised 
by configuration fields $\phi^J$ and canonically conjugate momentum fields 
$\pi_J$ in terms of which the classical continuum Hamiltonian $H$ is formulated. 
Here the index $J$ corresponds to an internal symmetry and is typically
Lie algebra valued.
We now would like to consider a discretisation of both the phase space and 
the Hamiltonian, one for each lattice $M$, while keeping track of how these
fields $\phi^J_M,\pi^M_J$ are related to the continuum fields $\phi^J,\pi_J$. 
The idea for how to do this stems from the observation that by construction of 
generally covariant field theories the fields 
$\phi^J,\pi_J$ are {\it dual} 
in the sense that there is a natural bilinear 
form $<\pi,\phi>'_{I'\times K'}:=\sum_J\; <\pi_J,\phi^J>_{I\times K}$ 
on the phase space (usually a cotangent bundle
$T^\ast K'$) $I'\times K'$ of 
momentum and configuration 
fields respectively
where $<.,.>$ is spatially diffeomorphism invariant. Note that $<.,.>',\;<.,.>$
just differ by tracing over the internal directions in field space i.e. 
$I'=I^d, K'=K^d$ where $d$ is the number of internal directions in field space.

For instance, the momentum of a scalar field is 
geometrically a scalar density of weight one, so that 
$<\pi,\phi>'=<\pi,\phi>:=\int_\sigma\; d^3x\; \pi(x)\; \phi(x)$. 
The momentum of a G connection 
is geometrically a Lie algebra valued vector field density so that
$<\pi,\phi>'=\sum_J \int_\sigma \; d^3x \; \pi^a_J(x) \phi_a^J(x)$
This also 
holds for higher $p-$forms as they occur in some supergravity theories
as well as for (standard model or Rarita-Schwinger) fermions.
Note that the bilinear 
form is in general not invariant under the internal symmetry group but
this will not be important for what follows. The fact that $\pi,\phi$ are 
conjugate is the statement that their canonical brackets are 
\be \label{5.3}
\{<\pi,k'>',<i',\phi>'\}=<i',k'>'
\ee
for all $(i',k')\in I'\times K'$.    

The fact that the bilinear form $<.,.>$ is naturally at our disposal 
motivates a natural choice for the index set $L,L^\ast$ of the stochastic 
process $\phi,\pi$. Namely we choose $L$ to be a certain distributional 
extension of $I$ and likewise $L^\ast$ as a 
certain distributional extension of $K$. These extensions should be such 
that $<i,k>$ remains well defined for $i\in L, k\in L^\ast$.  
For instance, for a scalar field we may choose $L$ as the set of $\delta$ 
distributions $f_p(x)=\delta_p(x)$ with support at single points 
$p\in\sigma$ and $L^\ast$ as the 
set of characterisic functions $g^R(x)=\chi_R(x)$ 
of connected $D-$dimensional submanifolds $R$ of 
$\sigma$. For a compact $G-$connection we can choose $L$ as the set 
of form factors 
$f^a_c(x):=\int_c dy^a \delta(x,y)$ with support 
on (piecewise analytic) curves $c$.  
For $L^\ast$ we would consider the set of dual form factors of the 
form $g^S_a(x):=1/(D-1)! 
\int_S \epsilon_{ab_1 .. b_{D-1}} dy^{b_1}\wedge
.. dy^{b_{D-1}} \delta(x,y)$ with support on (piecewise analytic) $D-1$ 
submanifolds $S$. We may also have opportunity to consider their 
Lie algebra valued versions  
$f^{aJ}_c(x) =\tau^J f^a_c(x)\in L', \; g^S_{aJ}(x)=\tau_J g^S_a(x)\in 
(L')^\ast$
were $\tau^J,\tau_J$
are dual bases 
in the defining representation of the Lie algebra 
of $G$ such that ${\rm Tr}(\tau^J \tau_K)=\delta^J_K$.
Note that we deliberatively do not make use of the 
fact that these distributions are elements of vector spaces. This is because 
we would like to have a uniform description of both linear and non-linear
theories. In the case of linear theories the description can be significantly
simplified as we have done in \cite{13}.  
  
The connection to section \ref{s4.2} is then as follows: For each 
$f\in L, \phi\in K$ we consider a map $(f,\phi)\mapsto \phi(f) \in \tilde{S}$.
For linear theories one usually takes $\tilde{S}=U(1)$ and for a G gauge theory 
one takes $\tilde{S}=G$. The object $\phi(f)$ exploits the existence 
of the natural bilinear 
form $<.,.>$. For instance for a scalar field one considers 
$\phi(f_p)=\exp(i<f_p,\phi>)$ while for a G connection we consider the holonomy 
$\phi(f_c)={\cal P}\exp(<f_c,\phi^J>\tau_J)$. For each 
$N\in \mathbb{N}$ we consider $F=(f_1,..,f_N)\in L^N$ and define 
$\phi(F)=(\phi(f_1),..,\phi(f_N))\in \tilde{S}^N$. The space of elementary 
functions 
${\cal W}_N$ consists of maps $\tilde{S} \to \mathbb{C}$ subject to the 
conditions 
listed at the beginning of section \ref{s4.2}. We may generate ${\cal W}_N$ 
from monomials 
labelled by matrix element functions of finite dimensional unitary 
representations of $\tilde{S}$, see (\ref{4.14}).\\
\\
For each $M\in\mathbb{N}$ let $L_M$ be the space of discrete functions 
on the lattice consisting of 
$M^D$ points with values in 
$\mathbb{R}^{\mathfrak{t}}$ where $\mathfrak{t}$ is tensorial number of 
configuration
(or momentum) degrees of freedom per spatial point ($\mathfrak{t}=1$ for 
scalar fields,
$\mathfrak{t}=D$ for a G Yang-Mills theory in $D+1$ spacetime dimensions etc).
That is, an element $l_M\in L_M$ assigns to each point $m\in \mathbb{Z}_M^D$
a vector in $\mathbb{R}^\mathfrak{t}$. The space $L_M$ carries an auxiliary
real Hilbert space sructure ($L_M$ is of course a 
finite dimensional vector space), e.g. for a G Yang-Mills theory
\be \label{5.3}
<l_M,\tilde{l}_M>_{L_M}=\sum_{m\in \mathbb{Z}_M^D}\;\sum_{a=1}^{\mathfrak{t}}\; 
l_M(m,a)
\;\tilde{l}_M(m,a)
\ee
for any $l_M, \tilde{l}_M\in L_M$ and we wrote $[l_M]^a(m)=:l_M(m,a)$. 
\begin{Definition} \label{def5.1} ~\\
A discretisation of the continuum phase space $I\times K$ 
subordinate to $M\in \mathbb{N}$ is a pair of linear maps 
\be \label{5.4}
I_M:\; L_M \to L;\; K_M:\; L_M \to L^\ast
\ee
with the following properties:\\
i. For any $l_M, l'_M \in L_M$ 
\be \label{5.5}
<I_M l_M, K_M l'_M>_{I\times K}=<l_M,l'_M>_{L_M}
\ee
That is to say $I'_M K_M=K'_M I_M ={\rm id}_{L_M}$ where 
$I'_M:\; I \to L_M,\; K'_M:\; K\to L_M$ are the dual maps defined 
by 
\be \label{5.6}
<I_M l_M, \phi>_{I\times K}=<l_M,I'_M \phi>_{L_M},\;
<\pi, K_M l_M>_{I\times K}=<K'_M \pi,l_M >_{L_M},\;
\ee
ii. For any $M<M'$ define the {\it injection  maps}
\be \label{5.7}
I_{M M'}:=K_{M'}' I_M;\;\;K_{M M'}:=I_{M'}' K_M:\;
L_M \to L_{M'}
\ee
Then we require
\be \label{5.8}
I_{M'} I_{MM'}=I_M,\;\; K_{M'} K_{MM'}=K_M
\ee
\end{Definition}
To see how this gives rise to discretised configuration and momentum variables 
let $\delta_M^{m,a},\;\delta^M_{m,a}\in L_M$ with 
$m\in \mathbb{Z}_M^D,\; a=1,..\mathfrak{t}$
be the Kronecker functions 
$[\delta_M^{m,a}]_b(\tilde{m})
:=\delta^a_b\;\delta_{m,\tilde{m}}$ and $[\delta^M_{m,a}]^b(\tilde{m})
:=\delta^b_a\delta_{m,\tilde{m}}$. Then the following functions 
on the continuum phase space 
\be \label{5.9}
(\pi_M)^a_J(m):=<\pi_J, K_M \delta_M^{m,a}>_{I\times K}, 
(\phi_M)_a^J(m):=<I_M \delta_M^{m,a},\phi^J>_{I\times K} 
\ee
enjoy canonical brackets 
\be \label{5.10}
\{(\pi_M)^a_J(m),(\phi_M)_b^K(\tilde{m})\}=\delta^a_b\; \delta^K_J\;
\delta_{m,\tilde{m}}
\ee
where the first condition (\ref{5.5}) was used. Thus, (\ref{5.5}) makes sure 
that the discretisations (\ref{5.9}) enjoy canonical brackets, so we 
call (\ref{5.5}) the {\it symplectomorphism property}.
The motivation for the second condition (\ref{5.8}) will become clear only 
later, however we note that it implies that for all $M<M'<M^{\prime\prime}$
\be \label{5.11}
I_{M' M^{\prime\prime}} I_{M M'}
=K'_{M^{\prime\prime}} [I_{M'} I_{MM'}]=I_{M M^\prime\prime}
\ee
which we thus call {\it cylindrical consistency property}. Likewise
$K_{M'M^{\prime\prime}} K_{MM'}=K_{MM^{\prime\prime}}$. It says that 
injecting a function into the continuum is independent from which resolution 
scale $M$ this is done.

Finally we will impose a further restriction on the maps $I_M,K_M$ which 
amounts to a convenient choice of normalisation and thus is called {\it 
normalisation property}. Namely we require that for all $M<M'$ the 
map $I_{M M'}:\; L_M\to L_{M'}$ restricts to $B_M\to B_{M'}$ where $B_M$
is the set of functions on $\mathbb{Z}_M^D$ with values in the {\it bit space}
$\{0,1\}^{\mathfrak{t}}$. This condition is only necessary in the non 
Abelian case and there avoids overcounting.
   
We note that (\ref{5.8}) defines elements $\pi_M=K_M' \pi,\phi_M=I_M' \phi$ 
of $L_M^{d\mathfrak{t}}$ that
we can now use to try to define a discretisation 
$H_M=H_M[\{(\pi_M)^a_j(m),(\phi_M)_a^j(m)\}_{a,j,m}]$ 
of the Hamiltonian $H=H[\pi,\phi]$.
For instance, if the Hamiltoinian depends only quadratically on the fields then
one may try (including discretisations of spatial derivatives and some 
spatial averages)
\be \label{5.12}
H_M:=H[\pi=I_M \pi_M,\phi=K_M \phi_M]
\ee
For interacting Hamiltonians, more sophisticated approximations must be used.
Certainly the expression for $H_M$ is in general plagued by a large amount
of discretisation ambiguity beyond the choice of discretised variables. 
On the other hand, the fact that 
$\pi_M=K'_M \pi,\;\phi_M=I_M' \phi$ are conjugate will be convenient when 
constructing $H_M$ and it is efficient to construct them motivated by the 
naturally available bilinear form $<.,.>'$ on the phase space.

To see that there are non-trivial examples for such maps, consider a scalar
field in $D$ spatial dimensions compactified on a torus with Euclidian 
coordinate length $R$ in all directions. Then (recall $\epsilon_M=R/M$)
\be \label{5.13}
(I_M l_M)(x):=\sum_{m\in \mathbb{Z}_M^D} l_M(m) \delta_{m\epsilon_M}(x),\;\;
(K_M l_M)(x):=\sum_{m\in \mathbb{Z}_M^D} l_M(m) \chi_{m\epsilon_M}(x)
\ee
where 
\be \label{5.14}
\chi_{m\epsilon_M}(x)=\prod_{a=1}^D \chi_{[m^a\epsilon_M,(m^a+1)\epsilon_M)}(x)
\ee
where the latter denotes the characteristic functions of left closed -- right 
open intervals. This clopen interval structure is very important in order that 
(\ref{5.5}) and (\ref{5.8}) are satisfied \cite{13}. Similar constructions work 
for gauge fields, see appendix \ref{sa}. 
Note that we changed here the notation as compared to 
\cite{13}: The maps $I_M, E_M$ used there are called here $K_M, I'_M$ 
respectively. The motivation for this change of notation is to make it manifest
how much of the structure is in fact already canonically provided by the 
structure of the classical theory. \\
\\
Given the lattice in $D$ spatial dimensions labelled by $M\in \mathbb{N}$
we consider in general $N=M^D \mathfrak{t}$ degrees of freedom 
$\phi(I_M l_M)=:\phi_M(l_M):=\{\phi(I_M l_M^{m,a})\}_{m\in \mathbb{Z}_M^D,\;
a=1,..,\mathfrak{t}}\in \tilde{S}^N$ where $l_M^{m,a}=l_M \; \delta_M^{m,a}$ and 
$l_M$ is restricted to the subset $B_M\subset L_M$ of functions 
$\mathbb{Z}_M^D\to F_2^{\mathfrak{t}}$ where $F_2=\{0,1\}$ is the field 
in two elements ({\it bit space}). Thus $l_M(m,a)\in \{0,1\}$ is restricted 
to the information whether the degree of freedom $\phi(I_M l_M^{m,a})$ is 
excited or not. This is justified because 1. the missing information about 
the strength of the excitation is encoded in the representation label, see 
below and 2. because the maps $I_{M M'}$ restrict to maps $B_M\to B_{M'}$ by 
assumption.  

The space of elementary functions ${\cal W}_M$ 
on the lattice labelled by $M$ 
is then generated by 
\be \label{5.14}
w^M_{j,n,\tilde{n}}(\phi_M(l_M))=
w^M_{j,n,\tilde{n}}(\phi(I_M l_M))=
\prod_{m,a}\; 
[\pi_{j_{m,a}}(\phi(I_M l_M^{m,a}))]_{n_{m,a},\tilde{n}_{m,a}}
\ee
Here $j_{m,a}$ labels an irreducible representation 
$\pi_{j_{m,a}}$ of $G$ (one from 
each equivalence class), $d_{j_{m,a}}$ is its dimension, and 
$[\pi_{j_{m,a}}(.)]_{n_{m,a},\tilde{n}_{m,a}}$ denote its matrix elements 
with $n_{m,a},\tilde{n}_{m,a}\in \{1,..,d_{j_{m,a}}\}$.

To see how (\ref{5.14}) interacts with the map $I_{MM'}$ in the case 
of non-Abelian gauge theorie we note 
that the cylindrical consistency property of $I_{MM'}$ implies 
\be \label{5.15}
w^M_{j,n,\tilde{n}}(\phi_M(l_M))
=w^M_{j,n,\tilde{n}}(\phi(I_M l_M))
=\sum_\alpha  
w^{M'}_{j',n'_\alpha,\tilde{n}'_\alpha}(\phi_{M'}  I_{MM'} l_M))
\ee  
where the notation is as follows (see appendix \ref{sa}): 
$j'_{m' a}=[I_{MM'} l_M] (m',a)\; j_{[m'/M' M],a}$ where $[.]$ denotes the 
Gauss bracket, $n'_{m',a}=n_{m,a}$ if $m'=M'/M m$,    
$\tilde{n}'_{m',a}=\tilde{n}_{m,a}$ if $m^{\prime a}+1=M'/M (m^a+1), 
m^{\prime b}=M'/M m^b; b\not=a$, and otherwise the sum over $\alpha$ denotes 
the sum over all $\tilde{n}'_{m',a}=n'_{m'+\delta_a,a}\in 
\{1,..,d_{j'_{m',a}}\}$ with $[\delta_a]^b:=\delta_a^b$ that arise 
by writing 
the holonomy along the edge labelled by $m\in \mathbb{Z}_M^D,a=1,..,D$ 
as products of 
holonomies along edges labelled by $m'\in \mathbb{Z}_{M'}^D,a$.   

In general therefore we see that for any generating function 
$w^M\in {\cal W}_M$ we have for all $M<M'$
\be \label{5.16}
w^M(\phi_M(l_M))=\sum_\alpha z_\alpha\; w^{M'}_\alpha(\phi_{M'}(I_{MM'} l_{M'}))
\ee
where the sum over $\alpha$ involves a finite, unique set of generating 
functions $w^{M'} \in {\cal W}_{M'}$ and $z_\alpha$ are certain, definite  
complex numbers. Similar statements then  of course hold for the stochastic 
process labelled by $\mathbb{R}\times L$ and for the functions 
\be \label{5.17}
W^M(\Phi_M(t^{(1)},l_M^{(1)}),..,\Phi_M(t^{(T)},l_M^{(T)}))=
w^M_T(\Phi_M(t^{(T)},l_M^{(T)}))...  
w^M_1(\Phi_M(t^{(1)},l_M^{(1)}))
\ee

\subsection{Hamiltonian Renormalisation}
\label{s5.3}

Abstracting from the concrete lattice implementation and field content 
above we are 
in the following situation: There is a partially ordered and directed 
label set ${\cal M}$ and for each $M\in {\cal M}$ we have a map 
$I_M: B_M \to L^{N(M)}$ where $L$ is the index set of the stochastic process
$\phi$, $N(M)\in \mathbb{N}$ is the number of elements of $L$ in the image of 
$I_M$
and $B_M=\{0,1\}^{N(M)}$. Then $\phi_M(l_M):=\phi(I_M l_M)\in\tilde{S}^{N(M)}
=:\tilde{S}_M$ 
and we have a generating set of elementary functions 
$w^M: \tilde{S}^{N(M)}\to \mathbb{C}$.    

Suppose that for each $M\in {\cal M}$ we have 
discretised the system somehow as sketched above and picked some OS 
triple $({\cal H}^{(0)}_M, \Omega^{(0)}_M,\; H^{(0)}_M)$ with 
${\cal H}^{(0)}_M=L_2(d\nu^{(0)}_M,\tilde{S}_M)$. That is to say, we 
have a stochastic process $\{\phi_M(l_M)\}_{M\in {\cal M}}$ 
indexed by $B_M$ and probability 
measures $\nu^{(0)}_M$ on $\tilde{S}_M$. The Hamiltonian $H^{(0)}_M$ preserves 
${\cal H}^{(0)}_M$ and annihilates the unit 
vector $\Omega_M^{(0)}\in {\cal H}^{(0)}_M$ which is cyclic. We 
consider a space of elementary functions ${\cal W}_M$ such that 
in particular the $w^M(\phi_M(l_M))\Omega^{(0)}_M;\; w^M\in {\cal W}_M, 
l_M\in B_M$ lie dense 
in ${\cal H}^{(0)}_M$. 

Using the Feynman-Kac-Trotter-Wiener (FKTW) construction, we obtain a family of 
OS measures $\mu^{(0)}_M$ on $S_M=\prod_{t\in \mathbb{R}} \tilde{S}_M$ 
which 
can be probed using a stochastic process $\Phi_M(t,l_M)$ labelled by 
$\mathbb{R}\times B_M$. This measure family $\{\mu^{(0)}_M\}_{M\in {\cal M}}$
will generically not be cylindrically consistent, and therefore does not 
define a continuum measure $\mu$ because of the discretisation ambiguities 
involved in the construction of the $H^{(0)}_M$ which determines 
$\mu^{(0)}_M$. If it was then we would have for $w_1,..,w_T\in {\cal W}_M$
\be \label{5.15}
\mu(w_T(\Phi(t_T,I_M l^T_M))...w_1(\Phi(t_1,I_M l^1_M))
=\mu_M(w_T(\Phi_M(t_T, l^T_M)..w_1(\Phi_M(t_1,l^1_M))
\ee
Using $I_M=I_{M'} I_{MM'}$ for $M<M'$ we would find the identity
\be \label{5.16}
\mu_M(w_T(\Phi_M(t_T, l^T_M)..w_1(\Phi(t_1,l^1_M))=
\mu_{M'}(w_T(\Phi_{M'}(t_T, I_{MM'}l^T_M)..w_1(\Phi_{M'}(t_1,I_{MM'} l^1_M))
\ee
called the condition of cylindrical consistency. 

As reviewed in section \ref{s3}, condition (\ref{5.16}) grants the existence 
of $\mu$ under rather generic conditions. The strategy (see also 
\cite{39}) is therefore to construct an iterative 
sequence of measure families $\mathbb{N}_0\ni n\mapsto 
\{\mu^{(n)}_M\}_{M\in {\cal M}}$ {\it called renormalisation (group) flow}
with initial family as above 
such that the fixed point family satisfies (\ref{5.16}). We refer to section
\ref{sc} for the reader interested in more notions of the renormalisation
group in the language of measure theory. 

The scheme that we will employ in fact does not make use of (\ref{5.16}) for all 
$M<M'$ but only $M'=p^n M$ where $p$ is a prime. The simplest choice is $p=2$
but we have tested the formalism also for $p=3,5$ \cite{13} and mixtures 
thereof in the case of free scalar QFT. 
This in fact does cover all possible $M$ because any natural number 
can be written as $k p^l;\; k,p$ relative prime but the fixed point family 
could depend 
on $k$. Of course one assumes that the fixed point family is independent 
of the choices of $p,k$ as an expression of universality as confirmed again
for simple systems \cite{13}. Thus we define as renormalisation flow
\be \label{5.17}
\mu^{(n+1)}_M(w_T(\Phi_M(t_T, l^N_M)..w_1(\Phi(t_1,l^1_M))
=\mu^{(n)}_{M'}(w_T(\Phi_{M'}(t_T, I_{MM'}l^N_M)..
w_1(\Phi_{M'}(t_1,I_{MM'} l^1_M))
\ee
for $M':=2M$. 

Having then obtained $\mu^\ast$ from cylindrically consistent projections 
$\mu^\ast_M$ we want to construct the OS triple 
$({\cal H}^\ast, \Omega^\ast,H^\ast)$ using OS reconstruction. 
However, while we are sure that $\mu^{(0)}_M$ is an OS measure for each 
$M$ by theorem \ref{th4.1}, we are a priori not granted that $\mu^{(n)}_M$ is an 
OS measure, i.e. that the flow preserves the OS measure class. This 
in fact shown in \cite{13}.
\begin{Theorem} \label{th5.1} ~\\ 
The renormalisation flow (\ref{5.17}) preserves the OS measure class 
and its fixed points define OS measures $\mu^\ast$.
\end{Theorem}
Responsible for this result is the fact that the time operations that
define an OS measure commute with the spatial coarse graining operation.

Thus in  principle we can perform renormalisation in the measure (or path 
integral) language and then carry out OS reconstruction in order to 
find the continuum Hamiltonian theory that we are interested in.
On the other hand, the fact that FKTW construction and OS reconstruction are 
inverses of each 
other (theorem \ref{th4.1}) allows for the possibility to map the 
renormalistion flow of measures directly into a renormalisation flow of OS 
triples. In detail:\\
\\
Step 1: {\it Identifying the stochastic processes}\\
\\
We need to work out the null space of the reflection 
positive sesquilinear form determined by the measure $\mu^{(n)}_M$ 
from the vector space
$V_M$ of finite linear combinations of vectors of the form 
\be \label{5.18} 
w^M_T(\Phi_M(t_T,l^T_M))..w^M_T(\Phi_M(t_T,l^T_M))
\ee
with $t_T>t_{T-1}>..>t_1>0$ for $w^M_K\in {\cal W}_M$
(for coinciding points of time we can reduce the number of time steps by 
decomposing the products of elementary functions into linear combinations 
of those). 

The Hilbert space ${\cal H}^{n)}_M$ is then the completion of the span 
of equivalence classes $[\psi_M]_{\mu^{(n)}_M}, \; \psi_M\in V_M$, 
in particular the vacuum is $\Omega^{(n)}_M=[1]_{\mu^{(n)}_M}$. However,
the abstract description in terms of equivalence classes is not very 
useful in practice, 
rather we would like to describe them concretely in terms of stochastic 
processes and measures $\nu^{(n)}_M$ as outlined in section \ref{s4.3}. As 
the Hilbert spaces we deal with are separable, this is always possible, see 
appendix \ref{sb}, however, that construction does not directly 
refer to the spacetime stochastic process $\Phi$ we started from. The reason 
for why this happens is because of the appearing equivalence classes: To perform 
concrete calculations one will work with representatives which makes the 
construction non-canonical because the choice of such representatives 
is largely a matter of taste. In our setting, if $\mu^{(n)}_M$ is obtained 
by the FKTW construction from OS data, then of course $\Phi_M(0,.)=:\phi_M(.)$
is a possible choice.
However, in the renormalisation step, we are to deduce the OS data at 
resolution $M$ from the measure $\mu^{(n+1)}_M$ which was renormalised 
from $\mu^{(n)}_{M'},\; M<M'$ via (\ref{5.17}) and thus it is not a 
priori clear how the stochastic process $\phi_M$ can be chosen, in particular,
it is not clear whether it can be chosen as 
$\phi_M(.)=\Phi_{M'}(0,I_{MM'}.)$ 
which appears to be a natural choice. 

However, we are in a better situation than in the generic case because it 
is clear that ${\cal H}^{(n+1)}_M$ can be formulated 
in terms of the fields $\phi_M(t,l_M)=\Phi_{M'}(t,I_{MM'} l_M)$ for a minimal
number of distinguished times $t\in \tau$ where the set $\tau$ is determined 
by the quotienting process, see \cite{13} for an example.  
Alternatively one can view the fields 
$\phi_M(t,l_M),t\in \tau$
as fields at time zero $\tilde{\phi}_M(\tilde{l}_M)$ but in a larger space of 
fields, i.e. a stochastic process $\tilde{\phi}_M$ with a larger index set 
$\tilde{B}_M=\tau\times B_M$ that still lives on the lattice labelled 
by $M$ \cite{13}.
It follows that without further input, which will be provided below, 
$\phi_M(.)=\Phi_{M'}(0,I_{MM'}.)$ is in general 
a {\it compound field}, i.e. a 
{\it collective degree of freedom} composed out of the $\tilde{\phi}_M$ which 
together with its momentum $\pi_M$ is insufficient to define the Hamiltonian
$H^{(n+1)}_M$ which will generically depend on the larger set 
of variables $\tilde{\phi}_M$ and its 
conjugate momentun $\tilde{\pi}_M$.  \\
\\
Step 2: {\it Working out the flow of OS triples}\\
\\
Using the correspondence between the Wiener measures $\mu^{(n)}$ and the 
corresponding operator expressions we have for $t_T>..>t_1$
\ba \label{5.19}
&& \mu^{(n+1)}_M(w^M_T(\Phi_M(t_T,l^T_M)..w^M_1(\Phi_M(t_1,l^M_1))
\nonumber\\
&=&
<\Omega^{(n+1)}_M, w^M_T(\phi_M(l^T_M))\; e^{-(t_T-t_{T-1}) H^{(n+1)}_M}\; 
w^M_{T-1}(\phi_M(l^{T-1}_M))\;e^{-(t_{T-1}-t_{T-2}) H^{(n+1)}_M}\; ...\;
\nonumber\\
&& e^{-(t_2-t_1) H^{(n+1)}_M}\; 
w^M_1(\phi_M(l^1_M))\; \Omega^{(n+1)}_M>_{{\cal H}_M^{(n+1)}}
\nonumber\\
&=& \mu^{(n)}_{M'}(w^M_T(\Phi_{M'}(t_T,I_{MM'}l^T_M)..
w^M_1(\Phi_{M'}(t_1,I_{MM'} l^M_1))
\nonumber\\
&=&
<\Omega^{(n)}_{M'}, w^M_T(\phi_{M'}(I_{MM'}l^T_M))\; e^{-(t_T-t_{T-1}) 
H^{(n)}_{M'}}\; 
w^M_{T-1}(\phi_{M'}(I_{MM'}l^{T-1}_M))\;
e^{-(t_{T-1}-t_{T-2}) H^{(n)}_{M'}}\; ...\;
\nonumber\\
&&
e^{-(t_2-t_1) H^{(n)}_{M'}}\; 
w^M_1(\phi_{M'}(I_{MM'} l^1_M))\; \Omega^{(n)}_{M'}>_{{\cal H}_{M'}^{(n)}}
\ea
for all choices of $M\in {\cal M};\; 
T\in \mathbb{N}_0;\; t_T>..>t_1;\; l^1_M,..,l^T_M$ (in practice e.g.
$M'=2M$ is fixed). 

We consider (\ref{5.19}) as the {\it master equation} from which everything 
must be be deduced. To avoid the compound field phenomenon mentioned above 
we use that (\ref{5.19}) i. is supposed to hold for an arbitrary number of 
time steps and ii. we add as further input 
one more OS axiom namely {\it uniqueness of the 
vacuum} which is in fact a standard axiom to impose in  QFT on Minkowski
space \cite{20}. In terms of measures it can be stated as {\it ergodicity
of time translations}
\be \label{5.19a}
\lim_{T\to \infty}\; \frac{1}{2T}\;
\int_{-T}^T\; ds\; U(s)\Psi=_{\mu\; {\rm a.e.}}\; \mu(\Psi)\; \cdot 1,\;\;\Psi
\in L_2(S,d\mu)
\ee 
We separate this axiom from the minimal ones because it enters 
in a crucial way only at this last stage of the renormalisation process.
The subsequent discussion considerably extends the arguments of \cite{13}.\\
\\ 
First of all, going back to (\ref{5.19}) and picking $T=1$ we find   
\be \label{5.25}
<\Omega^{(n+1)}_M, w^M_T(\phi_M(l_M))\;\Omega^{(n+1)}_M>_{{\cal H}_M^{(n+1)}}
=
<\Omega^{(n+1)}_M, w^M(\phi_{M'}(I_{MM'} l_M))
\;\Omega^{(n)}_{M'}>_{{\cal H}_{M'}^{(n)}}
\ee
Using the fact that the $w_M\in {\cal M}$ form a $^\ast-$algebra, we can 
formulate (\ref{5.25}) as follows: Assuming by induction that
up to renormalisation step $n$ 
the vectors 
$w^M(\phi_M(l_M))\;\Omega^{(n)}_M,\; w^M\in {\cal W}_M$ span a dense 
subspace of ${\cal H}^{(n)}_M$, consider  
the closed linear span $\widehat{{\cal H}^{(n)}_{M'}}$ 
of vectors of the form
\be \label{5.26}
w^M(\phi_{M'}(I_{MM'} l_M))\;\Omega^{(n)}_{M'}
\ee
which is a subspace of ${\cal H}^{(n)}_{M'}$. Then (\ref{5.25}) 
is the statement that the map 
\be \label{5.27}
J^{(n)}_{MM'}:\; {\cal H}^{(n+1)}_M\to  
\widehat{{\cal H}^{(n)}_{M'}};\;
w^M(\phi_M(l_M))\;\Omega^{(n+1)}_M\mapsto
w^M(\phi_{M'}(I_{MM'} l_M))\;\Omega^{(n)}_{M'}
\ee
is an {\it isometry}, that is,
\be \label{5.28}
[J^{(n)}_{MM'}]^\dagger\; J^{(n)}_{MM'}=1_{{\cal H}^{(n+1)}_M}
\ee
which implies that 
\be \label{5.29}
P^{(n)}_{MM'}:=J^{(n)}_{MM'} \; [J^{(n)}_{MM'}]^\dagger:\; 
{\cal H}^{(n)}_{M'}\to 
\widehat{{\cal H}^{(n)}_{M'}}
\ee
is a projection. 

Next for $T=2, t_2-t_1=\beta$ we find from (\ref{5.19})
\ba \label{5.31}
&&<\Omega^{(n+1)}_M, w(\phi_M(l_M))
\; e^{-\beta H^{(n+1)}_M}\; w'(\phi_M(l'_M))
\Omega^{(n+1)}_M>_{{\cal H}^{(n+1)}_M}
\nonumber\\
&=&
<\Omega^{(n)}_{M'}, w(\phi_{M'}(I_{MM'} l_M))
\; e^{-\beta H^{(n)}_{M'}}\; w'(\phi_M(I_{MM'} l'_M))
\Omega^{(n)}_{M'}>_{{\cal H}^{(n)}_{M'}}
\ea
and using again the $^\ast-$ property of the algebra ${\cal W}_M$ and taking 
formally the first derivative of (\ref{5.31}) at $\beta=0$ we conclude 
\be \label{5.32}
H^{(n+1)}_M=[J^{(n)}_{MM'}]^\dagger\;H^{(n)}_{M'} \; J^{(n)}_{MM'}  
\ee
Note that (choose $w=1$ in (\ref{5.27}))
\be \label{5.33}
H^{(n+1)}_M \Omega^{(n+1)}_M=[J^{(n)}_{MM'}]^\dagger H^{(n)}_{M'}
\Omega^{(n)}_{M'}=0
\ee
hence the new vacuum is automatically annihilated by the new Hamiltonian. 

We notice that for finite $\beta$ (\ref{5.31}) is not implied by 
(\ref{5.32}) unless 
$[H^{(n)}_{M'},P^{(n)}_{MM'}]=0$ and it is here where we use the condition 
that the correspondence (\ref{5.19}) is to hold for an arbitrary number 
and choices of time as well as the uniqueness of the vacuum. 
Using the projection $P^{(n)}_{MM'}$ onto the closed linear span 
of the $w(\phi_{M'}(I_{MM'} l_M))\Omega^{(n)}_{M'}$ we see that 
the operators $w(\phi_{M'}(I_{MM'} l_M))$ on ${\cal H}^{(n)}_{M'}$ 
are block diagonal with respect to the  decomposition 
\be \label{5.36}
{\cal H}^{(n)}_{M'}=
P^{(n)}_{MM'} {\cal H}^{(n)}_{M'}  
\oplus [P^{(n)}_{MM'}]^\perp 
{\cal H}^{(n)}_{M'}
\ee
since they together with their adjoints 
leave $P^{(n)}_{MM'}{\cal H}^{(n)}_{M'}$ 
invariant (the $w\in {\cal W}_M$ generate a $^\ast-$algebra). Thus 
$P^{(n)}_{MM'} w(\phi_{M'}(I_{MM'} l_M)) [P^{(n)}_{MM'}]^\perp=
[P^{(n)}_{MM'}]^\perp w(\phi_{M'}(I_{MM'} l_M)) P^{(n)}_{MM'}=0$ but 
in general \\
$[P^{(n)}_{MM'}]^\perp w(\phi_{M'}(I_{MM'} l_M)) [P^{(n)}_{MM'}]^\perp
\not=0$. Thus, it is not sufficient to insert $w$ operators an arbitrary 
number of times and at arbitary 
places into the the correspondence (\ref{5.19}) in order to deduce 
(\ref{5.32}) from (\ref{5.19}).

Let $b^{(n)}_{M' I}$ be an orthonormal basis of 
$P^{(n)}_{MM'} {\cal H}^{(n)}_{M'}$. Then, since $\Omega^{(n)}_{M'}$ 
is cyclic for the algebra $W^{(n)}_{MM'}$ generated by the 
$w(\phi_{M'}(I_{MM'} l_M))$ with respect to 
$P^{(n)}_{MM'}\;{\cal H}^{(n)}_{M'}$ 
we find $w'_I \in W^{(n)}_{MM'}$ such that $b^{(n)}_{M'I}=
w'_I \Omega^{(n)}_{MM'}$ (or can be made  at least arbitrarily close).
Next, assume that $\Omega^{(n)}_{M'}$ is the unique ground state for 
$H^{(n)}_{M'}$ then  
\be \label{5.37}
e^{-\beta H^{(n)}_{M'}}\to \Omega^{(n)}_{M'} 
<\Omega^{(n)}_{M'},.>_{{\cal H}^{(n)}_{M'}} 
\ee
becomes the projection on the ground state for $\beta\to \infty$. It follows 
in the limit $\beta\to\infty$
\be \label{5.38}
P^{(n)}_{MM'}=
\sum_I b^{(n)}_{M' I}\; <b^{(n)}_{M' I},.>_{{\cal H}^{(n)}_{M'}} 
=\sum_I\;  w'_I \Omega^{(n)}_{M'}\; 
<w'_I \Omega^{(n)}_{M'},.>_{{\cal H}^{(n)}_{M'}}  
\to \sum_I\; w'_I \; e^{-\beta H^{(n)}_{M'}}\; (w'_I)^\dagger
\ee
Let $w_I$ be the element in the algebra generated by the $w(\phi_M(l_M))$ such 
that $J^{(n)}_{MM'} w_I \Omega^{(n+1)}_M=w'_I \Omega^{(n)}_{M'}$ (which exists 
because $P^{(n)}_{MM'} {\cal H}^{(n)}_{M'}$ is the closure of the image 
of $J^{(n)}_{MM'}$). Then due to isometry $[J^{(n)}_{MM'}]^\dagger 
J^{(n)}_{MM'}
=1_{{\cal H}^{(n+1)}_M}$ (\ref{5.28}) we have  
\ba \label{5.39}
&& \sum_I\;  w_I \Omega^{(n+1)}_M\; 
<w_I \Omega^{(n+1)}_M,.>_{{\cal H}^{(n+1)}_M}  
=[J^{(n)}_{MM'}]^\dagger\; 
\sum_I\;  w'_I \Omega^{(n)}_{M'}\; 
<w'_I \Omega^{(n)}_{M'},.>_{{\cal H}^{(n)}_{M'}}  \; J^{(n)}_{MM'}
\nonumber\\
&=& [J^{(n)}_{MM'}]^\dagger\; P^{(n)}_{MM'}\; J^{(n)}_{MM'}
=1_{{\cal H}^{(n+1)}_M}
\ea
On the other hand, if $\Omega^{(n+1)}_M$ is the unique ground state for 
$H^{(n+1)}_M$ we have by the same argument as in (\ref{5.37}), (\ref{5.38}) 
in the limit $\beta\to \infty$
\be \label{5.40}
1_{{\cal H}^{(n+1)}_M}\to \sum_I\; w_I e^{-\beta H^{(n+1)}_M}\; w_I^\dagger
\ee
Since the identity operator $1_{{\cal H}^{(n+1)}_M}$ can be inserted an 
arbitrary number of times and at arbitrary places on the left hand side 
of (\ref{5.19}) and since it can be written as (\ref{5.40}) which under 
the correspondence (\ref{5.19}) translates into (\ref{5.38}) 
the correspondence 
(\ref{5.19}) is to hold also when we insert $P^{(n)}_{MM'}$ an arbitrary number
of times and at arbitary places on the right hand side of (\ref{5.19}).
In particular this means that we must replace on the right hand side 
of (\ref{5.19}) the operator $e^{-\beta H^{(n)}_{M'}}$ by 
\be \label{5.41}
\lim_{N\to \infty} 
P^{(n)}_{MM'} (e^{-\frac{\beta}{N} H^{(n)}_{M'}} P^{(n)}_{MM'})^N
\ee
To see this, we write in (\ref{5.19}) for each $k=2,..,T$ and for any $N\in 
\mathbb{N}$ on the lhs
$e^{-(t_k-t_{k-1})H^{(n+1)}_M}=(e^{-(t_k-t_{k-1})H^{(n)}_M/N}\;
1_{{\cal H}_M^{(n+1)}})^N$ and 
replace $1_{{\cal H}^{(n)}_M}$ by the approximants (\ref{5.40}) 
or more precisely the $P(n,k,\beta)$ of 
appendix \ref{se} for $P=1_{{\cal H}^{(n+1)}_M}$. Using multi-linearity of 
(\ref{5.19}) we can 
rewrite the resulting expression in terms of (\ref{5.19}) again, just 
that now we have not $T$ insertions of $w$ operators but $T'=2\;N\; T$ 
insertions at times $t'_1<..<t'_{T'}$ such that 
\be \label{5.41a}
t'_{2kN+2l}-t'_{kN+2l-1}=\frac{t_k-t_{k-1}}{N},\;
t'_{2kN+2l-1}-t'_{kN+2l-2}=\beta,\;
\ee
for $k=1,..,T-1; l=1,..,N$. By the correspondence (\ref{5.19}) this 
translates into the corresponding expressions on the rhs with 
approximants (\ref{5.38}) or more precisely the $P(n,k,\beta)$ of 
appendix \ref{se} for $P=P^{(n)}_{MM'}$.     
Then one takes strong limits in the appropriate order, see 
appendix \ref{se}, in particular $\beta\to\infty$, keeping $t_k-t_{k-1}$ 
fixed. As this is to hold for all $N$, we take $N\to\infty$.   

Formula (\ref{5.41}) is known in the mathematics 
literature \cite{67} as as a degenerate 
case of a Kato-Trotter product \cite{66} of which there are many 
versions. One of them states that for 
contraction semigroups generated by self-adjoint operators $A,B$ such that
$A+B$ is essentially self-adjoint on the dense domain $D(A)\cap D(B)$ we have 
strong convergence 
\be \label{5.42}
\lim_{N\to \infty} [e^{-A/N} \; e^{-B/N}]^N=e^{-(A+B)}
\ee
In our case the second contraction semigroup $s\mapsto e^{-s B}$ is replaced by 
the degenerate one $K(s)=K(0)=P^{(n)}_{MM'}$. In \cite{67} sufficient 
criteria for the existence of a degenerate semigroup $K(\beta),\; K(0)$ an
invariant projection, rather than the identity are studied, such that in, 
say the strong 
operator topology $K(\beta)=\lim_{N\to \infty} (e^{-\beta/N A} P)^N$. 
Assuming that existence $K(\beta)$ 
of the limit (\ref{5.41}) is secured, we deduce 
\be \label{5.43}
H^{(n+1)}_M:=-[J^{(n)}_{MM'}]^\dagger\;
[\frac{d}{d\beta} K(\beta)]_{\beta=0}\; J^{(n)}_{MM'},\;
K(\beta):=\lim_{N\to \infty}
\;P^{(n)}_{MM'}\;[e^{-\frac{\beta}{N} H^{(n)}_{M'}}\; P^{(n)}_{MM'}]^N
\ee
In particular, if the solution of (\ref{5.43}) is given by  
\be \label{5.44}
K(\beta)=P^{(n)}_{MM'} e^{-\beta P^{(n)}_{MM'} H^{(n)}_{M'} P^{(n)}_{MM'}}
\ee
we recover (\ref{5.32}) since 
$P^{(n)}_{MM'}\; J^{(n)}_{MM'}=J^{(N)}_{MM'}$. In appendix \ref{sd} we prove 
(\ref{5.44}) for the case that $H^{(n)}_{MM'}$ is bounded, that is,
(\ref{5.44}) is strictly true when replacing $H^{(n)}_{M'}$ by 
its bounded spectral projections $E^{(n)}_{M'}(B), \; B$ Borel. 

In what follows we will assume this to hold also when 
$e^{-\beta H^{(n)}_{M'}}$ is a general contraction semi-group. In \cite{67}
we find proofs for existence of a resulting degenerate semi-group under 
special circumstances but no concrete formulae in terms of the original 
projection and semi-group are given. Thus for the time being we will use
(\ref{5.32}) as a plausible solution of the 
{\it exact relation} (\ref{5.43}) but keep in mind that (\ref{5.43})
may contain more information.

To conclude this step, under the assumption that uniqueness of the vacuum
is is preserved under the renormalisation flow and that the 
degenerate Kato-Trotter product formula applies to general contraction 
semi-groups, we can strictly derive 
(\ref{5.28}) and (\ref{5.32}) as equivalent to (\ref{5.19}). Unfortunately,
it is not possible to show that the uniqueness 
property is automatically preserved under the flow for suppose that 
$H^{(n)}_{M'}$ has unique vacuum $\Omega^{(n)}_{M'}$ and that 
$H^{(n+1)}_M v_M=[J^{(n)}_{MM'}]^\dagger H^{(n)}_{M'} J^{(n)}_{MM'} v_M=0$
then we can just conclude that $H^{(n)}_{M'} J^{(n)}_{MM'} v_M\in 
[P^{(n)}_{MM'}]^\perp {\cal H}^{(n)}_{M'}$. Hence without further 
input, the uniqueness property
must be checked self-consistently.\\
\\
Step 3: {\it Constructing the continuum theory from the fixed point data}\\
\\
Once we found a fixed point family $J_{MM'}; ({\cal H}_M, \Omega_M,\; 
H_M)$ with $M<M', \; M,M'\in {\cal M}$ we have an inductive limit structure
$(J_{MM'},{\cal H}_M)$ of Hilbert spaces since 
$J_{M' M^{\prime\prime}} J_{MM'}=J_{MM^{\prime\prime}}$ is inherited 
from 
$I_{M' M^{\prime\prime}} I_{MM'}=I_{MM^{\prime\prime}}$ 
for $M<M'<M^{\prime\prime}$ and therefore can define the continuum Hilbert 
spcace ${\cal H}$ as its inductive limit which always exists
\cite{51}. Thus, there 
exist isometries $J_M:\; {\cal H}_M\to {\cal H}$ such that 
$J_{M'} J_{MM}'=J_M,\;M<M'$. Moreover, there exists a consistently defined 
quadratic form $H$ (not necessarily an operator) such that 
$H_M=J_M^\dagger H J_M$. Note that we can compute matrix elements of $H$ 
between 
the subspaces $J_M {\cal H}_M, J_{M'} {\cal H}_{M'}$ of $\cal H$ for any 
$M,M'$ without actually knowing $H$, just its known finite resolution 
projections are needed, by using any $M,M'<M^{\prime\prime}$
\be \label{5.35}
<J_M \psi_M, H J_{M'} \psi_{M'}>_{{\cal H}}
=<J_{M^{\prime\prime}} J_{MM^{\prime\prime}} \psi_M, H 
J_{M^{\prime\prime}} J_{M'M^{\prime\prime}} \psi_{M'}>_{{\cal H}}
=<J_{MM^{\prime\prime}} \psi_M, H_{M^{\prime\prime}} 
J_{M'M^{\prime\prime}} \psi_{M'}>_{{\cal H}_{M^{\prime\prime}}}
\ee
We stress that $H$ is {\it not} the inductive limit of the $H_M$ since that 
would require
$H_{M'} J_{MM'}= J_{MM'} H_M$. This inductive limit condition is much
stronger than the quadratic form condition 
$H_M=J_{MM'}^\dagger H_{M'} J_{MM'}$ which can be seen by 
multiplying the inductive limit condition 
from the left with $J_{MM'}^\dagger$ and using isometry. 
It is not possible to derive the inductive limit condition from the quadratic 
form condition because $J_{MM'}^\dagger$ has no left inverse. 
 
We emphasise that this Hamiltonian renormalisation scheme can be seen as an 
{\it independent, real space, kinematical} renormalization flow 
different from the OS measure (or path integral) 
scheme even if the assumptions that were made during its derivation 
from the measure theoretic one are violated. Note that both schemes 
are {\it exact}, i.e. make no 
truncation error. This is possible because we do not need to compute 
the spectra of the Hamiltonians (which is practically impossible to do 
analytically without error even at finite resolution) but only matrix
elements which is computationally much easier and can often performed 
anaytically, even if the Hilbert spaces involved are {\it infinite dimensional}
as is the case in bosonic QFT even at finite resolution. 

As a final remark, recall that the reduction of (\ref{5.19}) to 
(\ref{5.28}) and (\ref{5.28}) rests crucially on the assumption that 
the vacuum vectors $\Omega^{(n)}_M$ remain the unique ground states 
of the Hamiltonians $H^{(n)}_M$ in the course of the renormalisation,
a condition which is difficult to keep track off in practice and 
which in fact contains dynamical information. Is it possible that 
the OS measure flow and the Hamiltonian flow
(\ref{5.28}) and (\ref{5.32}) nevertheless 
deliver the same continuum theory, even if we drop the vacuum
uniqueness assumption?  
In that respect, note that one arrives 
at (\ref{5.28}), (\ref{5.32}) from (\ref{5.19})  
by deleting by hand the off-block diagonal terms in $H^{(n)}_{M'}$ with 
respect to the decomposition (\ref{5.36}). 
When deleting those terms by hand then (\ref{5.19}) 
indeed becomes equivalent to (\ref{5.28}) and (\ref{5.32}). 
This is reminescent of the Raleigh-Ritz procedure of diagonalising a 
self-adjoint 
operator \cite{66}: There the statement is that for any self-adjoint operator 
$H$ bounded from below (which is precisely our situation) and any finite 
dimensional projection $P$ the $\dim(P)$ eigenvalues of $PHP$ ordered by size
are upper bounds of the the $\dim(P)$ eigenvalues, ordered by size,  
in the discrete part of the spectrum (i.e. isolated eigenvalues of 
finite multiplicity) of $H$. Here we deal with an infinite projection
instead of a finite one, but the general setting is the same. The idea is that 
as we increase $M$ we approach the continuum Hamiltonian for which eventually 
there are no off-diagonal elements.

\section{Connection with density matrix and entanglement renormalisation}
\label{s7}

In this brief section we display natural points of contact between 
density matrix renormalisation \cite{25,30}, entanglement renormalisation
\cite{33,34} and 
projective renormalisation \cite{29a,29b,29c,29d}. In particular the notions of 
ascending and descending superoperators, isometries and disentanglers 
from the multi scale entanglement renormalisation Ansatz (MERA) 
can be applied.\\
\\
Recall that our Hamiltonian renormalisation 
scheme produced the following structures: At each resolution $M$ we are given
an OS triple $({\cal H}_M,\Omega_M, H_M)$ and in addition we have the 
$C^\ast-$algebra $\mathfrak{A}_M$ of bounded operators acting on ${\cal H}_M$. 
The $w(\phi_M(l_M))$ form a commutative subalgebra $\mathfrak{B}_M$ 
of $\mathfrak{A}_M$ for 
which $\Omega_M$ is cyclic and separating. The 
full algebra $\mathfrak{A}_M$ contains in addition 
operators that depend on the momenta $\pi_M$ conjugate to $\phi_M$. 
Furthermore, the renormalisation scheme has naturally produced an
inductive system of isometric 
injections $J_{MM'}:\; {\cal H}_M\to {\cal H}_{M'}$ for 
$M<M'$ such that $J_{MM'} w(\phi_M(l_M))\Omega_M=
w(\phi_{M'}(I_{MM'} l_M))\Omega_{M'}$ and such that 
$J_{M' M^{\prime\prime}} J_{MM'}=J_{M M^{\prime\prime}}$ for 
$M<M'<M^{\prime\prime}$. Isometry $J_{MM'}^\dagger J_{MM'}=1_{{\cal H}_M}$ 
implies that $P_{MM'}=J_{MM'} J_{MM'}^\dagger$ is an orthogonal projection 
in ${\cal H}_{M'}$.\\
\\
In section \ref{s5} we established that the Hamiltonian renormalisation derived 
from OS reconstruction, under the assumption of vacuum uniqueness at each 
resolution, has given rise to a flow of {\it contraction semigroups} 
\be \label{7.0a}
e^{-\beta H^{(n+1)}_M}=[J^{(n)}_{MM'}]^\dagger\;
\{\lim_{N\to \infty} \; [e^{-\beta H^{(n)}_{M'}} P^{(n)}_{MM'}]^N\}
J^{(n)}_{MM'}
\ee
using the master equation 
(\ref{5.19}) consequently. Suppose we actually use the weaker condition 
(\ref{5.31}) which can be stated as 
\be \label{7.0b}
e^{-\beta H^{(n+1)}_M}= [J^{(n)}_{MM'}]^\dagger\; e^{-\beta H^{(n)}_{M'}} \;
J^{(n)}_{MM'}
\ee
and which appeared as an intermediate step. This flow 
condition on {\it Gibbs operators}
cannot hold for all $\beta$ hence we fix it for the rest of this section. 
Note that $\Omega^{(n+1)}_M$
keeps being an eigenvector of the Gibbs operator with eigenvalue unity during 
the flow since $\Omega^{(n)}_{M'}=J_{MM'}^{(n)} \; \Omega^{(n+1)}_M$.\\ 
\\
Consider now the linear embedding of algebras 
\be \label{7.1}
\alpha_{MM'}:\;\mathfrak{A}_M \to \mathfrak{A}_{M'};\;\;
a_M\mapsto J_{MM'} \; a_M\; J_{MM'}^\dagger
\ee
sometimes called an {\it ascending superoperator}.
This is a $^\ast-$homomorphism 
\be \label{7.2}
[\alpha_{MM'}(a_M)]^\dagger=\alpha_{MM'}(a_M^\dagger),\;\;
\alpha_{MM'}(a_M\;b_M)=\alpha_{MM'}(a_M)\;\alpha_{MM'}(b_M)
\ee
by virtue of isometry of the injections. However, $\alpha_{MM'}$ is 
not unital 
\be \label{7.3}
\alpha_{MM'}(1_{{\cal H}_M})=P_{MM'}\not=1_{{\cal H}_{M'}}
\ee
Nevertheless it is consistently defined 
\be \label{7.3a}
\alpha_{M'M^{\prime\prime}}\circ \alpha_{MM'}=
\alpha_{M M^{\prime\prime}}
\ee
by virtue of the inductive properties of the injections. 

We may want to twist $\alpha_{MM'}$ by unitary operators $U_M,U_{M'}$ in 
${\cal H}_M,\;{\cal H}_{M'}$ respectively, sometimes called 
{\it disentanglers} 
to obtain
$\tilde{\alpha}_{MM'}(.)=U_{M'}\; \alpha_{MM'}\; U_M^\dagger$
which preserves all properties and amounts to the substitution of
cylindrically consistent 
isometries $J_{MM'}\to \tilde{J}_{MM'}=U_{M'}\; J_{MM'}\; U_M$
with no conditions on the $U_M$. For the injections into the continuum
this amounts to $\tilde{J}_M=U J_M U_M^\dagger$ for some unitary $U$ on 
$\cal H$. 

We may use these 
to define a density matrix (positive trace class operators of unit trace)
renormalisation scheme as follows: Given an initial 
system of density matrices $\rho^{(0)}_M$ for each resolution scale $M$
we set for $M<M'$ (in practice $M'=2M$) 
\be \label{7.4}
\rho^{(n+1)}_M:=
\frac{[J^{(n)}_{MM'}]^\dagger\; \rho^{(n)}_{M'}\; J^{(n)}_{MM'}}{
{\rm Tr}_{{\cal H}^{(n)}_M}[[J_{MM'}^{(n)}]^\dagger\; 
\rho^{(n)}_{M'}\; J^{(n)}_{MM'}]}
=\frac{[J_{MM'}^{(n)}]^\dagger\; \rho^{(n)}_{M'}\; J^{(n)}_{MM'}}{
{\rm Tr}_{{\cal H}_{M'}}[\rho^{(n)}_{M'}\; P^{(n)}_{MM'}]}
:=D_{MM'}[\rho^{(n)}_{M'}]
\ee
which is sometimes called a {\it descending (non-linear) superoperator}.
The lhs of (\ref{7.4}) may be considered as the {\it reduced density matrix}
of the density matrix on the rhs. 

A fixed point of 
(\ref{7.4}) defines a consistently defined system of density matrices
since for $M<M'<M^{\prime\prime}$
\ba \label{7.5}
\rho_M &=& 
=\frac{J_{MM'}^\dagger\; \rho_{M'}\; J_{MM'}}{
{\rm Tr}_{{\cal H}_{M'}}[\rho_{M'}\; P_{MM'}]}
\nonumber\\
&=& \frac{J_{MM'}^\dagger\;
\{ 
\frac{J_{M'M^{\prime\prime}}^\dagger\; \rho_{M^{\prime\prime}}\; 
J_{M'M^{\prime\prime}}}{
{\rm Tr}_{{\cal H}_{M^{\prime\prime}}}[\rho_{M^{\prime\prime}}\; 
P_{M'M^{\prime\prime}}]}
\}
\;J_{MM'}}{
{\rm Tr}_{{\cal H}_{M'}}[
\{
\frac{J_{M'M^{\prime\prime}}^\dagger\; \rho_{M^{\prime\prime}}\; 
J_{M'M^{\prime\prime}}}{
{\rm Tr}_{{\cal H}_{M^{\prime\prime}}}[\rho_{M^{\prime\prime}}\; 
P_{M'M^{\prime\prime}}]}
\}
\; P_{MM'}]}
\nonumber\\
&=&
\frac{[J_{M'M^{\prime\prime}}\; J_{MM'}]^\dagger\; \rho_{M^{\prime\prime}}\; 
[J_{M'M^{\prime\prime}}\; J_{MM'}]}{
{\rm Tr}_{{\cal H}_{M'}}[J_{M'M^{\prime\prime}}^\dagger\;\rho_{M^{\prime\prime}}\; 
J_{M'M^{\prime\prime}} \; P_{MM'}]}
\nonumber\\
&=&
\frac{J_{MM^{\prime\prime}}^\dagger\; \rho_{M^{\prime\prime}}\; 
J_{MM^{\prime\prime}}}{
{\rm Tr}_{{\cal H}_{M^{\prime\prime}}}[\rho_{M^{\prime\prime}}\; 
J_{M'M^{\prime\prime}} \; P_{MM'}\; J_{M'M^{\prime\prime}}^\dagger\;]}
\nonumber\\
&=&
\frac{J_{MM^{\prime\prime}}^\dagger\; \rho_{M^{\prime\prime}}\; 
J_{MM^{\prime\prime}}}{
{\rm Tr}_{{\cal H}_{M^{\prime\prime}}}[\rho_{M^{\prime\prime}}\;\; 
P_{MM^{\prime\prime}}]}
\ea
If one wants to have a linear descending superoperator one needs to extend  
the homomorphism to a unital one or work with trace class operators 
of non-unit norm because one can always normalise afterwards. This 
is possible because the {\it partition functions}
\be \label{7.0c}
Z_M(\beta)={\rm Tr}_{{\cal H}_M}[e^{-\beta H_M}]
\ee
enjoy the {\it invariance property} \cite{25}
\be \label{7.0d}
Z_M(\beta):={\rm Tr}_{{\cal H}_{M'}}[e^{-\beta H_{M'}}\; P_{MM'}]
\ee
for all $M<M'$ 

The relation (\ref{7.5}) says that there exists a density matrix $\rho$
on the inductive limit ${\cal H}$ such that
\be \label{7.6}
\rho_M=\frac{J_M^\dagger\; \rho\; J_M}{{\rm Tr}_{{\cal H}}[\rho\; P_M]}
\ee
where the isometries $J_M\; {\cal H}_M \to {\cal H}$ satisfy 
$J_{M'} J_{MM'}=J_M$ and $P_M=J_M J_M^\dagger$ is a projection. One may 
explicitly check that (\ref{7.6}) solves (\ref{7.5}). 

Density matrices explore the folium of the algebraic vacuum state 
$\omega_0(.)=<\Omega,.\;\Omega>_{{\cal H}}$, i.e. we obtain algebraic 
states on the continuum algebra $\mathfrak{A}$ and the one at finite resolution
respectively
\be \label{7.7}
\omega_\rho(.)={\rm Tr}_{{\cal H}}[\rho\;.],\;\;
\omega_{\rho_M}(.)={\rm Tr}_{{\cal H}_M}[\rho_M\;.],\;\;
\ee
for which we have the identity
\be \label{7.8}
\omega_\rho(\alpha_M(a_M))={\rm Tr}_{{\cal H}}[\rho P_M]\;
\omega_{\rho_M}(a_M)
\ee
with $\alpha_M(a_M)=J_M a_M J_M^\dagger$. 

Examples for unital homomorphisms are for example 
available if the Hilbert spaces 
${\cal H}_M$ have the following, non-trivial additional structure \cite{29c}: 
For all 
$M<M'$ we have ${\cal H}_{M'}={\cal H}_M\otimes {\cal H}_{M'|M}$ (or at least 
there is a unitary map between them) with additional consistency conditions 
between these factorisations, see \cite{29c} for details. Then for 
instance 
\be \label{7.9}
\alpha_{MM'}(a_M)=a_m\otimes 1_{{\cal H}_{M'|M}}
\ee
and density matrix renormalisation amounts to taking partial traces with 
respect to the factor ${\cal H}_{M'|M}$ \cite{25}. The map (\ref{7.9}) can 
be described 
in terms of isometric embeddings $J_{MM'}^I \psi=\psi\otimes b_I$ with 
$\psi\in {\cal H}_M,\; b_I$ an ONB of ${\cal H}_{M'|M}$ so that 
$\alpha_{MM'}(a_M)=\sum_I J_{MM'}^I a_M [J_{MM'}^I]^\dagger$ which explains 
why just one isometry is not sufficient to make $\alpha_{MM'}$ unital.
Such a tensor product situation is available if we consider as label set
the mutually commuting modes of a Fock representation. The latter can 
be obtained by semaring creation and annihilation operators with respect an 
orthonormal basis of the 1-particle Hilbert space and if the corresponding 
smearing is supposed to be smooth and spatially local as one would like 
for a real space coarse graining scheme then one can use a (e.g. Meyer) 
wavelet basis \cite{70}. However in view of Haag's theorem \cite{60a} 
it is unclear whether Fock space techniques can be applied in the interacting 
case, which is why the scheme adopted in this paper keeps track of the
interacting rather than free vacuum during the flow. Even for free theories 
note that in case one is dealing with an infinite number of degrees of freedom  
(continuum limit) a Fock representation must be carefully adapted to the 
Hamiltonian as otherwise it is not even densely defined.

For completeness we mention that a unitalisation of 
$\alpha_{MM'}(.)=J_{MM'}(.)J_{MM'}^\dagger$ can be provided if the 
system $(\mathfrak{A}_M,\alpha_{MM'})_{M<M'}$ admits a consistently defined 
sytem of {\it weights}, that is, 
unital
$^\ast$homomorphisms $w_M:\;\mathfrak{A}_M\to \mathbb{C}$
such that $w_{M'}\circ \alpha_{MM'}=w_M$. Then 
it is not difficult to check that 
\be \label{7.10}
\beta_{MM'}(a_M)=\alpha_{MM'}(a_M)+w_M(a_M)\; P_{MM'}^\perp
\ee
is unital, has all the properties of $\alpha_{MM'}$ and extends 
$\alpha_{MM'}$ in the sense that
$\beta_{MM'}(a_M)\Omega_{M'}=\alpha_{MM'}(a_M)\Omega_{M'}=J_{MM'}
a_M \Omega_M$.

Thus our Hamiltonian renormalisation scheme also suggests as a variant 
an independent density matrix renormalisation scheme corresponding to the 
Gibbs operator at fixed temperature that simultaneously keeps track
of the vacuum. 
Various combinations and generalisations of these ideas are conceivable.
For instance, one could use as density matrices the Gibbs 
operator corresponding to a Hamiltonian $\rho^{(n)}_M=e^{-\beta H^{(n)}_M}$ 
and any {\it fixed} inductive limit Hilbert space with inductive 
structure $J_{MM'}$ or a fixed consistent system of 
unital homomorphisms and then use the 
flow defined by the scheme 
considered in this section to derive a flow of Hamiltonians as an independent 
technique, see also \cite{25,29c,29d}. This will likely give rise to a
(finite) temperature 
state rather than a vacuum state as a fixed point of the renormalisation flow. 
Or one could keep the density matrices independent of the Hamiltonians 
and study the time evolution of the reduced density matrices 
$\rho^{(n+1)}_M(t):=\alpha'_{MM'}(\gamma^{(n)}_{M',t}(\rho^{(n)}_{M'}))$ 
where $(\alpha'_{MM'}\rho_{M'})[a_M]:=\rho_{M'}[\alpha_{MM'}(a_M)]$ (here 
the density matrices are understood as tracial states) 
with 
$\gamma^{(n)}_{M',t}=e^{it H^{(n)}_{M'}}(.)\;e^{-it H^{(n)}_{M'}}$ 
whose infinitesimal form gives rise to a Lindblad equation \cite{65}
in the unital case (with tivial non-linear corrections in the non-unital 
case that can be computed from the linear flow that one obtains 
by dropping the normalisation by the partition function).
Its non-dissipative component can be used to define a flow of Hamiltonians.
Yet more generally we could in fact consider a system of quadruples 
$(\omega_M,\mathfrak{A}_M,H_M,\alpha_{MM'})$ where the Gibbs states have been 
replaced by general algebraic states $\omega_M$
and we could study their analogous flow.

\section{Connections with the functional renormalisation group}
\label{s8}

The functional renormalisation group (FRG) 
(see the first two references in \cite{15}) is the technique underlying 
the asymptotic safety aproach to quantum gravity. We will sketch it here 
in measure theoretic terms to display the rather obvious point
of contact with the Hamiltonian renormalisation programme which arises from 
OS reconstruction. We will restrict to scalar field theories for 
simplicity. The discussion will rely on several assumptions which we spell 
out in some detail and will be only qualitative in nature.\\
\\
Suppose that $\mu$ is a probability measure with stochastic process $\Phi$ 
labelled by smearing functions $F$ and that 
$\mathbb{R}^+\ni k\mapsto \mu_k$ is a one parameter family of such probability 
measures defined in terms of its generating functional as 
\be \label{8.1} 
\mu_k(w[F]):=\frac{\mu(\rho_k \; w[F])}{\mu(\rho_k)};\; 
w[F](\Phi):=e^{i\int\; d^Dx\; F(x)\Phi(x)}
\ee
Here $\rho_k$ is a $\mu$ measurable, positive $L_1(d\mu)$ function with 
positive $\mu$ measure. This 
means that $d\mu_k=\frac{\rho_k}{\mu(\rho_k)} d\mu$, i.e. the measures 
$\mu_k$ are mutually absolutely continuous and $\rho_k$ is essentially 
the Radon-Nikodym derivative of $\mu_k$ with respect to $\mu$ \cite{58}.

The family of measures enjoy the obvious {\it cylindrical consistency} 
identity
\be \label{8.2}
\mu_k(w[F])=\frac{\mu_{k'}(\frac{\rho_k}{\rho_{k'}}\; w[F])}{
\mu_{k'}(\frac{\rho_k}{\rho_{k'}})}
\ee
Suppose now that in fact the continuum measure $\mu$ is not known but that 
we guess a family of measures $\mu_k$ and use (\ref{8.2}) to define 
a renormalisation group flow where the $\rho_k$ are fixed and supposed 
to be measurable functions for all the $\mu_k$ in the process of the flow.
This is, in most general terms, the idea of the FRG, that is, the rhs 
of (\ref{8.2}) is the renormalisation of $\mu_{k'}$ and serves as the 
new definition of $\mu_k$.  If we take 
the derivative of the rhs of (\ref{8.2}) with respect to $k$ at $k'=k$
we obtain what is known in the literature as the {\it Polchinsky equation}. If 
we take the Legendre transformation with respect to $F$ of 
\be \label{8.3}
i\varphi[F] -\ln(\mu(\rho_k w[F]))+\ln(\rho(k)[\varphi])
\ee 
obe obtains the {\it effective average action} $\Gamma_k[\varphi]$.
One finds a closed equation for $\Gamma_k$ for the case that 
$\ln(\rho_k)(\varphi)=-\frac{1}{2}<\varphi,R_k\varphi>$ 
is quadratic in $\varphi$ known as the 
{\it Wetterich equation}
\be \label{8.4}
\frac{\partial}{\partial k} \Gamma_k[\varphi]
=\frac{1}{2}\; {\rm Tr}([\partial_k R_k]\;[R_k+\Gamma^{(2)}_k[\varphi]]^{-1})
\ee
where $R_k=R_k(x,y)$ and the second functional derivative 
$\Gamma^{(2)}_k(x,y)$ are symmetric integral kernels. As (\ref{8.4}) is a 
consequence of having a consistent family of measures derived from $\mu$,
it can be used as a flow equation as well which is strictly equivalent to the 
Polchinsky equation.

To translate (\ref{8.2}) into Hamiltonian renormalisation we must be more 
specific on the choice of $\rho_k$ which we take to be of the form 
cited above in order that it yields the Wetterich equation together with
the assumptions usually made: We consider a translation invariant kernel
$R_k$ whose Fourier transform $\hat{R}_k(p)$ enjoys the following 
properties: 1. $\hat{R}_k(p)\ge 0$, 2. $\hat{R}_k(p)$ is monotonously growing 
in $k$ at fixed $p$ and 3. $\lim_{p\to \infty} \hat{R}_k(p)=0$. 
One usually adds 4. $\hat{R}_{k=0}(p)=0,\;
\hat{R}_k(p=0)>0,\;\lim_{k\to\Lambda} \hat{R}_k(p)=\infty$ (where $\Lambda$ 
is a cut-off which for a a continuum theory can be taken to $\infty$).
We will only need 1.-3. in what follows. They imply that $R_k$ is an 
{\it infrared cut-off} i.e. momentum modes $p$ lower than $k$ 
(with respect to the Euclidian metric) are suppressed, high momentum modes 
are unaffected. Thus as we increase $k$ we suppress more and more modes, 
hence the number of degrees of freedom becomes smaller. 
Accordingly, the coarse graining here happens by {\it increasing} $k$ 
rather than lowering, i.e. the flow is towards the UV rather than towards 
the IR in contrast to real space block spin transformations from 
fine to coarse lattices.

Accordingly, we update a measure at $k$ from a measure at $k'<k$. 
Let us further assume that the measure $\mu_{k'}(.)$ and 
$\mu_{k'}(\rho_k/\rho_{k'}\;.)/\mu_{k'}(\rho_k/\rho_{k'})$
are in fact Wiener mesures with respect to the same stochastic process $\phi$
and Hilbert space ${\cal H}_{k'}$ but with underlying Hamiltonians and vacua 
given by $H_{k'},\Omega_{k'}$ and $H'_k:=H_{k'}+h_k-h_{k'},\;\Omega_{kk'}$ 
respectively. This puts some restriction on $R_k$ as far as the dependence 
on the time 
component of the momentum is concerned if we wish $h_k$ to be an at 
most quadartic expression in the momenta $\pi$ conjugate to $\phi$. 
Thus we consider e.g. $\hat{R}_k(p)=[(p^0)^2\kappa_k(||\vec{p}||)
+\sigma_k(||\vec{p}||)]$ where $\kappa_k,\sigma_k$ enjoy the properties 
1.-3. above but only with respect to the spatial momenta and 
$\lambda>0$. Similar restrictions
have been made in \cite{15a} which is concerned with making contact 
to Lorentzian signature GR. This way temporal momentum is not suppressed 
which is in accordance with our Hamiltonian renormalisation scheme which 
does not renormalise in the time direction.

Then the rhs of (\ref{8.2}) can be written for 
$F(t,\vec{x})=\sum_{k=1}^T\;\delta(t,t_k)f_k(\vec{x})$
\be \label{8.5}
<\Omega_{kk'},e^{i\phi[f_1]}\;e^{-[t_T-t_{T-1}]H'_k}\;..\;
e^{-[t_2-t_1]H'_k}\;e^{i\phi[f_1]}\;\Omega_{kk'}>
\ee
To get rid of the dependence on $k$ in $\Omega_{kk'}$ we use as usually that 
\be \label{8.5a}
\Omega_{kk'}=\lim_{\beta\to \infty} 
\frac{e^{-\beta H'_k} \Omega_{k'}}{||e^{-\beta H'_k} \Omega_{k'}||}
\ee
assuming that the strong limit exists and thus the rhs of (\ref{8.2})
can be written
\be \label{8.6}
\lim_{\beta\to \infty}
\frac{
<\Omega_{k'},e^{-[\beta-t_T] H'_k}\;e^{i\phi[f_1]}\;
e^{-[t_T-t_{T-1}]H'_k}\;..\;
e^{-[t_2-t_1]H'_k}\;e^{i\phi[f_1]}\;e^{-[t_1+\beta] H'_k}\;\Omega_{k'}>
}
{
<\Omega_{k'},e^{-2\beta\; H'_k}\;\Omega_{k'}>
}
\ee
Using the Trotter product formula we can replace every exponential of the 
form $e^{-sH'_k}$ by $[e^{-s H_{k'}/N}\; e^{-s[h_k-h_{k'}]/N}]^N$ as $N$ 
turns to $\infty$. By construction $h_k-h_{k'}$ is non-vanishing for 
momentum modes $k'\le ||\vec{p}||\le k$ and thus its Gibbs exponential 
suppresses those. As emphasised in the first two references of \cite{15} one 
obtains an actual coarse graining renormalisation 
group only when $R_k(p)$ takes the 
more special form $Z(\vec{p})[(p_0)^2+\vec{p}^2]I_k(\vec{p})^2/[1-I_k(\vec{p})^2]$
where i. 
$I_k(\vec{p}) I_{k'}(\vec{p})=I_{\tilde{k}(k,k')}(\vec{p})$ for 
some $\tilde{k}(k,k')$ 
which ensures the group property and moreover ii. $I_k(\vec{p})$ approaches 
a step function $\theta(k^2-\vec{p}^2)$.
In this limit $e^{-s (h_k-h_{k'})/N}\to P_{kk'}$ becomes a projection 
operator which projects on a subspace ${\cal H}_k=P_{kk'} {\cal H}_{k'}$ 
independent of the excess modes $k' <||\vec{p}||<k$. To see this note that 
the form of $R_k-R_{k'}$ allows us to write $h_k-h_{k'}$ in terms of 
suitable annihilation and creation operators on a suitable Fock space
which we assume to be unitarily related to ${\cal H}_{k'}$ (this will be 
the case in presence of both IR and UV cut-offs). On that Fock space 
in the step function limit the above {\it is} a projector 
onto the subspace with no excitations from the excess modes and a unitary 
transformation does not change the projection  property. Accordingly the Trotter
product becomes a degenerate one and $e^{-sH'_k}$ becomes 
$P_{kk'} e^{-s P_{kk'} H_{k'} P_{kk'}}$ see appendix \ref{sd}. Taking $\beta\to\infty$ 
in (\ref{8.6}) the rhs of (\ref{8.2}) can be written
\be \label{8.7}
<\Omega_k,e^{i\phi[f_1]}\;e^{-[t_T-t_{T-1}]H_k}\;..\;
e^{-[t_2-t_1]H_k}\;e^{i\phi[f_1]}\;\Omega_k>
\ee
with the vacuum $\Omega_k$ of $H_k$
\be \label{8.7a}
H_k=P_{kk'} H_{k'} P_{kk'},\;
\Omega_k=\lim_{\beta\to \infty} 
\frac{e^{-\beta H_k} P_{kk'}\Omega_{k'}}{||e^{-\beta H_k} P_{kk'}\Omega_{k'}||}
\ee
If we now compare to the end of section 
\ref{s5} we find qualitatively an {\it exact match}, just that the projections
now involve momentum modes with respect to a Fock representation determined 
by $R_k$.

\section{Conclusion}
\label{s6}

In this contribution we have reviewed, extended and clarified the proposal
\cite{13}. The extension consisted in i. an improved derivation of the 
renormalisation scheme (\ref{5.28}), (\ref{5.32}) from OS reconstruction 
using an extended minimal set of OS axioms that also includes the uniqueness 
of the vacuum (which is in fact always assumed in QFT on Minkowski space)
as well as ii. a much more systematic approach to the choice of coarse graining 
maps for a general QFT which are motivated by structures naturally provided 
already by the classical theory. The clarification consisted in separating off
the null space quotienting process imposed by OS reconstruction as an 
independent part of the renormalisation flow whose formulation naturally 
uses the language of stochastic processes.   

We also had the opportunity to make several points of contact with other 
renormalisation programmes that are currently being further developed. 
For instance, the reduced density matrix approach 
on which entanglement renormalisation 
schmemes rest occurs naturally in our scheme as well when looking at the 
flow of the vacuum and Hilbert space. Next, since we consider a real space 
renormalisation scheme, when translated in terms of the flow 
of Wiener measures that we obtain from the flow of OS data we are rather 
close to the asymptotic safety programme because our spatial lattices can
of course be translated into momentum lattices by Fourier transformation
that are used in the asymptotically safe quantum gravity programme.      
Finally, our proposal is obviously very close in language and methods 
to all other Hamiltonian renormalisation schemes and while we currently 
focus on a kinematical coarse graining scheme our approach also contains 
dynamical components such as the flow of the vacuuum.

In \cite{13} and \cite{68} we have successfully applied our scheme 
to free QFT (scalar fields and Abelian gauge theories) exploiting their 
linear structure. Obviously 
one should construct further solvable examples of interacting theories, 
e.g. interacting 2D scalar QFT \cite{60} or free Abelian gauge theories but 
artificially discretised in terms of non-linear holonomies in order to 
simulate the situation in Loop Quantum Gravity, see \cite{68} for 
further remarks. 

Of course the ultimate goal is to use Hamiltonian renormalisation to find a 
continuum theory for canonical Quantum Gravity. Here we can use the LQG 
candidate as a starting point because it is rather far developed, but of course
the flow scheme developed can be applied to any other canonical programme.
However, using LQG and  the concrete scheme that employs a fixed subset of 
graphs $\gamma_M$ labelled by 
$M\in  \mathbb{M}$ of cubical topology is,  
at each resolution $M$, precisely the AQG version of LQG \cite{55}.  
Hence we can already speculate on what can be expected from the 
renormalisation flow:\\
\\
The Hamiltonian $H^{(0)}_M$ defined on the corresponding 
${\cal H}^{(0)}_M=L_2(SU(2)^{3 M^3},\;d^{3 M^3}\mu_H)$ ($\mu_H$ being 
the $SU(2)$ Haar measure) 
could be, but not needs to be, ordered in such a way as to annihilate 
the vacuum $\Omega^{(0)}_M=1$ of a discretised volume operator $V^{(0)}_M$ 
as it is standard in current regularisations of the Hamiltonian constraint.
In fact it may be desirable to choose the vacuum of $H^{(0)}_M$ not to 
coincide with that of $V^{(0)}_M$ in order to imprint its  
algebraic structure. 
The operator $H^{(0)}_M$ preserves ${\cal H}^{(0)}_M$ but not each subspace 
defined by sublattices of $\gamma_M$ and is thus not super local in contrast to  
the definition \cite{5}, for instance it will use volume operators local to 
a vertex and holonomies along plaquettes incident at that vertex (next neighbour
interaction). When running the renormalisation scheme, next to next neighbour 
interactions will be switched on (this is exactly what happens in the examples 
\cite{13}, \cite{68}) etc. and upon reaching the fixed point the Hamiltonian 
$H_M$ will involve all possible interactions {\it with precise coefficients}
and thus be spatially non local but hopefully quasi local 
(i.e. the interactions die off exponentially 
with the distance between vertices defined by the 3D taxi driver metric on 
the graph (each edge counting one unit)). Note that this quasi locality 
at finite resolution can be straightforwardly computed in the examples 
\cite{13}, \cite{68} by using the spatially
local continuum Hamiltonian and projecting it with $J_M, J_M^\dagger$
(blocking from the continuum) and is thus {\it physically correct}. In other 
words spatial locality in the continuuum is not in conflict with spatial
non locality at finite resolution. In fact we even {\it expect} a high degree of 
spatial
non-locality for very small $M$ for which the naive dequantisation of 
$H^{(0)}_M$ at any phase space point $p$ 
will be far off the classical value $H(p)$ which matches with the remarks 
made at the end of section \ref{s3}.\\
\\
Several questions arise from this picture should the flow display any 
fixed points:\\ 
First, for compact $\sigma$ and if indeed we use a countable set 
of lattices $\gamma_M$ as above, 
the resulting inductive limit Hilbert space could be separable (since 
there is a countable basis defined by vectors at finite resolution), thus  
would not be the standard LQG  representation space 
${\cal H}_{{\rm LQG}}=L_2(\overline{{\cal A}},d\mu_{{\rm AL}})$ of square 
integrable functions with respect to the Ashtekar-Lewandowski measure 
$\mu_{AL}$ on a space $\overline{{\cal A}}$ of distributional 
connections\footnote{In the non-compact case one may need to take the infinite 
tensor product
extension \cite{57} which is also non-separable but in a different sense and 
there one regains separability by passing to irreducible representations 
of the observable algebra.}.
In view of the uniqueness theorem \cite{4} one of its assumptions will
then be violated. The most likely possibility is that the corresponding
vacuum expectation value functional is not spatially diffeomorphism 
invariant since the diffeomorphism symmetry was explicitly 
broken in the renormalisation process. If the continuum 
Hamiltonian is still spatially diffeomorphism invariant we would be 
in the situation of spontaneous symmetry breakdown and could 
view this as a phase transition from the symmetric ${\cal H}_{{\rm LQG}}$
phase to this broken phase. Note that in our  
gauge fixed situation the diffeomorphism group is considered as a 
continuous symmetry group and not as a gauge group.\\   
Next, precisely due to this separability the resulting 
theory may not suffer from the discontinuity of holonomy operators which 
otherwise gives rise to what has been called the ``staircase problem''
in the literature \cite{69}: The cubical graphs $\gamma_M$ contain paths 
only along the coordinate axes. Since all $M\in \cal{M}$ are allowed,
these paths separate the points of the classical configuration space 
but not of the distributional space $\overline{{\cal A}}$. In particular,
any path that is not a ``staircase'' path cannot be accomodated at any finite 
resolution. Yet, the continuum Hamiltonian in the example \cite{68} 
does not care about the fact that it was defined as a fixed point of a flow 
of its finite resolution projections of cubical lattices only, it 
also knows how to act on states which are excited on non-``staircase''
paths. The reason for why this happens is as follows: Consider any path $c$
and some staircase approximant $\tilde{c}$ with the same end points as $c$
which has zero winding number with respect to $c$ so that 
$c\circ\tilde{c}^{-1}=\partial S$ bounds a surface. Then for an Abelian 
connection $A$ we have 
$\int_c A-\int_{\tilde{c}} A=\int_S dA$ and in the classical  
theory the surface integral converges to zero. In the quantum theory a 
similar calculation can be made because the Hilbert space measure is
supported on a different kind of distributional connections than 
$\overline{{\cal A}}$.\\ 
Finally, although the scheme strictly 
speaking was derived for theories with gauge fixed spacetime diffeomorphism 
constraints and a true physical Hamiltonian bounded from below, 
we may of course ``abuse'' it 
and also consider constraint operators $C(f)$ as Hamiltonians, define their 
finite resolution expressions $C(f)^{(0)}_M$ and let them flow (here $f$ is a 
test function on the spatial manifold $\sigma$)\footnote{In fact, the physical 
Hamiltonian of section \ref{s3} is not manifestly 
bounded from below, hence we to abuse the formalism in the sense that 
we assumed the semi-boundedness.}. This will involve as a new
ingredient also a discretisation of the smearing function $f$ which 
could be done using the maps $I_M, K_M$ for scalar fields, see \cite{68}. 
Suppose then 
that for all $f$ fixed point families $\{C(f)_M\}_{M\in \mathbb{N}_0}$ 
can be obtained. Should we expect that
the $C(f)_M$ represent a finite resolution version 
of the classical continuum constraint (hypersurface deformation) 
algebra $\{C(f),C(g)\}=C(h(f,g))$ where 
$h(f,g)$ is another (in general phase space dependent) smearing function?
The answer is in the negative! Namely, what we want is that the continuum 
operators obey $[C(f),C(g)]=i C(h(f,g))$ (with appropriate orderings of 
$C, h(f,g)$ in place). But if $C(f)_M=J_M^\dagger C(f) J_M$ then 
\be \label{6.1}
[C(f)_M, C(g)_M]=
J_M^\dagger[C(f) P_M C(g)-C(g) P_M C(f)]\; J_M
=C(h(f,g))_M+ J_M^\dagger[C(f) [P_M, C(g)]\;J_M
\ee
Thus, even if the continuum algebra closes, one {\it does not see this at any 
finite resolution} unless $[C(f),P_M]=0$ for all $f,M$. This will 
generically not hold because not even $C(f)_{M'}$ preserves 
$J_{MM'} {\cal H}_M, \; M<M'$ unless 
$J_{MM'} C(f)_M=C(f)_{M'} J_{MM'}$ i.e. the 
$C(f)$ form an inductive family which is not expected. Of course the correction
term in (\ref{6.1}) is expected to become ``small'' in the limit $M\to \infty$ 
in which 
$P_M\to 1_{{\cal H}}$ and thus an appropriate criterion for closure of the 
continuum algebra using only finite resolution projections can be formulated 
(see \cite{13} for the simpler case of rotational invariance).   
Note that the quantisation performed for spatially
diffeomorphism invariant Hamiltonian operators 
on the Hilbert space ${\cal H}_{{\rm LQG}}$ displayed in section \ref{s3} 
was {\it forced} to have the unphysical property $[H,P_M]=0$, see the 
statement just before 
(\ref{3.22}). But the underlying theorem exploits in a crucial way the 
non-separability of ${\cal H}_{{\rm LQG}}$ and thus fortunately 
does not hold on separable Hilbert spaces.\\ 
\\
Before closing, note that even if this approach of taking the UV limit  
can be completed and unless the manifold 
$\sigma$ is compact, we still must take the thermodynamic or infrared limit 
and remove the IR cut-off $R$ (compactification scale). As is well known 
from statistical quantum field theory \cite{51} interesting phenomena 
related to phase transitions can happen here. Moreover, constructible 
eamples of low dimensional 
interacting QFT show that the thermodynamic limit requires techniques 
that go beyond what was displayeed here \cite{60}.
However, we consider this 
momentarily as a ``higher order'' problem and reserve it for future research.\\
\\
{\bf Acknowledgements}\\
\\
The author thanks Thorsten Lang for in depth discussions about reduced density 
matrices, decoherence and the Lindblad equation in the context of 
renormalisation, Klaus Liegener for clarifying conversations about 
renormalisation of constraints and Alexander Stottmeister for very fruitful 
exchanges about renormalisation in terms of algebraic states.

\begin{appendix}

\section{Label structures for gauge field theories}
\label{sa}

We consider first the case of a $U(1)$ gauge theory so that $L_M$ is 
the space of functions on $\mathbb{Z}_M^D$ 
with values in $\mathbb{R}^D$.

Denote by $a=1,..,D$ a direction on the lattice and by $m\in \mathbb{Z}_M^D$ 
a vertex in it. Let (recall $\epsilon_M=R/M$)
\ba \label{a1}
[f^M_{a,m}]^b(x) &:=& \delta_a^b\; [\prod_{c\not=a} \delta(x^c,m^c \epsilon_M)]
\int_{[m^a\epsilon_M, (m^a+1)\epsilon_M)}\; dy\; \delta(x^a,y)
\nonumber\\
&=& \delta_a^b\; [\prod_{c\not=a} \delta(x^c,m^c \epsilon_M)]\; 
\chi_{[m^a \epsilon_M,(m^a+1)\epsilon_M)}(x^a)
\nonumber\\
{[}F_M^{a,m}]_b(x) &:=& \delta_b^a\; \delta(x^a,m^a \epsilon_M)\;
[\prod_{c\not=a}
\int_{[m^c\epsilon_M, (m^c+1)\epsilon_M)}\; dy^c\; \delta(x^c,y^c)]
\nonumber\\
&=&
\delta_b^a\; \delta(x^a,m^a \epsilon_M)\;
[\prod_{c\not=a}
\chi_{[m^c\epsilon_M, (m^c+1)\epsilon_M)}(x^c)]
\ea
Note the clopen structure of these distributions.
We have for $M<M'\; \Leftrightarrow \; n:=\frac{M'}{M}\in \mathbb{N}$ 
\be \label{a2}
<f^M_{b,m},F_{M'}^{a,m'}>:=
\int \; d^3x\; [F_{M'}^{a,m'}]_c(x) \; [f^M_{b,m}]^c(x)
=\delta^a_b\; 
\chi_{[m^a\epsilon_M, (m^a+1)\epsilon_M)}(m^{\prime a}\epsilon_{M'})
\prod_{c\not=a}
\chi_{[m^{\prime c}\epsilon_{M'}, 
(m^{\prime c}+1)\epsilon_{M'})}(m^c\epsilon_M)
\ee
which is non vanishing and then equals unity iff $a=b$ and 
\be \label{a3}
n m^a\le m^{\prime a} < n(m^a+1),\; 
m^{\prime b}\le n m^b  < m^{\prime b}+1; b\not=a
\ee 
This implies $m=m'$ for $M=M'$ ($n=1$).

Given $l_M\in L_M$ we obtain maps (we use the notation $l_M(m,a):=[l_M]^a(m)$)
\ba \label{a3}
&& I_M: L_M \to L;\; (I_M l_M)^a(x):=\sum_{m\in \mathbb{Z}_M^D,b\in \{1,,.,D\}}\; 
l_M(m,b)\; f_{m,b}^a(x)
\nonumber\\
&& K_M: L_M \to L^\ast;\; (K_M l_M)_a(x)
:=\sum_{m\in \mathbb{Z}_M^D,b\in\{1,..,D\}}\; 
l_M(m,b)\; F^{m,b}_a(x)
\ea
For $M<M'$ (i.e. $M'/M\in \mathbb{N}$)  
we have for any $l'_{M'}\in L_{M'},\; l_M\in L_M$
\ba \label{a4}
&& <l'_{M'},I_{M M'} l_M>_{L_{M'}}:=<K_{M'} l'_{M'}, I_M l_M>_{K\times I}
=\sum_{m'\in \mathbb{Z}_{M'}^D,a}\; l'_{M'}(m',a)\; (I_{MM'} l_M)(m',a)
\nonumber\\
&=& \sum_{m',m,a}\; l_{M'}(m',a)\; l_M(m,a)\;
\chi_{[n m^a, n (m^a+1))}(m^{\prime a})\; 
\prod_{b\not=b} \; \chi_{[m^{\prime b}, m^{\prime a})}(n m^b) 
\nonumber\\
&=& \sum_{m',a}\; l_{M'}(m', a)\; 
[\prod_{b\not=a}
\chi_{\mathbb{Z}_M}(m^{\prime b}/n)] \; 
l_M(m^a=[m^{\prime a}/n],m^b=m^{\prime b}/n;\;b\not=a)
\ea
from which we read off
\be \label{a5}
(I_{M M'} l_M)(m',a)= 
[\prod_{b\not=a}
\chi_{\mathbb{Z}_M}(m^{\prime b}/n)] \; 
l_M(m^a=[m^{\prime a}/n],m^b=m^{\prime b}/n;\;b\not=a)
\ee
where $[.]$ denotes the Gauss bracket and $n=M'/M$. It follows for 
$M<M'<M^{\prime\prime}$ with $n=M'/M,\; n'=M^{\prime\prime}/M',\; 
n^{\prime\prime}=M^{\prime\prime}/M=n n'$
\ba \label{a6}
&& (I_{M' M^{\prime}} I_{MM'} l_M)(a,m^{\prime\prime})
=[\prod_{b\not=a} \chi_{\mathbb{Z}_{M'}}(m^{\prime\prime b}/n')]
(I_{M M'} l_M)([m^{\prime\prime a}/n'], m^{\prime\prime b}/n';b\not=a)
\nonumber\\
&=& 
=[\prod_{b\not=a} \chi_{\mathbb{Z}_{M'}}(m^{\prime\prime b}/n')]\;
[\prod_{b\not=a} \chi_{\mathbb{Z}_M}((m^{\prime\prime b}/n')/n)]\;
l_M)([[m^{\prime\prime a}/n']/n], (m^{\prime\prime b}/n')/n;b\not=a)
\ea
which is non vanishing iff for all $b\not=a$ we have 
$m^{\prime\prime b}/n'\in \mathbb{Z}_{M'}$ and 
$m^{\prime\prime b}/(n' n)\in \mathbb{Z}_M$. Since the latter condition 
implies the first, the first factor in the last line of (\ref{a6}) can be 
dropped. Next let $m^{\prime\prime a}=n' m^{\prime a}+k'$ with 
$m^{\prime a}\in \mathbb{Z}_{M'},\; k'=0,..,n'-1$ and   
$m^{\prime a}=n m^a+k$ with 
$m^a\in \mathbb{Z}_M,\; k=0,..,n-1$. Then $m^{\prime\prime a}=n' n m^a+
(n' k+k')\le n' n m^a+n' n-1$ whence  
\be \label{a7}
[[m^{\prime\prime}/n']/n]=[m^{\prime a}/n]=m^a=[m^{\prime\prime a}/
n^{\prime\prime}]
\ee
so that indeed 
\be \label{a8}
I_{M' M^{\prime\prime}} I_{M M'}=I_{M M^{\prime\prime}}
\ee
~\\
We note that $I_{MM'}$ restricts to a map $I_{MM'}\; B_M \to B_{M'}$ where 
$B_M$ is the set of {\it bit valued maps} $\mathbb{Z}_M^D
\to \{0,1\}^D$. 
The idea is that the bit valued function $l_M$ captures the information 
which edges of the lattice are {\it excited}, i.e. $l_M(m,a)=1$ means that 
the edge with starting point $m\epsilon_M$ that extends one lattice unit  
into direction $a$ is excited, otherwise not. The reason why this information
is sufficient is that the information about the strength of the excitation 
sits in the group representation label of the corresponding holonomy. 
For Abelian groups we could take group representations labelled by 
$\mathbb{Z}$ or $\mathbb{R}$ and thus absorb the representation label into 
$l_M$ which then becomes $\mathbb{Z}^D$ or $\mathbb{R}^D$ valued and which 
has the structure of a module over the ring $\mathbb{Z}$ or a vector space 
over the field $\mathbb{R}$. That absorption of the representation label
into $l_M$ is no longer possible for non-Abelian gauge groups. Thus, in order 
to develop a unified language for all physical theories we use the 
bit valued space $B_M$ in the main text. 
Note that in the bit valued case, $B_M$ is formally a vector 
space over the field $F_2^D=\{0,1\}^D$. However, these module or vector 
space structures are irrelevant for the purposes of renormalisation. 

\section{Cyclic Abelian $C^\ast-$algebras}
\label{sb}

Recall that a (unital) $C^\ast$ algebra $\mathfrak{A}$ 
is a $^\ast-$algebra with a norm topology with respect to which it is 
complete (i.e. $\mathfrak{A}$ is a (unital) Banach algebra) such that 
$||a^\ast a||=||a||^2$ for all $a\in \mathfrak{A}$. One of the best known 
examples is the algebra ${\cal B}({\cal H})$ of bounded operators on a Hilbert 
space $\cal H$. Given a Hilbert space $\cal H$ 
and some sub-C$^\ast$algebra $\mathfrak{B}\subset {\cal B}({\cal H})$   
a vector $\Omega\in {\cal H}$ is called cyclic for $\mathfrak{B}$ if 
$\mathfrak{B}\Omega$ is dense in ${\cal H}$. It is called separating if 
$b\Omega=0$ implies $b=0$ for all $b\in \mathfrak{B}$. 

In the OS reconstruction algorithm we construct ${\cal H},\Omega$ in terms 
of equivalence classes of the vector space $V$ of 
vectors which are linear combinations of vectors
of the form $\Psi[\Phi]=w_T(\Phi(t_T,l_T))..w(\Phi(t_1,l_1)$ where $\Phi$ 
is a stochastic 
process labelled by $\mathbb{R}\times L$ for the measure $\mu$ with 
$t_T> t_{T-1}>..>t_1$ and 
the reflection positive inner product is 
$<\psi,\psi'>=\mu(\psi^\ast R\psi')$ where $R$ denotes the time reflection.
Here the $w_k$ belong to some space of bounded functions ${\cal W}$ 
on the range of the $\Phi(t,l), l\in L$ which form 
an Abelian $^\ast$-algebra with respect to point wise operations.
The equivalence class $[\psi]=\psi+{\cal N}$ is with respect to the null 
space $\cal N$ of $<.,.>$, 
${\cal H}$ is the completion of $V/{\cal N}$ and $\Omega=[1]$. Consider 
the operators $[Q(w,l)],w\in {\cal W},\; l\in L$ defined by 
$[Q(w,l)][\Psi]:=[Q(w,l) \Psi],\; (Q(w,l)\Psi)[\Phi]=w(\Phi(0,l))\Psi[\Phi]$.
It is important that $Q(w,l)$ only depends on the time zero field in order for 
this definition to be well defined (independent of the representative) since 
for $\Psi\in {\cal N}$
\ba \label{b.1}
&& ||Q(w,l)\Psi||^2=\mu((Q(w,l) \Psi)^\ast\; R\; Q(w,l)\; \Psi)   
=\mu((\Psi)^\ast\; R\; Q(W,l)^\ast\; Q(w,l)\; \Psi)   
\nonumber\\
&=&|<\Psi,Q(W,l)^\ast\; Q(w,l)\; \Psi)>|
\le ||\Psi||\; ||Q(W,l)^\ast\; Q(w,l)\; \Psi||=0
\ea
where $Q(w,l)^\ast R=R Q(w,l)^\ast$ and the Cauchy-Schwarz inequality was 
used. Consequently it is not clear that $\Omega$ is cyclic for the $Q(w,l)$
and in general this will not be the case. \\
\\
Nevertheless, in the main text, we quoted the following result:
\begin{Theorem} \label{thb.1} ~\\
Let $\cal H$ be a separable Hilbert space and $\Omega\in {\cal H}$ a unit 
vector. Then there exists an Abelian, unital sub-$C^\ast-$algebra 
$\mathfrak{B}\subset {\cal B}({\cal H})$ of bounded operators on $\cal H$
such that $\Omega$ is cyclic for $\mathfrak{B}$.
\end{Theorem}
We could not find a proof of this statement in the literature, hence we 
supply an elementary one below.\\
\begin{proof}
Since ${\cal H}$ is separable, it has a countable orthonormal basis 
$b_I,\; I\in \mathbb{N}_0$ and we may assume $b_0=\Omega$ using the 
Gram-Schmidt algorithm. The idea is to construct inductively a basis 
$e_I,\; I\in \mathbb{N}_0$ such that $e_0=b_0=\Omega$, 
$<\Omega, e_I> >0$ for all $I\ge 0$ and $<e_I,e_J>=\delta_{IJ}$ for
all $I,J>0$. Then the bounded operators given by the self-adjoint
projections $P_0=:1_{{\cal H}},\; P_I:=e_I\; <e_I,.>$
generate the searched for unital Abelian $C^\ast-$algebra since
$P_0 P_I=P_I P_0=P_I,\; P_I P_J=\delta_{IJ} P_J,\;I,J>0$ and $P_I^\ast=P_I,\;
I\ge 0$ and $P_I \Omega=<e_I,\Omega>\; e_I,\; I\ge 0$ thus the finite linear 
combinations of the $P_I$ applied to $\Omega$ have dense range.

We pick $e_1=(b_0+b_1)/\sqrt{2}$ so that $<e_1,\Omega>=<e_1,b_1>=
1/\sqrt{2}>0$ and 
assume that we have constructed $e_1,..,e_I$ such that 
i. $e_J,\; 1\le J\le I$ 
is spanned by $b_0,..,b_J$ with real coefficients,  ii.  
$<e_J,e_K>=\delta_{JK},\; 1\le J,K\le I$, iii. $<b_J,e_J>>0,\; 
1\le J\le I$ and iv. $<e_J,\Omega> >0,\ 1\le J\le I$. 
Consider the Ansatz
\be \label{b2}
e_{I+1}=a b_{I+1}+b b_I+c e_I
\ee
for real coefficients $a,b,c$ to be determined. 
Note that $e_{I+1}$ is spanned by
$b_0,..,b_{I+1}$ so that i. is satisfied. Clearly $<e_J,e_{I+1}>=0$ for 
$1\le J\le I-1$ because $<e_J,e_I>=0$ by construction and 
$<e_J,b_I>=<e_J,b_{I+1}>=0$ since $e_J$ is spanned by $b_0,..,b_J,\; J<I$ 
which are orthonormal. We have 
\be \label{b3}
<e_I,e_{I+1}>=b <e_I,b_I>+c=0 \; \Rightarrow c=- b<e_I,b_I>
\ee
and normalisation gives 
\be \label{b4}
||e_{I+1}||^2=a^2+b^2+c^2+2bc <b_I, e_I>=a^2+b^2-c^2=a^2+b^2(1-<e_I,b_I>^2)=1
\ee
and a convenient choice is 
\be \label{b5}
a=-b=[2-<e_I,b_I>^2]^{-1/2}
\ee
so that orthonormalisation ii. is satisfied. Note that (\ref{b5}) is finite
as $0<|<e_I,b_I>|\le 1$. Since $a>0$ we have $<b_{I+1},e_{I+1}>>0$ 
satisfying iii. and since $b<0$ we have $<\Omega,e_{I+1}>=
c<\Omega,e_I>=-b <\Omega,e_I>\; <b_I,e_I>>0$ satisfying iv.

With this choice we may in fact solve the recursion since 
\be \label{b6}
c_{I+1}:=<b_{I+1},e_{I+1}>^2=a^2=(2-[<b_I,e_I>]^2)^{-1}=(2-c_I)^{-1}
\ee
which with $c_1=\frac{1}{2}$ is solved by $c_I=\frac{I}{I+1}$. Thus the 
explicit solution is $e_0=\Omega, e_1=(\Omega+b_1)/\sqrt{2}$ and for $I\ge 1$ 
\be \label{b7}
e_{I+1}=(b_{I+1}-b_I+\sqrt{\frac{I}{I+1}} e_I)\; \sqrt{\frac{I+1}{I+2}}
\ee
It follows for $I\ge 1$  
\be \label{b8}
<\Omega,e_{I+1}>=\sqrt{\frac{I}{I+2}} <\Omega,e_I>
=\sqrt{2\frac{I!}{(I+2)!}} <\Omega,e_1>=
\sqrt{\frac{I!}{(I+2)!}}
=\frac{1}{\sqrt{(I+1)(I+2)}}
\ee
which extends also to $I=0$. 

To see that the span of the $e_I$ is dense suppose that 
$<\psi,e_I>=0$ for all $I\ge 0$. This means $<\psi,b_0>=0$ and thus because 
$e_I$ is in the span of the $b_0,..,b_I$ succesively $<\psi,b_I>=0$ for all
$b_I,\; I\ge 0$ whence $\psi=0$ as the $b_I$ provide an orthonormal basis.
It is instructive to see that even the $e_I,\; I>0$ already lie dense which 
is demonstrated by showing that Bessel's inequality is saturated for 
$\Omega$. Indeed
\ba \label{b9}
&&||\Omega-\sum_{I>0} <\Omega,e_I> e_I||^2
=||\Omega||^2-\sum_{I>0} |<\Omega, e_I>|^2     
=1-\lim_{N\to \infty} \sum_{I=1}^N \frac{1}{I(I+1)}  
\nonumber\\
&=& 1-\lim_{N\to \infty} \sum_{I=1}^N [\frac{1}{I}-\frac{1}{I+1}]
=1-\lim_{N\to \infty} [1-\frac{1}{N+1}]=0
\ea
\end{proof}
In the non-separable case a non-constructive argument can be provided.
\begin{Theorem} \label{thb.2} ~\\
Suppose that ${\cal H}$ is not separable but that its dimension has 
the cardinality of $\mathbb{R}$. Let $\Omega\in {\cal H},\;||\Omega||=1$.
Then there exists an Abelian, unital sub-$C^\ast-$algebra 
$\mathfrak{B}\subset {\cal B}({\cal H})$ of bounded operators on $\cal H$
such that $\Omega$ is cyclic for $\mathfrak{B}$.
\end{Theorem}
\begin{proof}
Let $b_I$ be an orthonormal basis of ${\cal H}$ with $b_0=\Omega$ (this 
needs the axiom of choice) and 
consider the Abelian one-parameter group of unitary operators 
densely defined by 
$U_s b_I=b_{I+s},\; s\in \mathbb{R}$. Since $U_s^\ast=U_{-s}$ these generate 
an unital, Abelian $C^\ast-$algebra $\mathfrak{B}$ of bounded operators and 
obviously $\Omega$ is cyclic for $\mathfrak{B}$ as $U_I\Omega =b_I$.
\end{proof}
In the case of LQG we may label the basis $b_I$ by spin net work functions.
These have discrete labels (adjacency matrix (which edge connects which 
vertices), irreducuible representations on the edges and 
intertwiners at the vertices) and continuous labels (germs of piecewise 
analytic edges). These data have the cardinality of finite unions of 
$\mathbb{Z}$ and $\mathbb{R}$ respectively and thus have the cardinality 
of $\mathbb{R}$ again \cite{64} (even the space of real valued sequences 
$\mathbb{R}^{\mathbb{N}}$ has the 
same cardinality as $\mathbb{R}$), hence the above construction can be applied 
to LQG in principle.

In the separable case we may alternatively 
consider the index set $\mathbb{Z}$ and the 
unitary one parameter group $U_s b_I=b_{I+s},\; s\in \mathbb{Z}$ to apply 
the same construction. As elucidated by the proofs, the algebra $\mathfrak{B}$ 
while always existent under the assumptions made has in general not a very 
direct relation to the operators that have an immediate physical interpretation.

\section{Brief Account of Renormalisation Group Terminology}
\label{sc}

The purpose of this brief appendix is to relate notions of the path
integral renormalisation group usually formulated in terms of actions 
\cite{26} now in terms of measures.\\
\\
Recall from the main text that condition (\ref{5.16}) grants the existence 
of $\mu$ under rather generic conditions and that the strategy 
is therefore to construct an iterative 
sequence of measure families $\mathbb{N}_0\ni n\mapsto 
\{\mu^{(n)}_M\}_{M\in {\cal M}}$ {\it called renormalisation (group) flow}
with initial family as above such that  
either i. it converges or ii. a subsequence of it converges or iii. it has 
a fixed point. The limit or the fixed point family then satisfy 
(\ref{5.16}). As there is no natural choice for such a  
sequence, there are many possibilities which are guided mostly by the
principle of least technical complexity. In that sense, case ii. is really
case i. because we might have simply defined the sequence by passing to the 
subsequence right from the beginning. On the other hand, a fixed point maybe
unstable (i.e. it has non-convergent directions) so that iii. is 
qualitatively different from i. In practice, one is generically confronted 
with case iii. and what one finds is the following: After performing a sufficient
number $n$ of iterations, one observes that the flow stabilises in the sense  
that the family member $\mu^{(n)}_M$ can be parametrised 
in terms of vertex functions $v_{M,k}^{(n)}(m_1,..,m_k)$ 
on $k=0,1,2,...$ copies of the vertex set of $\gamma_M$ which are 
updated at each renormalisation step $n\mapsto n+1$ (in general $k$ can take 
all $M^3$ values). Such functions can be written as linear combinations
of products of Kronecker-$\delta$ functions on two copies of the vertex set, 
hence, the flow  can be stated in terms of the 
coefficients $c$ of those. At the fixed point in general not all those 
coefficients take fixed numerical valuews, rather there are fewer functional 
relations between them than their number. Pick some of them 
$\lambda_M^1,\lambda_M^2,..$ as independent ones and solve those relations 
for the 
rest which then become fixed functions of the $\lambda_M^1,\lambda_M^2,..$.
Due to the fixed point condition, the $\lambda_M^j$ can be written as 
functions of the $\lambda_{M'}^j$ for all $M<M'$. Obviously the theory 
is predictive iff the number of the $\lambda_M^j$ as $M\to \infty$ is finite.
The continuum measure $\mu^\ast$ is then determined trough the 
cylindrically consistent family $\{\mu^\ast_M\}_{M\in {\cal M}}$ 
and depends on those 
$M-$independent parameters which we denote by $\lambda_1,\lambda_2,..$.

The vertex functions 
alluded to above can be consdered as contributions to some 
discretised action $S^{(n)}_M$ and their coefficients in front of the 
algebraically independent terms can be called coupling
constants $g$. It is customary to choose them dimensionfree by multiplying 
by suitable powers of the cutoff (here $\epsilon_M=R/M$). This way 
there are more coupling constants $g$ that flow than coefficients $c$ 
that flow in our scheme.
In particular, if the decay of the two point function is determined by a 
mass parameter $\lambda$, the dimensionful correlation length
is $\lambda^{-1}$ and the dimensionless correlation length is 
$\xi_M=(\lambda\epsilon_M)^{-1}$ which diverges in the continuum limit 
$M\to\infty$. Under a renormalisation step 
(e.g. $\xi^{(n+1)}_M=\xi^{(n)}_{M'}=M'/M\; \xi^{(n)}_M,\; M<M'$ 
in the real space schemes considered in the main text, mostly $M'=2M$ is used)
the dimensionfree correlation length increases. The fixed point in this 
usual jargon
maybe characterised by the number
of a. stable, b. unstable c. and other directions 
of the flow in the parameter space of coupling constants. This distinction
is done by linearisation 
of the flow equation $g^{(n+1)}=F(g^{(n)})$ around the fixed point solutions
$g^\ast=F(g^\ast)$, a subset of which is relations  
among the coefficients $c$ of our scheme and the others correspond 
to the additional flow parameters just mentioned. 
The negative, positive and zero eigenvalues 
(critical exponents) of the $\beta^\ast=\beta(g^\ast)$ matrix associated 
to the $\beta$ function $\beta(g)=[\nabla F](g)$ defines those 
stable=irrelevant, unstable=relevant and other=marginal directions   
as the corresponding eigenvectors. If the integral curves of the 
irrelevant directions through the fixed point are integrable (surface forming)
then they are tangent to that integral manifold called the critical surface 
trough the fixed point which is of finite co-dimension for a predictive 
theory because the stable couplings run into the fixed 
point no matter what their initial values were (at least in the attraction
neighbourhood of that fixed point). This is the phenomenon of universality 
(for the fixed point under consideration -- there maybe several in which case 
one speaks of universality classes).
Note however, that the fixed point in this language is in a space 
of couplings $g$ that is higher dimensional than in our scheme. The relation 
between them is as follows: Consider the renormalisation group flow 
starting from a point $g$ close to the critical surface then the trajectory
will miss the point $g^\ast$ but comes close and then diverges into the 
relevant directions. We now tune the starting point $g^{(0)}(n)$ in such a way 
that $g^{(n)}(n)$ is independent of $n$ which means that in particular 
the starting value of the relevant couplings is closer and closer 
to the fixed point value and the critical surface. 
In particular since $\xi^{(n)}_M=2^n \xi^{(0)}_M$ 
we pick $\xi^{(0)}_M(n)=\xi_M 2^{-n}$. This means that in the limit 
$n\to \infty$ we end up at a point $g$ which depends on the fine tuned values 
(such as $\lambda=\epsilon_M/\xi_M$) and these correspond precisely
to the values $\lambda_j$ in our scheme. Such fine tuned trajectories that 
depend on the fine tuning parameters $\lambda$ are called renormalised 
trajectories which formally run inside the critical surface into the fixed 
point and then end up off it. If one wishes to describe the theories defined 
by such renormalised trajectories at finite resolution $M$ in terms of 
actions $S^\ast_M$ rather than measures $\mu^\ast_M$,
then these actions are termed {\it perfect} \cite{26}.
 
\section{Trotter-Kato formula for bounded degenerate semi-groups}
\label{sd}

Let ${\cal H}$ be a Hilbert space, $P$ a projection operator in $\cal H$ 
and $\mathbb{R}^+\to {\cal B}({\cal H});\; s\mapsto T(s)$ be 
a strongly continuous contraction semi-group (i.e. $||T(s)||\le 1$). 
One would like to show that the degenerate Kato-Trotter product 
\be \label{d.1}
K(s):=\lim_{N\to \infty} \; [T(\frac{s}{N})\;P]^N 
\ee
strongly converges to a strongly continuous 
degenerate contraction semi-group, that is,
$K(0)=Q\le P$ is a subprojection, $Q{\cal H}$ is an invariant subspace of
$K(s)$ and $K(s)(1-Q)=0$. A general proof of this result 
generalising 
the ususal case that $P$ in (\ref{d.1}) is replaced by 
another contraction semigroup $S(s)$ appears not be available yet in the 
literature.

In what follows we consider therefore the simpler case that $T(s)=e^{-s A}$ 
has a bounded generator $A$.
\begin{Theorem} \label{thd.1} ~\\
Let $T(s),P$ be as above with $T(s)=e^{-s A},\; ||A||<\infty$ bounded. Then 
\be \label{d.2}
K(s)=P\; e^{-s\; PAP}
\ee
that is $Q=P$ and $K(s)$ is a strongly continuous semigroup with  
generator $PAP$ on $P{\cal H}$.
\end{Theorem}
\begin{proof}
Consider the identity (no domain questions arise)
\ba \label{d.3}
B_N(s) &:=& [T(\frac{s}{N})\;P]^N-K(s)=[T(\frac{s}{N})\;P]^N-K(\frac{s}{N})^N=
\sum_{k=0}^{N-1}\; [T(\frac{s}{N})\;P]^k\;
[T(\frac{s}{N})\;P-K(\frac{s}{N})]\; K(\frac{s}{N})^{N-1-k}
\nonumber\\
&=& 
[T(\frac{s}{N})\;P-K(\frac{s}{N})] \; K(\frac{s}{N})^{N-1}
+
\sum_{k=1}^{N-1}\; [T(\frac{s}{N})\;P]^k\;
[P T(\frac{s}{N})\;P-K(\frac{s}{N})]\; K(\frac{s}{N})^{N-1-k}
\ea
Using $||P||\le 1$ and the 
sub-multiplicativity of the operator norm we estimate
\be \label{d.4}
||B_N(s)||\le \{||T(\frac{s}{N})\;P-K(\frac{s}{N})||\; 
+(N-1)\;||P\;T(\frac{s}{N})\;P-K(\frac{s}{N})||\}\; e^{2 \;|s|\; ||A||}
\ee
Since $A, PAP$ are bounded, the Taylor expansions converge in norm and we can 
estimate 
\ba \label{d.5}
||T(\frac{s}{N})\;P-K(\frac{s}{N})||
&\le& \sum_{n=1}^\infty\frac{|s|^n}{N^n\; n!} ||[A^n-(PAP)^n]P||
\le \frac{2}{N}\;e^{|s|\;||A||}
\nonumber\\
||P\;T(\frac{s}{N})\;P-K(\frac{s}{N})||
&\le& \sum_{n=2}^\infty\frac{|s|^n}{N^n\; n!} ||P[A^n-(PAP)^n]P||
\le \frac{2}{N^2}\;e^{|s|\;||A||}
\ea
whence 
\be \label{d.6}
\lim_{N\to \infty} ||B_N(s)||=0
\ee
hence the convergence is even uniform in this case.
\end{proof}
Note that the sign of $s$ played no role in the proof, but it will in the 
unbounded case.

\section{Projections and Limits}
\label{se}

In the main text we need the following result:
\begin{Theorem} \label{the.1}
Let $\cal H$ be a separable Hilbert space with projection $P$. Let $\Omega$ 
be cyclic 
for a $C^\ast-$subalgebra $\mathfrak{B}\subset {\cal B}({\cal H})$. Suppose 
that a self-adjoint operator $H$ has spectrum bounded from below by zero and 
that zero is an isolated eigenvalue with eigenvector $\Omega$. Then 
there exist elements $w_{I,k}\in \mathfrak{B};\;I,k\in \mathbb{N}$ such that 
\be \label{e.1}
s-\lim_{n\to\infty}\;[\lim_{k\to \infty}\;[\lim_{\beta\to\infty}\;
P(n,k,\beta)]]]=P,\;\;
P(n,k,\beta)=\sum_{I=1}^n\; w_{I,k}\; e^{-\beta H}\; w_{I,k}^\dagger
\ee
\end{Theorem}
\begin{proof}:
Let $b_I,\; I\in \mathbb{N}$ be an ONB of $P{\cal H}$. Then by cyclicity we 
find for each $I$ a sequence $w_{I,k}\in \mathfrak{B}$ such that 
$s-\lim_{k\to \infty} w_{I,k}\Omega=b_I$. Next, let $E_H$ be the 
projection valued measure of $H$, then by the spectral theorem and 
$H\Omega=0$
\be \label{e.2}
||e^{-\beta H}\psi-\Omega\; <\Omega, \psi>||^2
=<\psi,e^{-2\beta H}\; \psi>-|<\Omega,\psi>|^2
=\int_{\sigma(H)} \; d<\psi,E_H(\lambda)\psi>\; e^{-2 \lambda \beta}
-|<\Omega,\psi>|^2
\ee
where $\sigma(H)\subset \mathbb{R}^+_0$ denotes the spectrum of $H$. By
assumption $E_H(\lambda)=\theta(\lambda)\; \Omega\; <\Omega,.>+E'_H(\lambda)$
where $E'_H(\lambda)=0$ for $\lambda\in [0,\epsilon)$ and some $\epsilon>0$.
It follows 
\be \label{e.3}
||e^{-\beta H}\psi-\Omega\; <\Omega, \psi>||^2
\le \int_{\sigma(H)\cap [\epsilon,\infty)} \; 
d<\psi,E_H(\lambda)\psi>\; e^{-2 \lambda \beta}
\le e^{-2\beta\epsilon} ||\psi||^2
\ee
thus $s-\lim_{\beta\to \infty} e^{-\beta H} \psi=\Omega \; <\Omega, \psi>$.
Finally, $P{\cal H}$ is the closure of the finite linear span of the $b_I$, 
thus $s-\lim_{n\to \infty} P_n\psi=P\psi$ where 
$P_n\psi=\sum_{I=1}^n \; b_I \; <b_I,\psi>$ 

It remains to write $P(n,k,\beta)-P$ as a telescopic sum
\be \label{e.4}
(P-P(n,k,\beta))\psi
=(P-P_n)\psi+\sum_{I=1}^n\; [b_I\; <b_I,\psi>-w_{I,k}\Omega\;
<w_{I,k}\;\Omega,\psi>]+
\sum_{I=1}^n\;w_{I,k}[\Omega <\Omega,w_{I,k}^\dagger \psi>
-e^{-\beta H} w_{I,k}^\dagger \psi]
\ee
and to take the strong limits in the indicated order, i.e. first 
$\beta\to\infty$ at finite $n,k$ removes the third term, then 
$k\to\infty$ at finite $n$ removes the second term and finally $n\to\infty$
removes the first term.
\end{proof} 

\end{appendix}

}

\end{document}